\documentclass[11pt]{exam}
\usepackage{graphicx}
\usepackage{subcaption}
\usepackage{float}
\usepackage{booktabs}
\usepackage[usenames,dvipsnames]{color}
\usepackage{setspace}

\usepackage[colorinlistoftodos,textsize=tiny%
			]{todonotes}

\usepackage[letterpaper, top=1.5in, bottom=1.5in, left=1in, right=1in]{geometry}

\usepackage{amsfonts}
\usepackage{amssymb}
\usepackage{amsmath, amsthm}
\usepackage{latexsym}
\usepackage{xcolor}
\usepackage{xr-hyper}
\usepackage{hyperref}
\usepackage{bm, mathtools}
\newtheorem{theorem}{Theorem}[section]
\newtheorem{lemma}[theorem]{Lemma}
\newtheorem{remark}[theorem]{Remark}

\newtheorem{proposition}[theorem]{Proposition}
\newtheorem{definition}[theorem]{Definition}

\usepackage{natbib}

\newcommand{\E}{\mathbb{E}}

\newcommand{\R}{\mathbb{R}}
\newcommand{\F}{\mathcal{F}}

\newcommand{\llbracket}{[\![}
\newcommand{\rrbracket}{]\!]}

\newcommand{\G}{\mathcal{G}}
\newcommand{\HH}{\mathcal{H}}
\newcommand{\KK}{\mathcal{K}}
\newcommand{\BB}{\mathcal{B}}
\newcommand{\cU}{\mathcal{U}}
\newcommand{\filF}{\mathbb{F}}
\newcommand{\filG}{\mathbb{G}}
\newcommand{\filB}{\mathbb{B}}
\newcommand{\Wien}{\mathbb{W}}
\newcommand{\FG}{\mathbb{F}^{G}}
\newcommand{\nuv}{\bm{\nu}}
\newcommand{\zetav}{\bm{\zeta}}
\newcommand{\rhov}{\bm{\rho}}

\newcommand{\dbl}{[\![}
\newcommand{\dbr}{]\!]}

\newcommand{\RR}{\mathbb{R}}

\newcommand{\PP}{\mathbb{P}}
\newcommand{\1}{\mathbf{1}}
\newcommand{\cA}{\mathcal{A}}

\newcommand{\thetavec}{\bm{\theta}}
\newcommand{\sigmavec}{\bm{\sigma}}
\newcommand{\muv}{\bm{\mu}}
\newcommand{\piv}{\bm{\pi}}
\newcommand{\onev}{\bm{1}}
\newcommand{\Sigmav}{\bm{\Sigma}}

\title{Portfolio Choice with Competing Precautionary and Accumulation Goals}
\author{Steven Campbell\thanks{Department of Statistics, Columbia University, Email: \href{mailto:sc5314@columbia.edu}{sc5314@columbia.edu}.} \\
  Columbia University  \and Agostino Capponi\thanks{Department of Industrial Engineering and Operations Research and Columbia Business School, Columbia University, Email: \href{mailto:ac3827@columbia.edu}{ac3827@columbia.edu}.}\\
  Columbia University \and Ananya Parashar\thanks{Department of Industrial Engineering and Operations Research, Columbia University, Email: \href{mailto:ap4658@columbia.edu}{ap4658@columbia.edu}.} \\
  Columbia University
  }
\date{}
\begin{document}
\maketitle

% \begin{quote}
% \textit{``There is so much competition over getting alpha, but
% everyone has access to the same hedge funds.
% Give up! That's not what is important.''}\\[4pt]
% \textit{The focus should be on goals-based investing.}
% \begin{flushright}
% --- Robert C. Merton, Nobel Laureate in Economics\\
% \small\textit{Top1000funds.com}, December 2013
% \end{flushright}
% \end{quote}
\bigskip

\vspace{-0.75cm}
\begin{abstract}
We study optimal portfolio choice for a household managing two goals at once. A random-deadline precautionary goal, such as a medical emergency, must be paid in full whenever it arrives and is affordable; a fixed-deadline accumulation goal, such as a target retirement lifestyle, may be declined at its deadline. We show that precautionary saving crowds out growth-oriented investment, %, a coupling absent from either goal alone, 
and deadline pressure is amplified by the competing precautionary claim. We characterize the value function as the unique viscosity solution of an HJB equation, and derive the household's optimal terminal funding rule. The household funds the fixed-deadline goal exactly when the direct benefit at least offsets the resulting drop in the value of keeping the precautionary goal funded. 
We calibrate the model and find that both flexibility and fungibility are most valuable at intermediate wealth. Committing in advance to fund a goal sacrifices flexibility, while managing each goal in its own 
account sacrifices fungibility.
%In a stylized two-goal example, a household with separate accounts needs roughly \$6,000 more present-value wealth—about 18\%—to reach the same welfare as a pooled household holding \$33{,}000.
\end{abstract}

\vspace{-0.1cm}

\bigskip

\section{Introduction}
\label{sec:intro}

Households simultaneously manage two structurally different types of financial goals: pre-scheduled \emph{accumulation} liabilities with known due dates, such as college tuition, a mortgage down payment, or retirement; and unforeseeable \emph{precautionary} shocks with uncertain timing, such as a medical emergency, involuntary job loss, or an accident. Despite the empirical prevalence of both goal types, the investment literature has studied them almost entirely in isolation. 
For instance, \citet{Browne1999} studies the fixed-deadline problem, while 
the random-deadline problem is considered by \citet{BY16}. No prior framework addresses both simultaneously, which is precisely the object of our investigation.

We consider a household that invests dynamically across multiple risky assets and seeks to maximize a weighted combination of the probabilities of fully funding each goal when it comes due. The all-or-nothing funding rule is central to our formulation: when a goal arrives, it is paid in full if affordable and chosen to be funded, and partial funding is not allowed. The 
random-deadline precautionary goal is mandatory whenever it arrives, whereas the 
fixed-deadline accumulation goal is subject to a terminal funding decision. At the 
deadline $T$, after observing the realized goal amount $G$, the household chooses whether to fund it.

The interplay between a precautionary goal and an accumulation goal generates two economic effects. The first is a \emph{growth crowding-out effect}. Maintaining a precautionary buffer for the random-deadline goal can be costly: it diverts resources away from the growth-oriented portfolio needed to accumulate toward the fixed deadline. The mechanism is the discontinuity of the post-payment map $w\mapsto w-R\mathbf 1_{\{w\ge R\}}$, where $R$ denotes the realized amount of the random-deadline goal: crossing the affordability threshold moves the household discontinuously from $w$ to $w-R$, so the wealth accumulated toward the fixed-deadline goal is precisely what a mandatory random-goal payment consumes. The optimal portfolio therefore reflects a continuous tension between precautionary saving and growth-oriented investing that is absent when goals are managed independently. The second is a \emph{deadline pressure effect}. A household with a shorter horizon has less time for the risky asset to compound toward the fixed-deadline goal, leading to more risk-taking over that horizon.

The joint-goal problem has a natural recursive structure that guides our solution method. While both goals are pending, the household optimally manages its portfolio knowing that, once the first goal is resolved, it will face a simpler single-goal continuation problem. This triangular structure allows us to derive an exact dynamic programming representation of the joint-goal value function, in which
the single-goal value functions appear as endogenous continuation values. We characterize the joint-goal value function as the unique viscosity solution of a Hamilton--Jacobi--Bellman equation and establish existence and uniqueness.

At time $T$, if the precautionary goal is still pending, the household funds the accumulation goal if and only if the direct reward $\alpha_D$ at least offsets the resulting loss in the precautionary goal's continuation value from paying $G$ toward the accumulation goal. That is,
\begin{equation}
\label{eq:intro-threshold}
\alpha_D \ge \alpha_R\bigl[V^R(w)-V^R(w-G)\bigr]
\qquad\text{and}\qquad w\ge G.
\end{equation}
This decision rule is directly actionable: it tells the household, or the financial advisor managing its portfolio, exactly when paying the accumulation goal is optimal given current wealth, the size of the liability, the arrival intensity of the precautionary goal, and the relative priority of the two goals. High emergency risk or low accumulation-goal priority can make declining an affordable accumulation goal optimal. When $\alpha_D\ge\alpha_R$, however, we show funding every affordable accumulation goal is weakly optimal.

To build intuition before the general analysis, we study a simplified three-period binomial model that isolates the central economic mechanism: how much risk to take ex ante, knowing that funding one goal may deplete the resources available for the other. The benchmark gives intuition about how portfolio choice depends on the relative size, timing, and importance of the two goals, and reveals the crowding-out mechanism in its simplest closed-form structure. 

To measure the value of this discretion, we also study a {committed household} that has forfeited it and agreed in advance to fund the fixed-deadline accumulation goal whenever affordable. These situations would arise, for instance, when a college offers a prepaid tuition contract at a locked-in rate. They also arise when a signed home purchase agreement with a mortgage rate lock obligates the buyer to fund the down payment at closing, or else forfeit the earnest money deposit. We use the committed household purely as a benchmark against which to measure the \emph{value of flexibility} that most households in fact retain, characterizing the cost of this commitment through two complementary measures. The \emph{ex ante option value} $\Delta V_0(w_0) := V^{R,D}(0,w_0) - V_{\mathrm{commit}}^{R,D}(0,w_0)$ is the reduction in the optimized weighted goal-attainment value from committing at time $0$. The \emph{terminal option value} $\Delta\mathcal{T}(w)$ isolates the cost at the funding decision node. In our calibrations, both measures are hump-shaped, and are largest at intermediate wealth where paying the fixed-deadline goal would substantially deplete the emergency buffer. 

We also consider a split-bucket benchmark, which isolates a distinct cost that households bear in practice. Instead of managing one pooled portfolio, a household may divide its wealth between two dedicated accounts, one for each goal, and run each under its own single-goal optimal policy. This {split-bucket} benchmark $B(t,w)$ is the arrangement that behavioral portfolio theory describes (see \citet{Das2010}), in which goals occupy non-fungible mental accounts. The difference $V^{R,D}(0,w_0)-B(0,w_0)$ measures the \emph{value of fungibility}: what a dollar is worth when it remains available to whichever goal turns out to need it. This value need not be nonnegative because separate accounts can shield fixed-goal wealth from the mandatory random-goal payment. However, in our numerical experiments, this value is always nonnegative. Like the value of flexibility, it is hump-shaped in wealth and vanishes at both ends: at low wealth neither arrangement has much to allocate, and at high wealth both accounts can cover their respective goals. It is largest when the two goals carry comparable priority, since that is when committing a dollar to the wrong bucket is most costly. Expressed in dollars, a household using separate accounts needs roughly \$6{,}000 of additional initial wealth to match a pooled household holding \$33{,}000, in a stylized example with comparable goal scales. 

We calibrate the market, arrival, and goal-scale parameters using empirical estimates, drawing on job-separation rates for the precautionary-goal arrival intensity and College Board and J.P. Morgan estimates for the college and retirement goal amounts; see \citet{jpmorgan2025replacement,collegeboard2025pricing}. We treat the remaining parameters as modeling choices. We report the results using the formulation in which the household may decline the accumulation goal; see Section~\ref{sec:effects-numerics}.

Our numerical analysis quantifies the economic significance of the two effects. The deadline-pressure effect is large: a household planning for high-income retirement that starts at age 25 with $\$100{,}000$ of initial wealth, expressed in date $0$ dollars, attains an optimized weighted goal-attainment value of approximately $0.85$, whereas a household starting at age 50 requires roughly $\$1.07$ million to attain the same value. Here date $0$ is the corresponding household's saving-start date. This is a more than tenfold wealth-equivalent difference between the two saving-start dates. The growth crowding-out effect is also quantitatively significant. As the shock-arrival intensity $\lambda$ rises, a household at a given low level of initial wealth attains a lower optimized value; equivalently, it needs more initial wealth to attain the same value. The optimal risky position also remains elevated over a wider wealth range. Under the baseline equal weights, the accumulation goal is funded whenever it is affordable, so in both cases the optional and committed value problems coincide. 
%The preference-weight sensitivity further shows how the balance between the two goals shapes the value and policy curves.

Retaining discretion also has a quantitatively important effect on portfolio choice. Consider a household that begins with approximately $\$36{,}000$ and saves toward a median $\$124{,}000$ public-college goal due in 18 years, while assigning $10\%$ of its objective weight to college and $90\%$ to the precautionary goal. Conditional on the emergency goal remaining unresolved, the household's expected risky-asset share one month before the college deadline is approximately 30 percentage points higher when it retains the option to decline the college goal than when it has committed in advance to fund it whenever affordable.

Our framework has direct implications for robo-advisory platforms, several of which center their strategy on goal-based investing; see \citet{GrealishKolm2023} for the design of such platforms. For example, Betterment, one of the largest independent robo-advisors, explicitly organizes client portfolios around distinct financial goals including retirement, emergency funds, education, and major purchases.\footnote{See Betterment's Goal Projection and Advice Disclosure at
\url{https://www.betterment.com/legal/goal-projection}.} Charles Schwab likewise recommends assigning a separate sub-portfolio to each goal \citep{Schwab2025}. However, in both cases each goal is managed in a separate account with its own allocation and projection model, treating emergency preparedness and long-horizon saving as independent problems rather than as competing claims on the same portfolio. Our joint-goal framework provides the theoretical foundation for a unified advisory rule that explicitly accounts for the interaction between precautionary and accumulation goals, delivers the optimal dynamic portfolio rule while both goals are pending, and characterizes the optimal funding decision at the fixed deadline.

\subsection{Literature Review}
Traditional finance, rooted in \citet{Markowitz1952}, defines risk as the standard deviation of returns. The behavioral portfolio theory introduced by \citet{Shefrin2000} posits that investors naturally organize their wealth into mental accounts tied to specific goals. In this framework, risk is redefined as the probability of failing to reach a specific wealth threshold \citep{Das2010, Brunel2015}. \citet{Chhabra2005} further operationalizes this by suggesting a wealth allocation framework that buckets assets into personal, market, and aspirational risks to protect basic standards of living while pursuing growth. Goal setting also has measurable effects on household behavior. Using data from a FinTech saving app, \citet{GarganoRossi2024} show that setting specific saving goals causally and substantially increases users' savings.

Recent studies have developed frameworks for goal-based investing. \citet{Deguest2021} provide a comprehensive treatment of goal-based investing, including the operational and implementation challenges financial advisors face. \citet{Das2022} develop a dynamic programming algorithm using backward induction to prioritize competing annual goals, such as tuition and retirement, showing that this approach significantly outperforms target-date funds. \citet{Capponi2024} provide a sequential goal-based wealth management framework, distinguishing
between flexible goals, i.e., those yielding partial utility, and all-or-nothing goals. More recently, \citet{Bayraktar2025} expanded this by incorporating ``mental costs'', namely psychological penalties for moving funds between accounts, which reflect the behavioral reality that households treat goal-linked savings as non-fungible.

A challenge in household finance is balancing deterministic goals such as a planned retirement with random-arrival shocks, such as emergency medical expenses. While \citet{Browne1999} provides a benchmark for reaching a goal by a fixed deadline, households must maintain precautionary buffers for unpredictable events. This precautionary side of the problem is the subject of a parallel literature on minimizing the probability of lifetime ruin. \citet{Young2004} characterizes the optimal investment strategy that minimizes the probability of ruin before a random time of death. \citet{BY16} extend this random-deadline framework to the probability of reaching a bequest goal by an exponentially distributed time, the building block we adopt for our random-deadline goal in Section~\ref{sec:random}. \citet{Young2016} study the analogous problem under an ambiguous hazard rate. These papers characterize the precautionary buffer in isolation; they do not consider a competing fixed-horizon goal that draws on the same wealth. Our contribution to this literature is twofold. We identify a growth crowding-out effect and also use a committed-funding benchmark to characterize the value of the flexibility, which most households in fact retain, to decline the fixed-horizon goal.

\subsection{Notation}
\label{sec:notation-summary}

We collect the notation and definitions used throughout the paper in Tables~\ref{tab:notation-primitives} and~\ref{tab:notation-values}.

\begin{table}[H]
\centering
\caption{Notation: market, goals, and wealth processes.}
\label{tab:notation-primitives}
\begin{tabular}{@{}l p{0.55\textwidth}@{}}
\toprule
Symbol & Meaning \\
\midrule
$r,\ \muv,\ \sigmavec,\ \Sigmav=\sigmavec\sigmavec^\top$ & risk-free rate, risky drift, volatility, covariance \\
$\thetavec=\muv-r\onev,\ \gamma^2=\thetavec^\top\Sigmav^{-1}\thetavec$ & excess-return vector, squared Sharpe ratio \\
$\Pi_K,\ K$ & admissible portfolio set and its sup-norm bound \\
$\piv_t$ & portfolio weight process (control) \\
$T$ & fixed deadline \\
$G,\ F_G,\ b_G$ & fixed-deadline goal amount, its distribution, its support bound \\
$\tau,\ \lambda$ & random-deadline arrival time, $\tau\sim\mathrm{Exp}(\lambda)$ \\
$R,\ F_R,\ b_R$ & random-deadline goal amount, its distribution, its support bound \\
$b=b_R+b_G$ & upper bound on the combined goal amount \\
$\alpha_R,\ \alpha_D$ & preference weights on the two goals, $\alpha_R+\alpha_D=1$ \\
$W_t$ & pre-payment wealth process, in time-$T$ dollars \\
$\overline W_t$ & corresponding wealth in current dollars \\
$w_0,\ w_t,\ w$ & initial wealth, wealth at date $t$, and a generic wealth level in the analytical model, in time-$T$ dollars \\
$w_0^{\mathrm{PV}},\ w^{\mathrm{PV}}$ & initial wealth and a generic displayed wealth level in date-$0$ present-value dollars in the numerical figures \\
$x=\log w$ & log-wealth coordinate under section's unit convention \\
$\Gamma_A(w)$ & post-payment map for a mandatory goal of size $A$ \\
$\Theta(w,g),\ d_T$ & feasible funding decisions and terminal funding choice \\

\bottomrule
\end{tabular}
\end{table}

\begin{table}[H]
\centering
\caption{Notation: value functions and option-value objects.}
\label{tab:notation-values}
\begin{tabular}{@{}l p{0.55\textwidth}@{}}
\toprule
Symbol & Meaning \\
\midrule
$H(w,p,q)$ & Hamiltonian of the single-goal HJB equations \\
$V^D(t,w)$ & fixed-deadline single-goal value function \\
$V^R(w)$ & random-deadline single-goal value function \\
$\mathcal J(s,w)$ & joint-goal continuation operator following an early random-goal arrival \\
$D^*(w,g)$ & optimal terminal funding selector \\
$\mathcal T(w)$ & terminal operator at $T$ when the household retains discretion \\
$\mathcal T_{\mathrm{commit}}(w)$ & terminal operator at $T$ under committed funding \\
$\Delta V_0(w_0)$ & ex ante option value of retaining discretion \\
$\Delta\mathcal T(w)$ & terminal option value of retaining discretion  \\
$B(t,w)$ & split-bucket (non-fungible account) benchmark value \\
\bottomrule
\end{tabular}
\end{table}

Because the admissible portfolio set $\Pi_K$ is fixed throughout, all value functions, Hamiltonians, and continuation and terminal operators are understood relative to $\Pi_K$, and we suppress this dependence from their notation. We restore the cap index on the value functions only when comparing different values of $K$ in Section~\ref{sec:asymptotics}.

\section{An Illustrative Three-Period Model}
\label{sec:DiscreteBenchmark}

The full two-goal problem is challenging for two reasons. First, the household faces two goals with distinct deadlines, one deterministic and one random. Second, funding a goal requires a fixed-dollar withdrawal from a multiplicative wealth process which destroys the scale invariance of geometric Brownian motion and rules out closed-form solutions. To build intuition before the general analysis, we study a simplified three-period binomial model. It isolates the central tradeoff around how much risk to take ex ante with the knowledge that funding one goal may deplete the resources available for the other.

Our simplified benchmark has four dates, $t=0,1,2,3$. A random goal of size $R>0$ arrives at $\tau_R\in\{1,3\}$, with $\PP(\tau_R=1)=p$, and a deterministic goal of size $G>0$ arrives at date $2$. As would be the case with an emergency, the random goal is mandatory whenever it is affordable. On the other hand, the deterministic goal is optional. For consistency, all wealth levels and liabilities are expressed in time-$0$ units. Thus a goal of size $A$ due at date $t$ requires a nominal payment of $b^tA$, where $b:=1+r>1$ is the risk-free return for a single period.

Assume that the risky asset has independent per-period gross returns of either $b+\Delta$ or $b-\Delta$ where $0<\Delta<b$. Each case occurs with equal probability $1/2$ independent of $\tau_R$. Set $\delta:=\Delta/b$. If a household allocates a fraction $\pi$ of their wealth to the risky asset, their discounted wealth is multiplied by $1+\delta\pi$ or $1-\delta\pi$ depending on the realized return. Both shorting and leverage are allowed, subject to the constraint that wealth remains non-negative under both return states. Equivalently, $\pi$ is restricted to take values in $\Pi=[-\delta^{-1},\delta^{-1}]$.

In this model portfolio revision is \emph{event-driven}. The household chooses $\alpha\in\Pi$ at date $0$. If the random goal arrives at date $1$, the payment is processed and the household may choose a new weight for $[1,2]$. If it does not arrive, the portfolio is rebalanced to the previously chosen fraction $\alpha$, without a new investment decision. After a funding decision is made at time $2$ for the deterministic goal, the household may revise once more for $[2,3]$ if the random goal remains outstanding. Portfolio updates are therefore triggered by liability events. This simplifies economic reality but suffices to capture the key tension that the initial portfolio must balance preparedness for an early random liability against the ability to reach the deterministic goal.

The household's objective will be the weighted success criterion \[
q\,\PP(\text{random goal is met}) +
(1-q)\,\PP(\text{deterministic goal is met}).
\]
where $q\in[0,1]$ is the preference weight given to the random goal. For a mandatory goal of size $A$,
we write $\Gamma_A(x):=x-A\mathbf 1_{\{x\ge A\}}$ for post-payment wealth when starting with wealth $x$.

The continuation problem after a goal arrives reduces to a one-period target problem. For a target goal amount $A>0$, the probability of successfully meeting the goal after optimizing portfolio holdings is
\begin{equation}
\Phi_A(x):=\sup_{\pi\in\Pi}
\PP\bigl(x(1+\delta\pi\varepsilon)\ge A\bigr)=
\begin{cases}
1, & x\ge A,\\
\frac12, & A/2\le x<A,\\
0, & 0\le x<A/2,
\end{cases}
\label{eq:one-step-value-symmetric}
\end{equation}
where $\varepsilon$ takes the values $-1$ and $1$ with equal probability.
Wealth above $A$ defines a ``lock-in region.'' At these levels the household can secure the goal by holding the risk-free asset. Between $A/2$ and $A$, the household can only reach the goal in one market state. Below $A/2$, the goal cannot be reached in a single period regardless of the portfolio chosen. The derivation and optimal one-period portfolios are recorded in Appendix~\ref{app:three-period-details}.

Let $E(x)$ be the continuation value when the random goal arrives early at date
$1$ and pre-payment wealth is $x$. Let $L(z)$ be the continuation value at date $2$ when the random goal arrives late (at time $3$) and pre-decision wealth is $z$. If $d=1$ denotes the decision to fund
the deterministic goal, its feasible set is $\{0\}$ for $z<G$ and $\{0,1\}$ for $z\ge G$. Backward induction gives
\begin{align}
E(x)&:=q\mathbf 1_{\{x\ge R\}}+(1-q)\Phi_G\bigl(\Gamma_R(x)\bigr),
\label{eq:E-def-symmetric}\\
L(z)&:=\max_{d\in\Theta_G(z)}
\bigl\{(1-q)d+q\Phi_R(z-dG)\bigr\},
\label{eq:L-def-symmetric}
\end{align}
where $\Theta_G(z)$ denotes the feasible set for $d$. On the early branch, the first term in $E$ records success of the random goal and the second measures the chance of subsequently reaching the deterministic goal. On the late branch, the household compares funding the deterministic goal with preserving wealth for the outstanding random goal.

For $z\ge G$, with ties resolved in favor of funding, the deterministic
goal is funded exactly when
\begin{equation}
1-q\ge q\bigl[\Phi_R(z)-\Phi_R(z-G)\bigr].
\label{eq:discrete-funding-comparison}
\end{equation}
The left-hand side is the direct reward from funding the deterministic goal. The right-hand side measures what the payment costs in random-goal continuation value. This illustrates the crowding-out mechanism in its simplest form. If the cost on the right hand side of \eqref{eq:discrete-funding-comparison} is too large, the household declines the goal and preserves its emergency buffer. Allowing the household to forgo the goal ensures that attaining sufficient wealth to fund the goal cannot lower value. If it comes at too high a cost, the household can simply decline it and continue with its previous policy. Later, we will also study the committed benchmark which removes this discretion.

We now return to date $0$. Let $\varepsilon_1$ and $\varepsilon_2$ be independent $\{-1,1\}$-valued random variables that determine the risky asset returns; i.e., the risky asset has gross return $b+\varepsilon_i\Delta$ in period $i\in\{1,2\}$. For initial wealth $w$, define $X_1^\alpha:=w(1+\delta\alpha\varepsilon_1)$ to be the discounted wealth at date 1 and $X_2^\alpha:=w\prod_{i=1}^2(1+\delta\alpha\varepsilon_i)$ to be the discounted wealth at date 2 assuming no random goal arrival at date $1$. Note that the latter value uses the same portfolio weight $\alpha$ in both periods, as per our assumption that position updating is event-driven. By conditioning on the random arrival time, we can write the following Bellman equation for the date-0 value:
\begin{equation}
V(w)=\sup_{\alpha\in\Pi}
\left\{p\,\mathbb E\bigl[E(X_1^\alpha)\bigr]
+(1-p)\,\mathbb E\bigl[L(X_2^\alpha)\bigr]\right\},
\label{eq:V0-general}
\end{equation}
where the expectations are over the return signs. The parameters $p$ and $q$ play different roles. The arrival probability $p$ modulates the temporal urgency of the random goal, while the preference weight $q$ enters the date-$2$ funding decision through \eqref{eq:discrete-funding-comparison} and governs goal priority.

Because $E$ and $L$ are piecewise constant, the objective in \eqref{eq:V0-general} depends on $\alpha$ only through the values of $E$ and $L$ realized in each return state on the early and late branches. The objective is therefore a step function of $\alpha$, and the date-$0$ problem reduces to a combinatorial choice among the attainable configurations of continuation values. A conservative portfolio protects downside wealth and makes an early random payment less likely to exhaust the resources needed for the remaining goal. A more aggressive portfolio may meet higher funding thresholds when assets move favorably, but also weakens this buffer. When the two consecutive period returns have opposite signs, wealth is $w(1-\delta^2\alpha^2)$, which also falls as $|\alpha|$ rises. More risk can therefore weaken this intermediate branch as well. Appendix~\ref{app:three-period-details} provides the explicit continuation formulas, the solution structure of the initial portfolio problem, and a discussion of the illustrative special case when $R=G$. 

Interestingly, the continuation values themselves also contain a source of ambiguity. On the early branch, crossing the threshold $x=R$ simultaneously secures the random goal and forces its payment. The first effect raises value while the second lowers the resources available for the deterministic goal. Thus the change in $E$ can have either sign depending on which effect dominates. On the other hand, on the late branch, the ability to decline $G$ provides another buffer.

\section{Model Setup}\label{sec:gen.model}

We now turn to the general model in continuous time. Let $(\Omega,\mathcal F,\PP)$ support a $d$-dimensional Brownian motion $B=(B_t)_{t\ge0}$, and let $\mathbb F=(\mathcal F_t)_{t\ge0}$ be the usual augmentation of its natural filtration.

We consider a financial market with $n$ risky assets and one money-market account with constant interest rate $r\ge0$. The risky asset prices satisfy
\begin{equation}\label{eq:S}
dS_t=\operatorname{diag}(S_t)\bigl(\muv\,dt+\sigmavec\,dB_t\bigr),
\end{equation}
where $\muv\in\RR^n$ and $\sigmavec\in\RR^{n\times d}$ are constant,
$n\le d$, $\operatorname{rank}(\sigmavec)=n$, and
$\Sigmav:=\sigmavec\sigmavec^\top$ is positive definite. We write
$\thetavec:=\muv-r\onev$ for the excess return.

The investor faces two goals. A \emph{fixed-deadline goal} with random amount $G\ge0$ is due at the deterministic time $T>0$. After $G$ is revealed at time $T$, the household may choose whether to fund this goal. A \emph{random-deadline goal} with random amount $R\ge0$ arrives at time $\tau\sim\mathrm{Exp}(\lambda)$, where $\lambda>0$. This goal is mandatory and paid in full whenever it is affordable. If it is unaffordable, the goal is missed.

The two goals have two different economic interpretations. The random-deadline goal models a shock, such as a medical emergency or involuntary job loss, whose cost must be paid in full for a resolution. On the other hand, the fixed-deadline goal models a discretionary purchase, such as a college education or retirement lifestyle, that the household may rationally forgo if paying it would leave too little wealth for a future emergency. Our main model gives the household this discretion, but Section~\ref{sec:forced-funding} also considers the case in which the household commits to funding the deterministic goal subject to affordability.

We assume that $R$, $G$, and $\tau$ are mutually independent and independent of $\mathcal F_\infty$. The amounts $R$ and $G$ have continuous distribution functions $F_R$ and $F_G$ with compact supports contained in $[0,b_R]$ and $[0,b_G]$, respectively. Then $b:=b_R+b_G$ is an upper bound on the amount needed to fund both goals. In what follows, all of the wealth variables and goal amounts are measured in time-$T$ dollars. We write $W_s$ for the wealth process in these units \emph{before} any payment is submitted and $w$ for a generic realized wealth level. Thus an amount $y$ in time-$T$ dollars corresponds to $e^{-r(T-s)}y$ current dollars at time $s$. Equivalently, the current-dollar counterpart of $W_s$ is $\overline W_s:=e^{-r(T-s)}W_s$.

The investor observes market information continuously. The pair $(\tau,R)$ is revealed at $\tau$, whereas $G$ is revealed at $T$. Let $\mathbb G$ denote the corresponding augmented, right-continuous information filtration. Independence implies that $B$ remains a $\mathbb G$-Brownian motion. The formal filtration and conditional-space conventions are given in Appendix~\ref{app:reduction}.

Fix a nonempty compact portfolio set $\Pi_K\subset\mathbb R^n$ containing the origin with $K:=\max_{\piv\in\Pi_K}\|\piv\|_\infty$. Since no symmetry, convexity, or product structure is imposed, the set $\Pi_K$ can be used to reflect long-only, leverage, or other asset-specific constraints. A trading strategy $\piv=(\piv_t)_{t\ge0}$ is $\mathbb G$-predictable and takes values in $\Pi_K$, where $\pi_t^i$ is the fraction of wealth invested in risky asset $i$ and $1-\onev^\top\piv_t$ is invested in the money-market account. We write $\mathcal A^K_{t,u}$ for this class on $[t,u]$ and let $\mathcal A^{K,{\rm pre}}_{t,u}$ denote the analogous class of $\mathbb F$-predictable strategies when no goal information has been revealed. Boundedness of $\Pi_K$ ensures the local square integrability of controls, and $0\in\Pi_K$ makes the zero strategy (investment in the risk free security) admissible.

For $0\le t\le s$ and initial wealth $w\ge0$, we write $\widetilde W_s^{t,w;\piv}$ to denote the pre-payment wealth process parameterized by the initial data $(t,w)$ and invested according to the (admissible) portfolio $\piv$. In time-$T$ dollars it satisfies
\begin{equation}\label{eq:pre.pay.wealth}
d\widetilde W_s^{t,w;\piv}
=\widetilde W_s^{t,w;\piv}
\bigl(\piv_s^\top\thetavec\,ds+\piv_s^\top\sigmavec\,dB_s\bigr),
\qquad \widetilde W_t^{t,w;\piv}=w.
\end{equation}
In particular, $W_s=\widetilde W_s^{0,w_0;\piv}$, and a wealth level $w$ at time $t$ corresponds to current-dollar wealth $e^{-r(T-t)}w$.

For $a,w\ge0$, write $\Gamma_a(w):=w-a\mathbf 1_{\{w\ge a\}}$ for the wealth remaining after a mandatory all-or-nothing payment of size $a$. For $w_0\ge0$ and $\piv\in\mathcal A^K_{0,\infty}$, let $Y^{w_0;\piv}$ denote wealth on $[0,T)$ obtained by following \eqref{eq:pre.pay.wealth} from $w_0$ and, if $\tau<T$, replacing its value at $\tau$ by $\Gamma_R(Y_{\tau-}^{w_0;\piv})$. Thus $Y_{T-}^{w_0;\piv}$ is wealth immediately before the fixed-deadline decision and is determined by $\piv$ alone.

For wealth $w$ and a realized fixed-goal amount $g$, the feasible decisions are $\Theta(w,g):=\{d\in\{0,1\}:dg\le w\}$, where $d=1$ means funding the goal. A full admissible control is a pair $(\piv,d_T)$ with $\piv\in\mathcal A^K_{0,\infty}$ and a $\{0,1\}$-valued, $\mathcal G_T$-measurable decision satisfying
\begin{equation}\label{eq:terminal-feasibility}
d_T\in\Theta\bigl(Y_{T-}^{w_0;\piv},G\bigr)
\quad\text{a.s.},
\qquad\text{equivalently}\qquad
d_TG\le Y_{T-}^{w_0;\piv}\quad\text{a.s.}
\end{equation}
We denote the set of these decisions by $\mathcal D_T(w_0,\piv)$.

The post-payment wealth $\widehat W^{w_0;\piv,d_T}$ follows \eqref{eq:pre.pay.wealth} between goal times and obeys
\begin{align}
\Delta\widehat W_\tau^{w_0;\piv,d_T}
&=-R\mathbf 1_{\{\widehat W_{\tau-}^{w_0;\piv,d_T}\ge R\}},
\label{eq:jump-R}\\
\Delta\widehat W_T^{w_0;\piv,d_T}&=-d_TG.
\label{eq:jump-G}
\end{align}
In particular, $\widehat W_{T-}^{w_0;\piv,d_T}=Y_{T-}^{w_0;\piv}$.
Because $\PP(\tau=T)=0$, no simultaneous-payment convention is needed. The formal construction, well-posedness, and nonnegativity of this process are
deferred to Appendix~\ref{app:reduction}.

The household assigns weights $(\alpha_R,\alpha_D)\in[0,1]^2$ to the two goals, with $\alpha_R+\alpha_D=1$. The random goal is successfully funded when pre-payment wealth at $\tau$ is at least $R$, whereas the fixed goal is successfully funded when $d_T=1$. Hence
\begin{align}
J(w_0;\piv,d_T)
&:=\mathbb E\left[
\alpha_R\mathbf 1_{\{\widehat W_{\tau-}^{w_0;\piv,d_T}\ge R\}}
+\alpha_Dd_T\right],
\label{eq:simple-obj}\\
V^{R,D}(0,w_0)
&:=\sup_{\piv\in\mathcal A^K_{0,\infty},\,
d_T\in\mathcal D_T(w_0,\piv)}J(w_0;\piv,d_T).
\label{eq:VRDK.prob.def.t0}
\end{align}
The criterion is the optimized weighted sum of the probabilities of funding the two goals. For $0\le t<T$, $V^{R,D}(t,w)$ denotes the corresponding value conditional on $\{\tau>t\}$, starting from pre-payment wealth $w$ at time $t$. By memorylessness, conditional on $\{\tau>t\}$ the residual time $\tau-t$ is exponential with rate $\lambda$ and is independent of both the future
market increments and $(R,G)$.

We can notice that this problem has a triangular structure. If the random shock arrives before
$T$, it is resolved immediately and only the fixed-deadline goal remains. If $T$ arrives first, the household observes $G$, chooses whether to fund the fixed-deadline goal, and then faces only the random-deadline goal. Therefore, as in our motivating binomial example, the joint value can be constructed from the two single-goal continuation values developed in Section~\ref{sec:single}. Section~\ref{sec:dual} then combines these values and identifies the optimal decision for the household at the deadline time. As noted earlier, this discretion is contrasted with the committed benchmark in Section~\ref{sec:forced-funding} which instead fixes $d_T=\mathbf 1_{\{Y_{T-}^{w_0;\piv}\ge G\}}$.

\section{Single-Goal Continuation Values}
\label{sec:single}

We first analyze the two single-goal problems that serve as continuation values for the joint-goal problem. The fixed-deadline problem rewards terminal accumulation of wealth, whereas the random-deadline problem rewards the household for maintaining sufficient wealth at every potential arrival time.

We use the pre-arrival control class $\mathcal A^{K,{\rm pre}}_{t,u}$ defined in Section~\ref{sec:gen.model}. Before the only goal in either sub-problem is revealed, admissible controls can be represented using only market information. Independence of $G$ from market information implies that the probability of funding the fixed-deadline goal, conditional on terminal wealth, is $F_G(\widetilde W_T)$. Similarly, independence of $R$ and the exponential deadline implies that the probability of funding the random-deadline goal
can be written as an exponentially weighted average of $F_R(\widetilde W_t)$. The two single-goal values are therefore
\begin{align}
V^D(t,w)
&:=
\sup_{\piv\in\mathcal A^{K,{\rm pre}}_{t,T}}
\mathbb E\!\left[F_G\!\left(\widetilde W_T^{t,w;\piv}\right)\right],
&& (t,w)\in[0,T]\times[0,\infty),
\label{eq:VDK-def}\\
V^R(w)
&:=
\sup_{\piv\in\mathcal A^{K,{\rm pre}}_{0,\infty}}
\mathbb E\!\left[
\int_0^\infty \lambda e^{-\lambda s}
F_R\!\left(\widetilde W_s^{0,w;\piv}\right)\,ds
\right],
&& w\ge0.
\label{eq:VRK-random-amount}
\end{align}
The first value is time dependent as the remaining time to $T$ influences the viability of the goal. The second is stationary because the random deadline is memoryless. We can also determine two boundary values for $V^D$ and $V^R$. For one, when wealth starts at zero it remains there. Additionally, the admissible zero strategy locks in success once wealth reaches a level at which it is guaranteed to meet the goal by investing in the risk-free security. Thus $V^D(t,0)=0$ and $V^D(t,w)=1$ for $w\ge b_G$, while $V^R(0)=0$ and $V^R(w)=1$ for $w\ge b_R$.

Both problems share the Hamiltonian
\begin{equation}\label{eqn:HK}
H(w,p,q):=
\sup_{\piv\in\Pi_K}
\left\{
wp\,\piv^\top\thetavec+
\frac12 w^2q\,\piv^\top\Sigmav\piv
\right\}.
\end{equation}
In the following we set out equations that the two value functions should satisfy in terms of \eqref{eqn:HK}.

\subsection{Fixed-Deadline Goal}
\label{sec:fixed}

The fixed-deadline value is associated with the following finite-horizon problem. On $\Omega_D:=(0,T)\times(0,b_G)$, its HJB equation is
\begin{equation}
\label{eq:hjb-rgdt-clipped}
\partial_t V^D(t,w)+
H\!\left(w,V_w^D(t,w),V_{ww}^D(t,w)\right)=0,
\qquad (t,w)\in\Omega_D,
\end{equation}
with boundary data
\begin{equation}
\label{eq:rgdt-bdy-data}
V^D(T,w)=F_G(w),\quad 0\le w\le b_G,
\qquad
V^D(t,0)=0,\quad V^D(t,b_G)=1,\quad 0\le t\le T.
\end{equation}

\subsection{Random-Deadline Goal}
\label{sec:random}
On $\Omega_R:=(0,b_R)$, the HJB equation for the random-deadline value is
\begin{equation}
\label{eq:hjb-random-clipped}
\lambda V^R(w)-H\!\left(w,(V^R)'(w),(V^R)''(w)\right)-
\lambda F_R(w)=0,
\qquad w\in\Omega_R,
\end{equation}
with boundary data
\begin{equation}
\label{eq:VRK-bdy-data}
V^R(0)=0, \qquad V^R(b_R)=1.
\end{equation}

\begin{proposition}
\label{prop:single-goal-viscosity}
\label{prop:VDK-viscosity}
\label{prop:VRK-viscosity}
The functions $V^D$ and $V^R$ are bounded and continuous on $[0,T]\times[0,\infty)$ and $[0,\infty)$, respectively. The restriction of $V^D$ to $[0,T]\times[0,b_G]$ is the unique bounded continuous function satisfying \eqref{eq:hjb-rgdt-clipped} and \eqref{eq:rgdt-bdy-data} in the viscosity sense on $\Omega_D$. The restriction of $V^R$ to $[0,b_R]$ is the unique bounded continuous function satisfying \eqref{eq:hjb-random-clipped} and \eqref{eq:VRK-bdy-data} in the viscosity sense on $\Omega_R$.
\end{proposition}

The proof is given in Appendix~\ref{app:viscosity}. In Appendix~\ref{app:single-goal-benchmarks} we discuss two benchmarks for these solutions. Namely, the deterministic-amount, unconstrained benchmarks of \citet{Browne1999} and \citet{BY16}, together with the corresponding feedback formulas. They are not needed for the joint-goal characterization, but Section~\ref{sec:asymptotics} does make use of them to validate the numerical solutions.

\section{The Joint-Goal Problem}
\label{sec:dual}

We now combine the two single-goal building blocks. Throughout this section, the standing assumptions of Section~\ref{sec:gen.model} are in force. Recall that $V^{R,D}(t,w)$ is the joint-goal value, conditional on $\{\tau>t\}$, when pre-payment wealth at time $t<T$ is $w$. Under this conditional law, $\tau-t$ is exponential with rate $\lambda$ and is independent of future market returns and $(R,G)$. We write $\E_{t,w}$ for expectation under this conditional law when pre-payment wealth starts from $w$ at time $t$. We write $\E_R$ and $\E_G$ for expectations over the indicated goal amount only.

If the random-deadline goal arrives first, at time $s<T$ with pre-payment wealth $w$, the household receives the random-goal reward when the goal is affordable, applies the mandatory payment rule, and then faces the fixed-deadline single-goal problem. The resulting continuation payoff is
\begin{equation}
\label{eq:rigorous-JK-def}
\mathcal J(s,w):= \E_R\!\left[
\alpha_R\mathbf 1_{\{w\ge R\}} +\alpha_D V^D\bigl(s,\Gamma_R(w)\bigr)
\right].
\end{equation}

If instead the household reaches time $T$ before the random-deadline goal arrives, then $G$ is revealed and the household chooses whether to fund it. Given wealth $w$ and the realized goal amount $g$, the payoff from a decision $d\in\Theta(w,g)$ is $\alpha_Dd+\alpha_RV^R(w-dg)$. An optimal decision is
\begin{equation}
\label{eq:rigorous-terminal-selector}
D^*(w,g)=
\begin{cases}
0, & w<g,\\[3pt]
\mathbf 1_{\{\alpha_D+\alpha_RV^R(w-g)
\ge \alpha_RV^R(w)\}}, & w\ge g,
\end{cases}
\end{equation}
where we resolve ties in favor of funding. Thus, when the goal is affordable, it is funded if and only if its reward is at least the weighted loss in the continuation value for the random goal.

For later use, let us define
\begin{align}
M(w,g)
&:=
\max_{d\in\Theta(w,g)}
\left\{\alpha_Dd+\alpha_RV^R(w-dg)\right\}
\notag\\
&=
\begin{cases}
\alpha_RV^R(w), & w<g,\\[3pt]
\alpha_RV^R(w)
+\bigl(\alpha_D+\alpha_RV^R(w-g)
-\alpha_RV^R(w)\bigr)^+, & w\ge g,
\end{cases}
\label{eq:terminal-max-representation}
\end{align}
and define the terminal operator
\begin{equation}
\label{eq:rigorous-Topt-def}
\mathcal T(w):=\E_G[M(w,G)].
\end{equation}
Appendices~\ref{app:viscosity} and~\ref{app:reduction} verify that $D^*$ is a measurable and feasible decision, that it attains the maximum in \eqref{eq:terminal-max-representation}, and that the resulting operator $\mathcal T$ is continuous (see Lemmas~\ref{lem:dual-operators-continuity} and~\ref{lem:optional-terminal-selector}). In addition, we have that $\mathcal T(0)=0$ and
$\mathcal T(w)=1$ for $w\ge b$.

The operators $\mathcal J$ and $\mathcal T$ summarize the value after the first goal event. If we condition on that event then we obtain the following reduction of the joint-goal problem to a finite-horizon control problem.

\begin{proposition}\label{prop:reduction-opt}
For every $(t,w)\in[0,T)\times[0,\infty)$,
\begin{equation}
\label{eq:rigorous-Vopt-recursive}
V^{R,D}(t,w)=
\sup_{\piv\in\mathcal A^{K,{\rm pre}}_{t,T}}
\E_{t,w}\!\left[
\int_t^T
\lambda e^{-\lambda(s-t)}
\mathcal J\bigl(s,\widetilde W_s^{t,w;\piv}\bigr)\,ds
+e^{-\lambda(T-t)}
\mathcal T\bigl(\widetilde W_T^{t,w;\piv}\bigr)
\right].
\end{equation}
\end{proposition}

The proof is in Appendix~\ref{app:reduction}, which also verifies that the feedback rule \eqref{eq:rigorous-terminal-selector} entails no loss of optimality.

Recall $b=b_R+b_G$ and set
$\Omega_{RD}:=(0,T)\times(0,b)$. The HJB equation associated with \eqref{eq:rigorous-Vopt-recursive} is
\begin{equation}
\label{eq:dual-HJB-clipped}
\partial_t u(t,w)
+H\bigl(w,u_w(t,w),u_{ww}(t,w)\bigr)
+\lambda\bigl(\mathcal J(t,w)-u(t,w)\bigr)=0,
\qquad (t,w)\in\Omega_{RD},
\end{equation}
with boundary data
\begin{equation}
\label{eq:dual-HJB-bc-opt}
u(T,w)=\mathcal T(w),\quad 0\le w\le b,
\qquad
u(t,0)=0,\quad u(t,b)=1,\quad 0\le t\le T.
\end{equation}

\begin{theorem}
\label{thm:dual-HJB-opt}
The value function $V^{R,D}$ is bounded and continuous on $[0,T)\times[0,\infty)$, equals $0$ at $w=0$ and $1$ for $w\ge b$, and extends continuously to $[0,T]\times[0,b]$. This extension is the unique bounded continuous viscosity solution of \eqref{eq:dual-HJB-clipped}--\eqref{eq:dual-HJB-bc-opt}.
\end{theorem}

The proof is in Appendix~\ref{app:viscosity}. The argument verifies existence and uniqueness for the system \eqref{eq:dual-HJB-clipped}--\eqref{eq:dual-HJB-bc-opt}, in which the two single-goal values enter through $\mathcal J$ and $\mathcal T$, by adapting standard viscosity techniques to that structure.

\section{The Value of Flexibility and Fungibility}
\label{sec:flex-fung}

The main model grants the household two freedoms. It may decline the fixed-deadline goal at $T$, and it manages both goals together from a single wealth process. In reality, households do not always operate in this fashion. For instance, a prepaid tuition contract or a rate-locked purchase agreement removes the first freedom, while investing through a dedicated tax-advantaged account removes the second. This section investigates the additional value provided by this flexibility by comparing the main model with a version of the problem that does not have these features.

\subsection{The Value of Flexibility}
\label{sec:forced-funding}

Our main model allows the household to decline the fixed-deadline goal at time $T$. We now compare it with a committed household, which must fund that goal whenever it is affordable. Such a
commitment may accompany a direct benefit, such as a locked-in rate or prepayment discount. The model isolates the value of the flexibility forgone; any direct benefit must be valued separately.

\subsubsection{The Committed-Household Restriction}
\label{sec:forced-funding-model}

The continuation payoff is unchanged when the random-deadline goal arrives before $T$. After that goal has been resolved, funding the remaining goal whenever affordable is optimal. Commitment therefore changes only the terminal payoff on the branch where the household reaches $T$ with the random-deadline goal still pending. The corresponding terminal operator is
\begin{equation}
\label{eq:rigorous-Tforced-def}
\mathcal T_{\mathrm{commit}}(w):=
\E_G\!\left[
\alpha_D\mathbf 1_{\{w\ge G\}}
+\alpha_RV^R\bigl(\Gamma_G(w)\bigr)
\right].
\end{equation}
Let $V_{\mathrm{commit}}^{R,D}$ denote the joint-goal value when the
terminal decision is constrained to fund the fixed-deadline goal
whenever it is affordable.

\begin{proposition}\label{prop:committed-characterization}
The value $V_{\mathrm{commit}}^{R,D}$ is obtained from \eqref{eq:rigorous-Vopt-recursive} by replacing $\mathcal T$ with $\mathcal T_{\mathrm{commit}}$. Its continuous extension is the unique bounded continuous viscosity solution of \eqref{eq:dual-HJB-clipped} with terminal condition $u(T,w)=\mathcal T_{\mathrm{commit}}(w)$ and the same spatial boundary conditions as in \eqref{eq:dual-HJB-bc-opt}.
\end{proposition}

The detailed reduction and viscosity statements are in Appendices~\ref{app:viscosity} and ~\ref{app:reduction}.

The terminal funding rule also identifies a simple condition under which committing to fund every affordable fixed-deadline goal entails no loss in value.

\begin{proposition}\label{prop:costless-commitment}
If $\alpha_D\ge\alpha_R$, funding every affordable fixed-deadline goal is weakly optimal. Consequently, $\mathcal T=\mathcal T_{\mathrm{commit}}$ and $V^{R,D}=V_{\mathrm{commit}}^{R,D}$ throughout the horizon. 
\end{proposition} 
\begin{proof}
For every $w\ge g$, the bound $0\le V^R\le1$ gives 
\[ \alpha_R\bigl[V^R(w)-V^R(w-g)\bigr] \le\alpha_R\le\alpha_D. \]
Thus the funding condition is satisfied whenever the goal is affordable and both formulations give the payoff $\alpha_D+\alpha_RV^R(w-g)$. When $w<g$, the goal is unaffordable and remains unfunded under both formulations, therefore they both give the payoff $\alpha_RV^R(w)$.

The payoffs therefore agree for every realization $g$, and so the two formulations have the same terminal value. Since they also have the same wealth dynamics, admissible controls, and running rewards, their value functions coincide at every time before $T$.
\end{proof}

\subsubsection{Measuring the Value of Flexibility}
\label{sec:two-option-value-objects}

The \emph{ex ante option value} is the loss in weighted goal-attainment value from committing at time $0$, before the market path, arrival time, and goal amounts are observed:
\begin{equation}
\label{eq:full-time0-option-value}
\Delta V_0(w_0):=
V^{R,D}(0,w_0)-V_{\mathrm{commit}}^{R,D}(0,w_0).
\end{equation}
If a direct contractual benefit is measured in the same objective
units and leaves the other model primitives unchanged, it should be weighed against $\Delta V_0(w_0)$.

The \emph{terminal option value} is measured at time $T$, conditional on pre-decision wealth $w$ but before $G$ is revealed:
\begin{equation}
\label{eq:terminal-option-average-G}
\Delta\mathcal T(w):=
\mathcal T(w)-\mathcal T_{\mathrm{commit}}(w)
=\E_G[\Delta(w,G)],
\end{equation}
where the pointwise gain from retaining discretion is
\begin{equation}
\label{eq:pointwise-terminal-option-gain}
\Delta(w,g):=
\begin{cases}
\bigl(\alpha_R[V^R(w)-V^R(w-g)]-\alpha_D\bigr)^+,
& w\ge g,\\[3pt]
0, & w<g.
\end{cases}
\end{equation}
The amount $R$ does not appear explicitly because its distribution and optimal investment policy from $T$ onward are already incorporated in $V^R$. When $g$ is affordable, the gain is strictly positive exactly when
\begin{equation}
\label{eq:commitment-threshold}
\alpha_D < \alpha_R\bigl[V^R(w)-V^R(w-g)\bigr].
\end{equation}
As a result, discretion is valuable when the reward from funding the fixed-deadline goal is smaller than the resulting weighted loss in continuation value for the random goal.

We have observed that the two funding rules differ only on paths for which the random-deadline goal survives to $T$. For each pre-deadline control the reduced payoffs differ by at most
$e^{-\lambda T}\|\Delta\mathcal T\|_\infty$ and so taking suprema over controls gives
\begin{equation}
\label{eq:option-value-survival-bound}
0\le \Delta V_0(w_0)
\le e^{-\lambda T}\|\Delta\mathcal T\|_\infty.
\end{equation}
As the two households optimize their pre-deadline portfolios separately, the committed household may partly offset its disadvantage by choosing a different portfolio strategy. Therefore, the difference between the optimized values can be strictly smaller than this upper bound.

\subsection{Quantifying the Value of Fungibility}
\label{sec:fungibility}
We next ask when managing the two goals through a single fungible account provides the largest benefit. Some goals, such as education savings and retirement accounts, are held in dedicated tax-advantaged vehicles and are therefore not fungible with other wealth. Other goals can be managed either in dedicated accounts or jointly with other goals. Our analysis identifies when the benefit of using a single portfolio process is largest.

The joint-account model solves the full two-goal HJB with a single shared wealth process. By comparison, a household may split its wealth into buckets for each goal by allocating wealth $w$ at time $t$ into two separate accounts. We write $a$ for the allocation to the random-deadline goal and $w-a$ for the allocation to the fixed-deadline goal. Each bucket is then managed independently under the corresponding single-goal optimal policy. Therefore the value to the ``split-bucket'' household is
\[
B(t,w) =
\max_{0\le a\le w}
\left\{ \alpha_R V^R(a)+ \alpha_D V^D(t,w-a) \right\}.
\]
This exactly accounts for the optimal split at time $t$ followed by independent optimization within each bucket thereafter.
Only $B(0,\cdot)$ enters the results below. We also denote by $a_R^*(t,w)$ an optimal allocation to the random-goal bucket. The corresponding optimal fixed-goal allocation is given by
$a_D^*(t,w):=w-a_R^*(t,w)$. 

Notice that separate accounts can shield fixed-goal wealth from the mandatory random-goal payment. We therefore report the signed pooled-minus-split difference $V^{R,D}(t,w)-B(t,w)$, which need not be nonnegative in general. In the numerical experiments below, it is nonnegative up to solver tolerance over the displayed wealth ranges. We refer to this positive difference as the value of fungibility.

Moreover, note that there is no new HJB equation. Once $V^R$ and $V^D$ are computed, the split problem is a one-dimensional optimization over the allocation $a$. If the single-goal value functions are differentiable and the optimizer is interior, the optimal split satisfies
\[
\alpha_R V^R_w\bigl(a_R^*(t,w)\bigr)
=
\alpha_D V^D_w\bigl(t,a_D^*(t,w)\bigr).
\]
So the household allocates the marginal dollar to whichever bucket gives the larger weighted marginal increase in goal-attainment value. In general this condition does not give a closed form, so we solve the scalar split problem numerically.

\section{Model Calibration}
\label{sec:calibration}

In this section, we describe the numerical procedure used to compute the optional-funding value $V^{R,D}$. The committed value $V_{\mathrm{commit}}^{R,D}$, which is needed for the flexibility comparisons, is obtained from the same backward recursion after replacing the terminal condition. We then specialize to one risky asset and calibrate the accumulation-goal, precautionary-goal, arrival, and market parameters to U.S. labor-market, household expenditure and income, and asset-return data, detailed in Section~\ref{sec:estimation}. These inputs are then used in a sensitivity analysis of the value function and optimal policy. Sections~\ref{sec:optional-graphs} and \ref{sec:fungibility-numerics} quantify the value of flexibility and the value of fungibility introduced in Sections~\ref{sec:forced-funding} and~\ref{sec:fungibility}, respectively.

Although the analytical model is formulated in time-$T$ dollars, the calibrated monetary quantities and all figures beginning with Figure~\ref{fig:calibration-retirement-horizon} are reported in date-$0$ present-value dollars. This convention allows retirement, college, and precautionary goals to be stated in economically meaningful amounts consistent with the current income and expenditure data used to calibrate them. For a time-$T$ wealth coordinate $w$, we define
\[
w^{\mathrm{PV}}=e^{-rT}w,
\qquad
w_0^{\mathrm{PV}}=e^{-rT}w_0.
\]
In the numerical sections, $G$, $R$, their scale parameters, and their support bounds likewise denote date-$0$ goal amounts. Accordingly, a curve labeled $V^{R,D}(0,w_0^{\mathrm{PV}})$ represents
\[
w_0^{\mathrm{PV}}
\longmapsto
V^{R,D}\!\left(0,e^{rT}w_0^{\mathrm{PV}}\right),
\]
with the goal distributions rescaled by the same factor.
We suppress this conversion in the other value, policy, and
terminal objects as well. Scale invariance allows the solver to work directly in present-value coordinates without changing values, funding decisions, or portfolio weights at corresponding states.

We interpret the calibrated market returns as nominal so a present-value goal amount $A$ requires a nominal payment $e^{rs}A$ if it is funded at date $s$. The horizon comparisons hold these present-value goal budgets fixed, with date $0$ reset to each household's saving-start date. The two scale-free validation figures below retain ordinary wealth notation.

\subsection{Algorithm}
The reported implementation specializes the model to one risky asset. We first solve each of the single-goal benchmark problems. Let $x=\log w$ and define $U(t,x)=V(t,e^x)$ for the corresponding value function $V$. The fixed-deadline, random-deadline, and two-goal problems are solved on separate uniform log-wealth grids. In each case, the upper boundary is determined by the relevant goal support: $b_G$ for the fixed-deadline problem, $b_R$ for the random-deadline problem, and $b_R+b_G$ for the two-goal problem. The grids therefore need not coincide. For the time-dependent problems, we use a nonuniform backward time grid graded toward $T$, where the value function and policy change most rapidly. The date $T-1/12$ is included explicitly when it is needed for the reported one-month-before-deadline policy calculations.

Conditional on a provisional policy $\pi$, the controlled operator can be written as
$$
\mathcal L^{\pi}U=\left(\pi(\mu-r)-\tfrac12\sigma^2\pi^2\right)U_x
+ \left(\tfrac12\sigma^2\pi^2\right)U_{xx}.
$$
We approximate $U_{xx}$ by central differences and $U_x$ by the one-sided stencil selected according to the sign of the log-wealth drift
$$
a(\pi)=\pi(\mu-r)-\tfrac12\sigma^2\pi^2.
$$
This produces a monotone tridiagonal discretization. Importantly, policy evaluation and policy improvement use the same stencil-dependent discrete Hamiltonian.

For the fixed-deadline problem, an implicit step from $t_{n+1}$ to $t_n$ has the form
$$
\left(I-\Delta t_n \mathcal L^{\pi}\right)U^{n}
=U^{n+1},
$$
whereas for the two-goal problem the killing and continuation terms give
$$
\left(I-\Delta t_n \mathcal L^{\pi}
+\Delta t_n\lambda I\right)U^{n}
=U^{n+1}
+
\Delta t_n\lambda\mathcal J(t_n,\cdot).
$$
For a fixed policy, the resulting tridiagonal systems are solved using a banded linear solver.

At each time slice, Howard policy iteration alternates between policy evaluation and pointwise maximization of the discrete Hamiltonian over $[-K,K]$. The policy updates are undamped. Iteration stops when the largest scaled improvement in the discrete Hamiltonian falls below the prescribed tolerance. We then perform a final policy-evaluation step and re-evaluate the Hamiltonian at the resulting value function. Failure to satisfy the stopping criterion before the iteration cap raises an error rather than returning an unconverged solution. 

The random-deadline problem is stationary and is solved on its own log-wealth grid without time stepping. For a fixed policy, its discretized HJB reduces to
$$
\left(\lambda I-\mathcal L^{\pi}\right)U
=\lambda F,
$$
where $F(x)=F_R(e^x)$ is the transformed random-goal distribution. Howard iteration therefore alternates between a stationary banded solve and pointwise policy improvement, using the same stencil-consistent discrete Hamiltonian and final-evaluation procedure as in the time-dependent problems.

Once the single-goal value functions have been computed, they enter the two-goal problem through the continuation operator $\mathcal J$ and the terminal operator $\mathcal T$. Expectations under continuous goal distributions are evaluated using Gaussian quadrature split at the relevant affordability thresholds. This split ensures that discontinuities induced by the funding and payment constraints occur at quadrature boundaries rather than inside a quadrature interval. Exact point masses are handled separately and are used only in explicitly labeled analytical benchmark calculations.

Whenever post-payment wealth does not coincide with a grid point, the precomputed continuation value is evaluated by interpolation in log wealth. The numerical procedure therefore follows the theoretical decomposition directly: we first compute the two single-goal building blocks, use them to construct $\mathcal J$ and $\mathcal T$, and then solve the full joint-goal HJB backward in time using these objects as inputs.

\subsection{Validation Against Known Closed Forms}
\label{sec:asymptotics}

Before calibrating the model, we verify numerically, rather than illustrate a new economic effect, that the two single-goal building blocks approach their known unconstrained closed forms as the control cap $K$ increases. This check corroborates that the clipped-control HJB solutions used throughout the numerical analysis correctly recover the closed forms of \citet{Browne1999} and \citet{BY16} in the appropriate limit. Because wealth and goal amounts are multiplied by the same factor
under the present-value conversion, the ratios $w_0/g$ and $w_0/c$ are unchanged.

For a deterministic fixed-deadline goal $G\equiv g$, the clipped value $V_K^D(0,w_0)$ should approach the closed-form expression in \citet{Browne1999},
\[
V^D_{\mathrm{Browne}}(0,w_0;g) =
\Phi\!\left(
\Phi^{-1}\!\left(\frac{w_0}{g}\right)+\gamma\sqrt{T}
\right),
\qquad 0<w_0<g,
\]
with value zero at $w_0=0$ and value one for $w_0\geq g$.

For a deterministic random-deadline goal $R\equiv c$, the clipped
value $V_K^R(w_0)$ should approach
\[
V^R(w_0)=
\begin{cases}
(w_0/c)^\kappa, & 0\leq w_0<c,\\
1, & w_0\geq c,
\end{cases}
\qquad
\kappa=\frac{\lambda}{\lambda+\tfrac12\gamma^2}.
\]
These convergences are illustrated in Figures~\ref{fig:limit-fixed-to-browne} and \ref{fig:limit-random-to-power}. The online appendix, Section~\ref{app:ec-asymptotics}, examines three further limiting regimes not needed for the main argument. Namely, the $T\downarrow0$ terminal limit, the $\lambda\to\infty$ instant-arrival limit, and the resulting nonmonotonicity of the option value in $\lambda$.

\begin{figure}[H]
\centering
\includegraphics[width=0.4\textwidth]{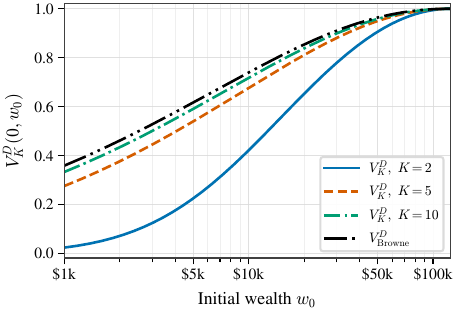}
\caption{Fixed-deadline single-goal benchmark as a function of initial wealth. Clipped HJB values $V_K^D(0,w_0)$ for $K\in\{2,5,10\}$ are compared with the unconstrained Browne value $V^D_{\mathrm{Browne}}(0,w_0;g)$.}
\label{fig:limit-fixed-to-browne}
\end{figure}

\begin{figure}[H]
\centering
\includegraphics[width=0.4\textwidth]{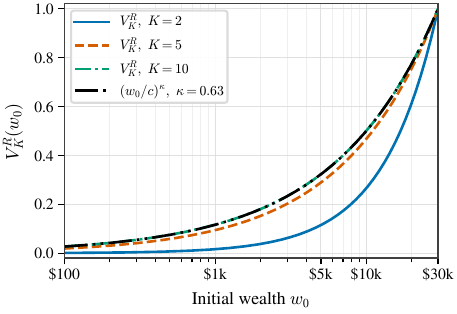}
\caption{Random-deadline single-goal benchmark as a function of initial wealth. The clipped HJB values $V_K^R(w_0)$ for $K\in\{2,5,10\}$ are compared with $(w_0/c)^\kappa$ below the goal level, where $\kappa=\lambda/(\lambda+\gamma^2/2)$. Once $K$ is at least as large as the unconstrained optimal constant risky fraction, the capped and unconstrained values coincide; remaining discrepancies reflect numerical approximation.}
\label{fig:limit-random-to-power}
\end{figure}

\subsection{Estimation Procedure}
\label{sec:estimation}
We now describe the calibration of each parameter, including the empirical sources used and the range of values considered in the sensitivity analysis. 

\paragraph{Arrival intensity $\lambda$ and probability distribution of the precautionary goal.}

The parameter $\lambda$ governs the rate of the exponential clock for the random-deadline precautionary goal. The Bureau of Labor Statistics Job Openings and Labor Turnover Survey (JOLTS) reports monthly layoff and discharge rates by state\footnote{We refer to \url{https://www.bls.gov/news.release/jltst.t05.htm} for details.}. The reported monthly rates range from $0.5\%$ to $2.9\%$, corresponding approximately to annual exponential intensities from $0.06$ to $0.35$. We use $\lambda\in[0.06,0.40]$, with $0.40$ as a wider upper sensitivity endpoint, and set the baseline to $\lambda=0.20$.

Next, we need to estimate the distribution of losses resulting from a job loss. To do so, we use two complementary sources of data. First, we take Consumer Expenditures data for the year 2024 published by the US Bureau of Labor Statistics\footnote{See \url{https://www.bls.gov/news.release/cesan.nr0.htm} for details.}. We take the third quintile of the average annual expenditures as our estimate of what a median-income household spends annually. The reasoning is as follows: the five income quintiles divide households into five equal-sized groups ranked from lowest to highest income. The third quintile is the middle group, containing households ranked from the 40th to the 60th percentile of the income distribution. The household at the exact median of the income distribution, i.e., the 50th percentile, falls inside the third quintile by construction. Hence, the average expenditure of the third quintile, given by \$66,900, is the best available approximation of what a median-income household spends annually. Next, we use \citet{duly2003necessities}, which shows that for the third income quintile the shares are 14.9\% for food, 25\% for housing, and 4.7\% for apparel. Summing up, $44.6\%$ of total expenditure is allocated to necessities. Combined with the previous estimate, we obtain that the total dollar amount of necessities is \$29,837.40. We therefore choose the continuous near-degenerate distribution for $R$, truncating symmetrically to preserve the mean:
$$
R\sim\mathcal N(R_0,s_R^2)
\mid R\in[R_0-3s_R,R_0+3s_R],
\qquad R_0=\$29{,}837.40,
\qquad s_R=10^{-3}R_0.
$$
The support is then approximately $[\$29{,}747.89,\ \$29{,}926.91]$, and the conditional standard deviation is approximately $\$29.44$.

\paragraph{Accumulation-goal amount distributions.} We model the fixed-deadline accumulation-goal size through a (truncated\footnote{We truncate the lognormal distribution symmetrically in log space at three log-standard deviations from its center; that is, we set
$a_G=\exp(\mu_G-3\sigma_G)$ and $b_G=\exp(\mu_G+3\sigma_G)$, and use the lognormal law conditional on $G\in[a_G,b_G]$, with the density renormalized on this compact support to align with the theory.}) lognormal distribution, $\log G \sim \mathcal{N}(\mu_G, \sigma_G^2)$, calibrated separately for two fixed horizon goals, specifically retirement and college saving.

{\bf Retirement goal}: Let $G_0=e^{\mu_G}$ denote the median retirement wealth goal. We calibrate $G_0$ using the income replacement rates estimated by \citet{jpmorgan2025replacement} from longitudinal Chase household transaction data. For a \emph{high-income household} with pre-retirement annual income $Y=\$300{,}000$, the estimated replacement rate is approximately $55\%$. Under our simplifying assumption that accumulated wealth finances the entire income target using a $4\%$ annual withdrawal rule, this gives $G_0=0.55Y/0.04=13.75Y=\$4.125$ million. For a \emph{low-income household} with $Y=\$30{,}000$, the estimated replacement rate is approximately $104\%$, giving $G_0=1.04Y/0.04=26Y=\$780{,}000$ under the same assumption. We set $\mu_G=\log G_0$ and treat $\sigma_G$ as a modeling choice.

{\bf College goal}: We calibrate the median
$G_0=e^{\mu_G}$ using cost-of-attendance data from \citet{collegeboard2025pricing}. It reports average annual total cost-of-attendance budgets---including tuition and fees, housing, food, books, transportation, and other expenses---of \$30{,}990 for full-time in-state students at public four-year institutions and \$65{,}470 at private nonprofit four-year institutions in 2025--26. Multiplying by four yields approximate four-year average budgets of \$124{,}000 and \$262{,}000 respectively, giving two baseline calibrations:
\begin{itemize}
    \item {Public institution:} $\mu_G = \log(124{,}000)
          \approx 11.73$, corresponding to a median four-year cost
          of \$124{,}000.
    \item {Private institution:} $\mu_G = \log(262{,}000)
          \approx 12.48$, corresponding to a median four-year cost
          of \$262{,}000.
\end{itemize}

The variance $\sigma_G^2$ reflects uncertainty
about future tuition trajectories, arising from uncertain future tuition growth rates, changes in institutional aid policy, and
household income volatility that affects net costs. We can treat $\sigma_G$ as a free parameter and conduct sensitivity analysis over $\sigma_G \in [0.30, 0.70]$.

\paragraph{Accumulation-goal horizon.} The parameter $T$ represents the total saving horizon from the time the household begins accumulating wealth towards the fixed-deadline goal. For {retirement saving}, we assume households enter the labor market at age 25 and retire at age 65, giving a baseline horizon of $T = 40$ years. We conduct sensitivity analysis over $T \in \{15, 25, 35, 40\}$, where all cases assume retirement at age 65 and $T$ reflects the age at which the household holds $w_0^{\mathrm{PV}}$. Specifically, $T = 40$ corresponds to a household holding $w_0$ at age 25 upon entering the labor market; $T = 35$ to one holding $w_0^{\mathrm{PV}}$ at age 30; $T = 25$ to one holding $w_0^{\mathrm{PV}}$ at age 40; and $T = 15$ to one holding $w_0^{\mathrm{PV}}$ at age 50, close to retirement. \citet{lusardi2007baby} document heterogeneity in the age at which US households begin actively planning for retirement, which strongly depends on financial literacy.

For {college saving}, a household holding $w_0^{\mathrm{PV}}$ at the birth of a child faces a horizon of $T = 18$ years. This calibration applies regardless of the household's own age at the time of the child's birth. We set $T = 18$ as the baseline for college saving and note that households holding $w_0^{\mathrm{PV}}$ later, for example, when the child enters primary school at age 6, face a compressed horizon of $T = 12$. Hence, we conduct sensitivity analysis over $T\in\{10,12,15,18\}$.

\paragraph{Goal weights $\alpha_R$ and $\alpha_D$.}

The weights $\alpha_R$ and $\alpha_D$ are preference parameters that are not directly observable. We treat them as free parameters. We fix $\alpha_R:=1-\alpha_D$ and vary $\alpha_D$ to assess the impact of changes in preference weights on the optimal funding and investment strategy. Unless otherwise specified, preference weights in the following graphs are fixed at $(\alpha_D,\alpha_R)=(0.50,0.50)$.

\paragraph{Market parameters.} We use the quarterly data in \citet{LettauLudvigson2001}. The average excess log return on the market is approximately $0.016$, while the log risk-free rate (proxied by the Treasury bill rate) is about $0.01$ per quarter. Annualizing implies a log equity premium of roughly $0.064$ and a risk-free rate of about $4\%$ per year. To convert the log excess return into the corresponding arithmetic equity risk premium, we account for the volatility correction: if the annualized volatility of equity returns is approximately $16\%$ (consistent with the sample moments), then
\[
\mu - r \approx 0.064 + \tfrac{1}{2}\sigma^2 \approx 0.064 + 0.5 \times (0.16)^2 \approx 0.077,
\]
yielding an equity risk premium of about $7.5\%$--$8\%$. The corresponding Sharpe ratio is therefore
\[
\frac{\mu - r}{\sigma} \approx \frac{0.077}{0.16} \approx 0.5.
\]
Overall, these estimates imply a risk-free rate of about $4\%$, an equity premium close to $8\%$, annualized volatility around $16\%$, and a Sharpe ratio in the range of $0.4$--$0.5$.

Table~\ref{tab:baseline-calibration} summarizes the baseline value used for each parameter, together with the range considered in the sensitivity analysis below and the source of the estimate.

\begin{table}[H]
\centering
\caption{Baseline calibration and sensitivity ranges. Dollar goal amounts are stated in date-$0$ values. Ranges marked ``free'' have no direct empirical counterpart and are varied for sensitivity analysis only.}
\label{tab:baseline-calibration}
\small
\setlength{\tabcolsep}{4pt}
\begin{tabular}{@{}p{0.27\textwidth} l p{0.20\textwidth} p{0.28\textwidth}@{}}
\toprule
Parameter & Baseline & Range considered & Source \\
\midrule
Arrival intensity $\lambda$ & $0.20$ & $[0.06,\,0.40]$ & BLS JOLTS layoff/discharge rates \\
Precautionary goal scale $R_0$ (near-degenerate continuous) & $\$29{,}837.40$ & [\$29{,}747.89,\ \$29{,}926.91] & BLS Consumer Expenditure Survey; \citet{duly2003necessities} \\
College goal scale $G_0$ (public) & $\$124{,}000$ & --- & \citet{collegeboard2025pricing} \\
College goal scale $G_0$ (private) & $\$262{,}000$ & --- & \citet{collegeboard2025pricing} \\
Retirement goal scale $G_0$ (high income) & $\$4.125$ million & --- & \citet{jpmorgan2025replacement} \\
Retirement goal scale $G_0$ (low income) & $\$780{,}000$ & --- & \citet{jpmorgan2025replacement} \\
Goal dispersion $\sigma_G$ & $0.50$ & $[0.30,\,0.70]$ & free \\
College horizon $T$ & $18$ years & $\{10,12,15,18\}$ & saving from birth to college \\
Retirement horizon $T$ & $40$ years & $\{15,25,35,40\}$ & saving from age 25 to retirement at 65 \\
Preference weight $\alpha_D$ & $0.50$ & $[0,1]$, $\alpha_R=1-\alpha_D$ & free \\
Risk-free rate $r$ & $4\%$/year & --- & \citet{LettauLudvigson2001} \\
Equity premium $\mu-r$ & $7.5$--$8\%$/year & --- & \citet{LettauLudvigson2001} \\
Volatility $\sigma$ & $16\%$/year & --- & \citet{LettauLudvigson2001} \\
Portfolio cap $K$ & $5.0$ & $\{2,5,10\}$ for asymptotic checks & modeling choice \\
\bottomrule
\end{tabular}
\end{table}

%Three experiments below depart from the baseline weights. 
At $\alpha_D=\alpha_R=0.50$, funding every affordable fixed-deadline goal is weakly optimal, so the optional and committed values coincide. %Among preference weights in the fungibility experiment, 
We find that the pooled-minus-split difference has its largest peak at equal priority. In Section~\ref{sec:optional-graphs}, we deviate from the baseline preference weights to illustrate different parts of the model.We set $\alpha_D=0.10$ because declining an affordable fixed goal can be strictly optimal only when $\alpha_D<\alpha_R$. Appendix~\ref{app:senscalibr} sets $\alpha_D=0.99$ as the wealth nonmonotonicity generated by the mandatory random-goal payment requires $\alpha_D>\alpha_R$. %These  extreme weight makes the mechanism clearly visible. %These choices isolate different parts of the model.

\subsection{Growth Crowding-Out and Deadline Pressure}
\label{sec:effects-numerics}

We now report the two economic effects identified in Section~\ref{sec:intro}, specifically the growth crowding-out effect and the deadline pressure effect, under the baseline calibration. We document these effects using the optional-funding value function $V^{R,D}(0,w_0^{\mathrm{PV}})$ and its associated optimal risky position $\pi^*(0,w_0^{\mathrm{PV}})$. Under the baseline weights $\alpha_D=\alpha_R=0.50$, Proposition~\ref{prop:costless-commitment} implies that funding every affordable accumulation goal is weakly optimal and the optional and committed value functions coincide. The horizon and arrival-intensity results therefore apply to both formulations. For the comparative statics on preference weights, we use the optional-funding solution at every weight. Robustness to goal dispersion, goal specification, and horizon is reported in the online appendix. %Appendix~\ref{app:senscalibr} also documents a nonmonotonicity of the value function in wealth that emerges only at arrival intensities far outside the calibrated range.

\begin{figure}[H]
\centering
\includegraphics[width=0.8\textwidth]{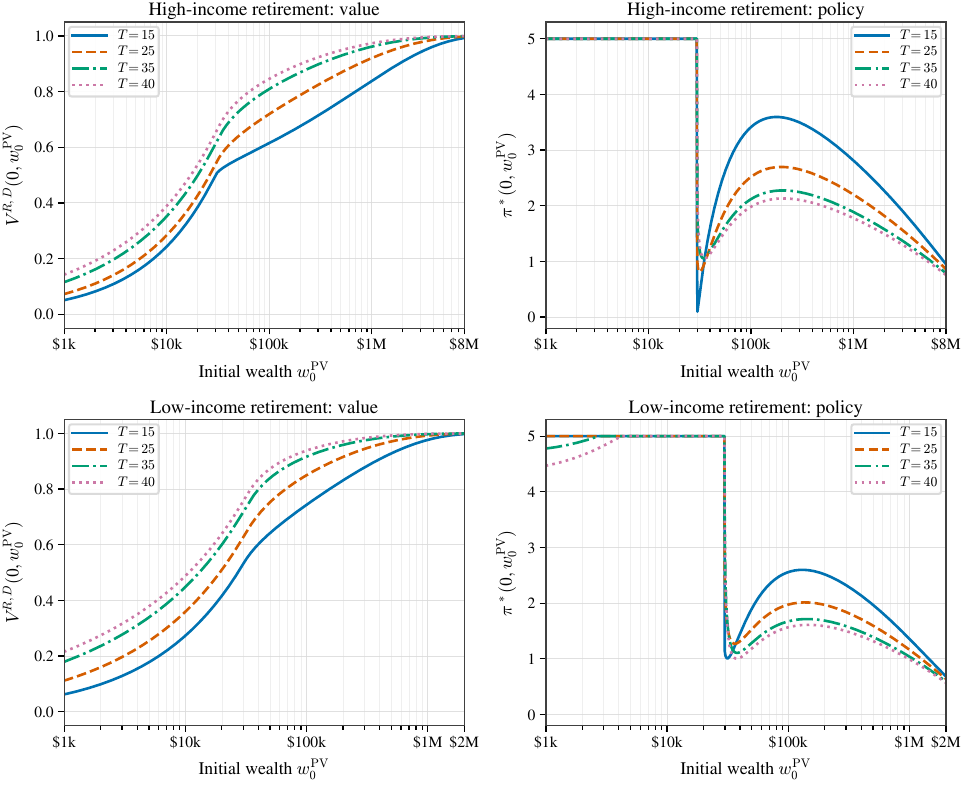}
\caption{Deadline pressure under optional funding. The left column plots $V^{R,D}(0,w_0^{\mathrm{PV}})$ and the right column plots $\pi^*(0,w_0^{\mathrm{PV}})$ against initial wealth for high-income (top) and low-income (bottom) accumulation retirement targets. Curves compare $T\in\{15,25,35,40\}$ years, corresponding to saving from ages $50,40,30,25$ to age 65. At intermediate wealth, shorter horizons require more initial wealth to attain a given value and, farther above the precautionary-goal affordability region, generate more aggressive catch-up risk-taking.}
\label{fig:calibration-retirement-horizon}
\end{figure}

In Figure~\ref{fig:calibration-retirement-horizon}, we vary the horizon $T\in\{15,25,35,40\}$ years, corresponding to households holding initial wealth $w_0^{\mathrm{PV}}$ at ages $50$, $40$, $30$, and $25$, respectively, while all households retire at age $65$. The deadline-pressure effect is quantitatively large. In the high-income panel, the $T=40$ curve has value $0.8459$ at $w_0^{\mathrm{PV}}=\$100{,}000$, whereas the $T=15$ household requires approximately $w_0^{\mathrm{PV}}=\$1.073$ million to attain the same optimized weighted goal-attainment value. This is a $10.73$-fold wealth-equivalent difference. The retirement-goal calibration is held fixed across horizons.

In the policy panels, all four high-income curves leave the cap near $w_0^{\mathrm{PV}}=\$29{,}646$, and the corresponding low-income curves leave it near $w_0^{\mathrm{PV}}=\$29{,}714$. Thus, shorter
horizons do not visibly extend the capped region. Immediately after the sharp policy drop, the shorter-horizon positions can be lower. Farther above the precautionary-goal affordability region, however, shorter-horizon households take substantially more risk because they have less time to rebuild wealth and return to the accumulation-goal
funding region. Finally, a longer horizon increases exposure to the precautionary goal because $\Pr(\tau\leq T)=1-e^{-\lambda T}$ rises with $T$. This creates a countervailing force at low wealth, compressing the longer-horizon value curves at the far left relative to the middle and upper wealth ranges.
\begin{figure}[H]
\centering
\includegraphics[width=0.9\textwidth]{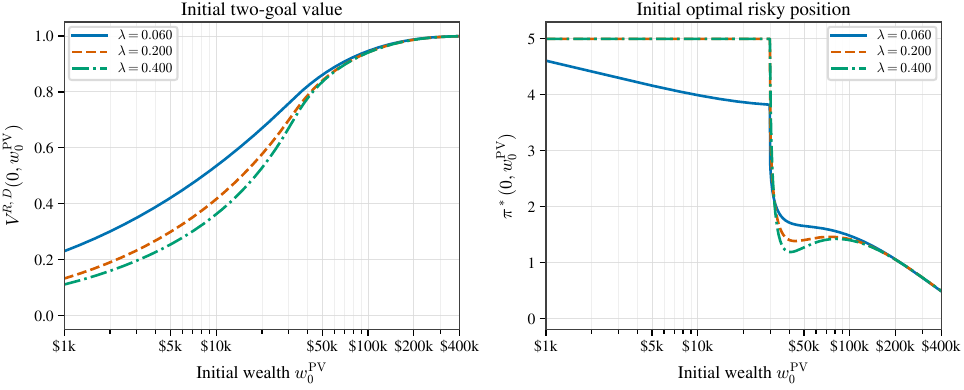}
\caption{Sensitivity of the optional-funding value
$V^{R,D}(0,w_0^{\mathrm{PV}})$ and initial risky position $\pi^*(0,w_0^{\mathrm{PV}})$ to $\lambda\in\{0.06,0.20,0.40\}$ under the public-college calibration. Higher $\lambda$ lowers value at a given level of low initial wealth and keeps the risky position at its cap over a wider range by making the precautionary claim more imminent.}
\label{fig:calibration-lambda-public-college}
\end{figure}

Figure~\ref{fig:calibration-lambda-public-college} varies $\lambda$ under the public-college baseline. Lower $\lambda$ gives the household more time to accumulate wealth before the precautionary payment, making this problem more similar to a single fixed-deadline problem. Therefore, the value and policy curves are closer to the benchmark of \citet{Browne1999}. Nevertheless, at $\lambda=0.06$, the random goal still arrives before $T=18$ with probability $1-e^{-\lambda T}\approx0.66$, so it is not negligible.

As $\lambda$ rises, two effects appear. First, the value curve lies lower at a given level of low initial wealth, reflecting the increased probability of an imminent shock that depletes resources before the household has accumulated sufficiently. Equivalently, the household needs more initial wealth to attain the same optimized value. Second, the optimal policy becomes more aggressive in the intermediate-wealth region. The $\lambda=0.20$ and $\lambda=0.40$ policies remain at the cap $K=5$ throughout much of this region. By comparison, when $\lambda=0.06$, the risky position is approximately $3.99$ at $w_0^{\mathrm{PV}}=\$10{,}000$ and $3.86$ at $w_0^{\mathrm{PV}}=\$20{,}000$. These patterns are a manifestation of the growth crowding-out effect: over the displayed sensitivity range, a higher $\lambda$ tightens the resource conflict between the two goals by making the precautionary claim more imminent, without yet producing the nonmonotonicity documented in
Appendix~\ref{app:senscalibr}.

In Figure~\ref{fig:calibration-weight-public-college}, we vary the preference weights $\alpha_D\in\{0.1,0.3,0.5,0.7,0.9\}$, with $\alpha_R=1-\alpha_D$, holding all other parameters at their calibrated
values. Once wealth approaches the random-goal scale, the value is higher for households that place more weight on that relatively attainable goal and lower for households that place more weight on the larger college target. The policy curves separate sharply near the precautionary-goal  affordability region. Households with low $\alpha_D$ reduce risk once the heavily weighted precautionary goal is nearly secured and retain the option to decline an unfavorable college payment. Households with high $\alpha_D$ remain aggressive over a wider wealth range because the college accumulation goal carries more of the objective weight. At higher wealth, the policies gradually converge as both goals become increasingly attainable.

\begin{figure}[H]
\centering
\includegraphics[width=0.9\textwidth]{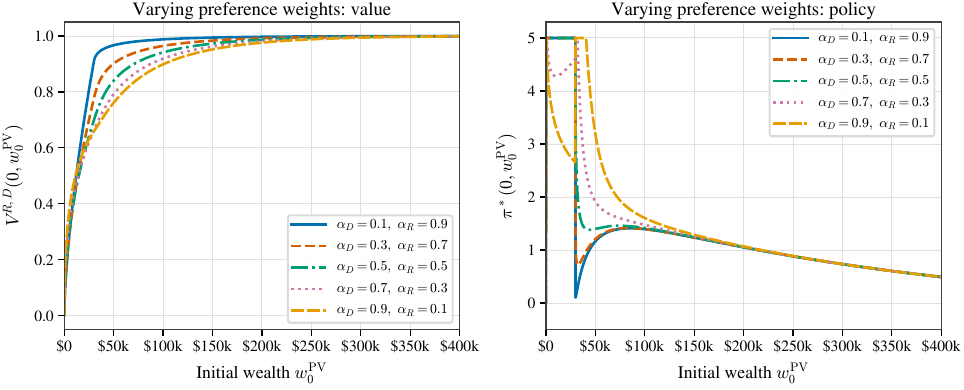}
\caption{Sensitivity of the optional-funding value $V^{R,D}(0,w_0^{\mathrm{PV}})$ and initial risky position $\pi^*(0,w_0^{\mathrm{PV}})$ to $\alpha_D$, with
$\alpha_R=1-\alpha_D$, under the public-college calibration. Low-$\alpha_D$ households reduce risk once the heavily weighted precautionary goal is nearly secured, whereas high-$\alpha_D$ households are aggressive over a wider wealth range while pursuing the larger college accumulation goal.}
\label{fig:calibration-weight-public-college}
\end{figure}

\subsection{The Value of Flexibility}
\label{sec:optional-graphs}

We now quantify the value of flexibility numerically. The baseline calibration uses preference weights $\alpha_D=\alpha_R=0.5$, under which Proposition~\ref{prop:costless-commitment} implies that flexibility has no value. To illustrate cases where the two models differ, we set $(\alpha_D,\alpha_R)=(0.10, 0.90)$ throughout this section (unless otherwise indicated): the random-deadline goal carries most of the household's objective weight, so preserving wealth for it can outweigh the direct reward from funding the accumulation goal. We fix $T=18$, $\lambda=0.20$, the accumulation goal as a truncated lognormal with median $\$124{,}000$ and $\sigma_G=0.5$, and the near-degenerate precautionary-goal distribution centered at $R_0=\$29{,}837.40$.

\subsubsection{Two Notions of the Value of Flexibility}
\label{sec:two-option-value-objects-numerical}

Recall from Section~\ref{sec:two-option-value-objects} the two measures of the value of flexibility: the ex ante option value $\Delta V_0(w_0^{\mathrm{PV}})$ of \eqref{eq:full-time0-option-value}, and the terminal option value $\Delta\mathcal T(w^{\mathrm{PV}})$ of \eqref{eq:terminal-option-average-G}--\eqref{eq:pointwise-terminal-option-gain}.
We show these numerically. By \eqref{eq:commitment-threshold}, the pointwise option gain is positive exactly when the accumulation goal is affordable, but funding reduces precautionary-goal continuation value more than the direct reward $\alpha_D$. The terminal measure captures this tradeoff at given wealth, while the ex ante measure also reflects the probability of reaching such states with the precautionary goal still pending.

On a high level, Figures~\ref{fig:option-value-main}
and~\ref{fig:three-option-terminal-average-G} show that both the ex ante and terminal option values are hump-shaped in wealth under this calibration. The cost of committing is small at low initial wealth, where the fixed-deadline goal is unlikely to become affordable by the deadline; it is also small at high wealth, where the household can fund the fixed-deadline goal while still maintaining an ample emergency buffer. The cost is largest at intermediate wealth, where paying the fixed-deadline goal under committed funding would substantially deplete the wealth available for the emergency --- precisely where a household considering a tuition prepayment plan should think hardest before committing, and where the option value characterized is most relevant.

\begin{figure}[H]
    \centering
    \begin{subfigure}[t]{0.48\textwidth}
        \centering
        \includegraphics[width=0.9\linewidth]{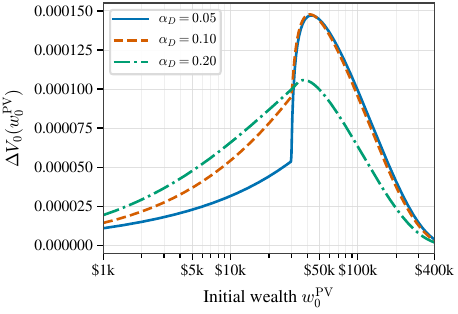}
        \caption{}
    \end{subfigure}
    \hfill
    \begin{subfigure}[t]{0.48\textwidth}
        \centering
        \includegraphics[width=0.9\linewidth]{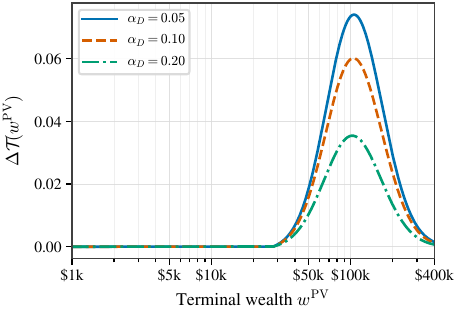}
        \caption{}
    \end{subfigure}
    \caption{Value of flexibility in fixed-deadline funding for
    $\alpha_D\in\{0.05,0.10,0.20\}$ and $\alpha_R=1-\alpha_D$.
    Panel (a) plots $\Delta V_0(w_0^{\mathrm{PV}}) =V^{R,D}(0,w_0^{\mathrm{PV}}) -V_{\mathrm{commit}}^{R,D}(0,w_0^{\mathrm{PV}})$, and panel (b) plots $\Delta\mathcal T(w^{\mathrm{PV}})=\mathcal T(w^{\mathrm{PV}}) -\mathcal T_{\mathrm{commit}}(w^{\mathrm{PV}})$. Both are largest at intermediate wealth, where funding the accumulation goal would materially reduce the continuation value of the unresolved precautionary goal. The terminal option value is zero when every possible realization of the accumulation goal is unaffordable, while the ex ante option value is small, but not necessarily zero, at very low initial wealth. Both vanish once wealth reaches the safe region in which the two goals can be funded without conflict.}
    \label{fig:option-value-main}
\end{figure}

Figure~\ref{fig:three-option-terminal-average-G} varies the dispersion $\sigma_G$ of the accumulation-goal distribution and plots the terminal option value $\Delta\mathcal T(w^{\mathrm{PV}})$ against terminal wealth $w^{\mathrm{PV}}$. The option value depends on the probability mass assigned to realizations of $G$ that are affordable but costly in terms of lost precautionary-goal continuation value. Very small realizations of $G$ generate little option value because funding them barely reduces $V^R$. Very large realizations are typically unaffordable and generate no option value either. The largest contribution comes from intermediate realizations that are affordable but reduce the continuation value of the precautionary goal. A more diffuse accumulation-goal distribution shifts mass into and out of this region, producing the variation in terminal option value visible across the $\sigma_G$ curves.

\begin{figure}[H]
\centering
\includegraphics[width=0.4\textwidth]{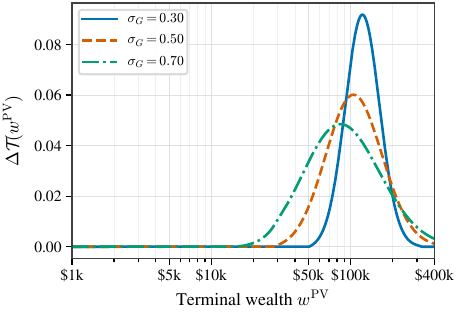}
\caption{Terminal option value $\Delta\mathcal T(w^{\mathrm{PV}}) =\mathbb E_G[\Delta(w^{\mathrm{PV}},G)]$ against terminal wealth $w^{\mathrm{PV}}$, for $\sigma_G\in\{0.30,0.50,0.70\}$. The expectation is over $G$ only. Greater dispersion spreads the option-value region over a wider wealth range and lowers its peak by moving probability mass away from accumulation-goal realizations that are simultaneously affordable and costly in terms of lost precautionary-goal continuation value.}
\label{fig:three-option-terminal-average-G}
\end{figure}

Figure~\ref{fig:three-option-policy-value} shows that optional funding changes the optimal feedback policy over time. The policy panels are computed deterministically under each funding rule's own controlled finite-state chain. One month before the deadline, the optional-funding policy assigns a larger expected risky position throughout the displayed wealth range. The difference reaches a maximum of approximately $0.3$ near $w_0^{\mathrm{PV}}=\$36{,}000$. Earlier in the horizon, however, the effect can reverse. For the profile beginning at $w_0^{\mathrm{PV}}=\$40{,}000$, the expected policy difference is negative for much of the horizon before increasing sharply near the deadline. Under the lower-intensity case $\lambda=0.06$, reported in the online appendix, the one-month-before-deadline cross-section is also positive throughout and reaches a maximum of approximately $0.33$ near $w_0^{\mathrm{PV}}=\$31{,}000$. For $t<T$, write
\[
\Delta V(t,w^{\mathrm{PV}}):=V^{R,D}(t,w^{\mathrm{PV}})-V_{\mathrm{commit}}^{R,D}(t,w^{\mathrm{PV}}).
\]

For each funding rule $q\in\{\mathrm{opt},\mathrm{commit}\}$, let
$\pi_q^*(t,w^{\mathrm{PV}})$ denote its optimal feedback policy in the displayed present-value coordinate, and let
$W_t^{q,w_0^{\mathrm{PV}}}$ denote the corresponding wealth process, also expressed in that coordinate, starting from $w_0^{\mathrm{PV}}$. Define the expected risky position by
\[
\bar\pi_q(t,w_0^{\mathrm{PV}}) :=
\mathbb E_{w_0^{\mathrm{PV}}}^{q}\!\left[
\pi_q^*\!\left(t,W_t^{q,w_0^{\mathrm{PV}}}\right)
\,\middle|\,\tau>t
\right],
\]
and define the expected policy difference by
\[
\Delta\bar\pi(t,w_0^{\mathrm{PV}}) :=
\bar\pi_{\mathrm{opt}}(t,w_0^{\mathrm{PV}}) -
\bar\pi_{\mathrm{commit}}(t,w_0^{\mathrm{PV}}).
\]
Each expectation is evaluated under the corresponding funding rule's own optimally controlled wealth distribution. Thus, $\Delta\bar\pi(t,w_0^{\mathrm{PV}})$ is a difference between two rule-specific expectations, rather than a pointwise policy difference evaluated at a common wealth state. A positive value means that retaining discretion produces the larger expected risky position.

\begin{figure}[H]
    \centering
    \begin{subfigure}[t]{0.48\textwidth}
        \centering
        \includegraphics[width=0.9\linewidth]{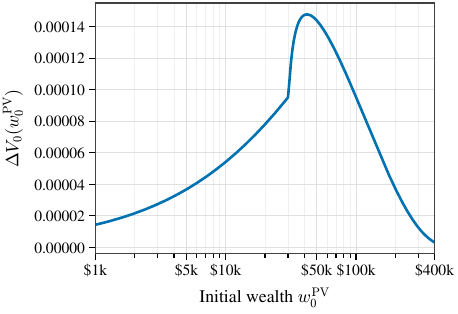}
        \caption{}
    \end{subfigure}
    \hfill
    \begin{subfigure}[t]{0.48\textwidth}
        \centering
        \includegraphics[width=0.9\linewidth]{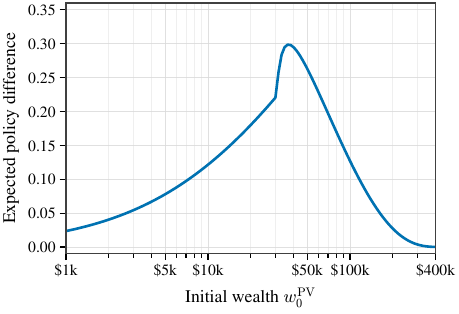}
        \caption{}
    \end{subfigure}

    \par\medskip

    \begin{subfigure}[t]{0.48\textwidth}
        \centering
        \includegraphics[width=0.9\linewidth]{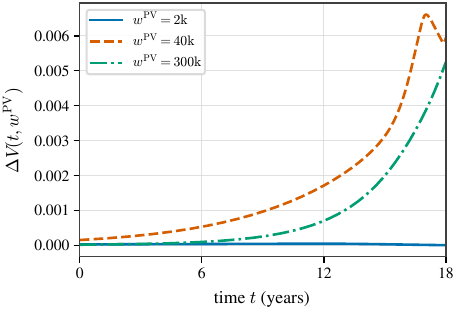}
        \caption{}
    \end{subfigure}
    \hfill
    \begin{subfigure}[t]{0.48\textwidth}
        \centering
        \includegraphics[width=0.9\linewidth]{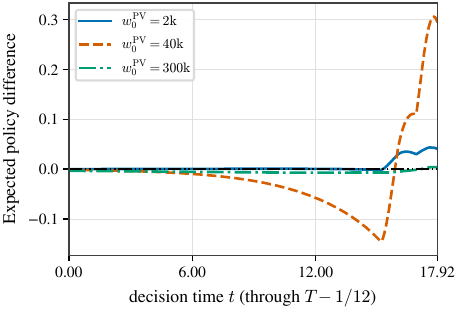}
        \caption{}
    \end{subfigure}

    \caption{Effect of retaining discretion under the baseline $\lambda=0.20$ and $(\alpha_D,\alpha_R)=(0.10,0.90)$. Panel (a) plots the time-$0$ value of discretion, $\Delta V_0(w_0^{\mathrm{PV}})$, against initial wealth. Panel (b) plots the expected policy difference $\Delta\bar\pi(T-\ell,w_0^{\mathrm{PV}})$ across initial wealth, where $\ell=1/12$. Panel (c) plots $\Delta V(t,w^{\mathrm{PV}})$ while holding wealth fixed. Panel (d) plots $\Delta\bar\pi(t,w_0^{\mathrm{PV}})$ over time, starting from the stated initial wealth. Each component expectation is computed under the corresponding funding rule's own controlled finite-state chain, conditional on survival to the evaluation date.}
    \label{fig:three-option-policy-value}
\end{figure}

Panels (c) and (d) make the timing mechanism explicit. At the fixed wealth level $w^{\mathrm{PV}}=\$2{,}000$, the fixed goal is effectively out of reach, so flexibility has almost no value. For the household beginning at $w_0^{\mathrm{PV}}=\$2{,}000$, the largest policy effect is a small positive spike immediately before the deadline, reflecting the value of preserving upside when the fixed goal is otherwise unlikely to become affordable. At $w^{\mathrm{PV}}=\$40{,}000$, wealth lies in an intermediate region, and the option value grows rapidly as $T$ approaches because a committed payment would reduce the wealth available for the random goal. For the household beginning at $w_0^{\mathrm{PV}}=\$40{,}000$, the expected policy difference is negative for much of the horizon before turning sharply positive near the deadline. At $w^{\mathrm{PV}}=\$300{,}000$, most accumulation-goal realizations are affordable, and both the value and policy effects remain muted for much of the horizon. Because $\sigma_G=0.5$, the compact support of $G$ extends approximately from $\$27{,}668$ to $\$555{,}729$. Goal realizations that are large but still affordable can leave little wealth for the precautionary goal, so the fixed-wealth value difference need not vanish near $T$. The corresponding expected policy difference also remains small but nonzero.

\subsubsection{Priority per Dollar of the Fixed-Deadline Goal}

Finally, we vary the accumulation-goal priority per dollar, measured by $(\alpha_D/G_0)\times100{,}000$. This statistic combines the reward from funding the fixed-deadline goal with its median scale and provides a convenient comparative-statics index across goals of different sizes. It is not the pointwise funding criterion, which depends on the realized goal and the resulting loss in precautionary-goal continuation value.

\begin{figure}[H]
    \centering

    \begin{subfigure}[t]{0.48\textwidth}
        \centering
        \includegraphics[width=0.9\linewidth]{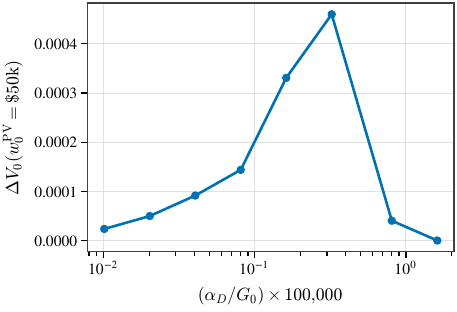}
        \caption{}
    \end{subfigure}
    \hfill
    \begin{subfigure}[t]{0.48\textwidth}
        \centering
        \includegraphics[width=0.9\linewidth]{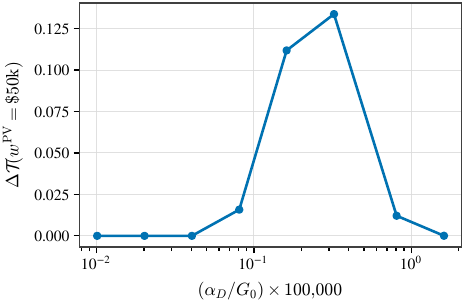}
        \caption{}
    \end{subfigure}

    \caption{Value of flexibility as fixed-goal priority per dollar varies, with $\alpha_D=0.10$ fixed and $G_0$ varied. The horizontal axis is $(\alpha_D/G_0)\times100{,}000$. Panel (a) shows $\Delta V_0$ at $w_0^{\mathrm{PV}}=\$50{,}000$ and panel (b) $\Delta\mathcal T$ at $w^{\mathrm{PV}}=\$50{,}000$. Both are hump-shaped: flexibility has little value when the goal is either unaffordable or cheap to fund.}
    \label{fig:asym-alpha-over-G-option}
\end{figure}

Both panels in Figure~\ref{fig:asym-alpha-over-G-option} display the same basic mechanism. When $\alpha_D/G_0$ is very low, the accumulation goal is large relative to its reward and unaffordable for most realizations, so the option to decline it has limited value. At intermediate values, the goal is often affordable, but funding it substantially reduces the continuation value of the precautionary goal, so flexibility is worth the most. When $\alpha_D/G_0$ is high, the accumulation goal is inexpensive, so declining it preserves little continuation wealth and the option value approaches zero.

\subsection{The Value of Fungibility}
\label{sec:fungibility-numerics}

To isolate the value of account fungibility, we use a stylized calibration with both goal scales centered at $\$30{,}000$, $(\alpha_D,\alpha_R)=(0.50,0.50)$, $T=10$, and $\lambda=0.10$ unless a panel varies that parameter. The random-deadline goal amount is a near-degenerate truncated normal. The fixed-deadline goal amount is truncated lognormal: we use $\sigma_G=0.10$ in the deadline and priority panels and $\sigma_G=0.50$ in the goal-scale and arrival-intensity panels. We report the signed pooled-minus-split difference $V^{R,D}(0,w_0^{\mathrm{PV}})-B(0,w_0^{\mathrm{PV}})$. It is nonnegative up to solver tolerance over the displayed wealth ranges, and we refer to this difference as the value of fungibility.

Figure~\ref{fig:fungibility-comparative-statics} reports this difference. At very low wealth, there is little wealth to reallocate across the two goals. At high wealth, the two dedicated accounts can cover their respective goals, so the difference collapses to zero. The difference is generally largest at intermediate wealth, where the household has enough wealth for the allocation decision to matter but not enough wealth to make both goals safe separately, although some curves display secondary humps.

More specifically, in Figure~\ref{fig:fungibility-comparative-statics}, the top-left panel shows that the shorter accumulation deadline produces the largest difference at intermediate wealth, although the curves cross at very low wealth, where additional accumulation time is valuable under both formulations. The top-right panel shows that scaling $G_0$ shifts the location of the peak but changes its height nonmonotonically, as a larger accumulation goal initially increases competition between the two goals and therefore raises the value of fungibility. Once the accumulation goal becomes too difficult to fund, however, it is unaffordable more often, causing the two formulations to behave more similarly. The bottom-left panel shows a similar effect for $R_0$: a larger precautionary goal shifts the peak toward higher wealth, while its height again changes nonmonotonically. Finally, the bottom-right panel shows that $\lambda$ has a nonmonotone effect on peak height, with $\lambda=0.05$ producing the largest peak among the displayed curves. Higher values of $\lambda$ move the high-wealth decline to the left. Overall, the effect of arrival intensity reflects both the urgency of the random-deadline goal and the time available to build the precautionary bucket.

\begin{figure}[H]
    \centering
    \begin{subfigure}[t]{0.48\textwidth}
        \centering
        \includegraphics[width=0.9\linewidth]{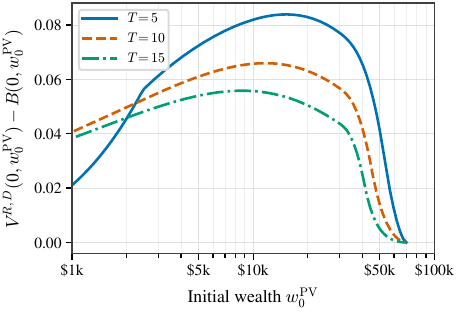}
        \caption{}
    \end{subfigure}
    \hfill
    \begin{subfigure}[t]{0.48\textwidth}
        \centering
        \includegraphics[width=0.9\linewidth]{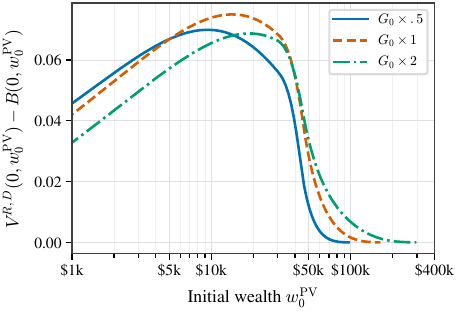}
        \caption{}
    \end{subfigure}

    \vspace{0.25cm}

    \begin{subfigure}[t]{0.48\textwidth}
        \centering
        \includegraphics[width=0.9\linewidth]{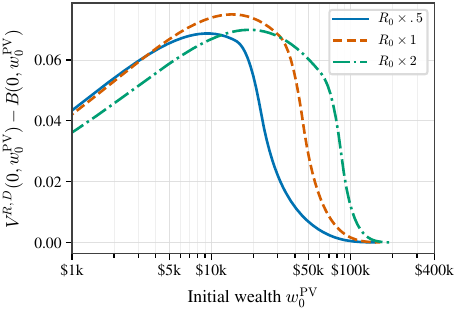}
        \caption{}
    \end{subfigure}
    \hfill
    \begin{subfigure}[t]{0.48\textwidth}
        \centering
        \includegraphics[width=0.9\linewidth]{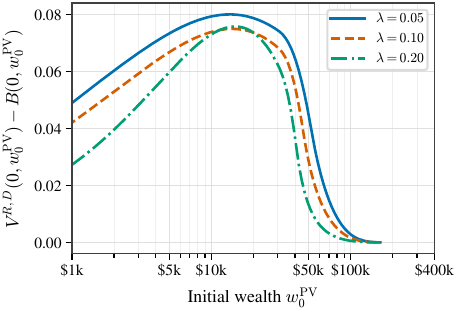}
        \caption{}
    \end{subfigure}

   \caption{Comparative statics for the pooled-minus-split difference. Panels vary (a) the accumulation-goal deadline $T$, and (b) scale $G_0$, (c) the precautionary-goal scale $R_0$, and (d) arrival intensity $\lambda$; the dashed middle curve is the panel-specific baseline. At intermediate wealth, the shorter-deadline curve is largest, although the deadline curves cross at very low wealth. Increasing either $G_0$ or $R_0$ shifts the location of the peak but has a nonmonotone effect on its height. The effect of $\lambda$ on peak height is also nonmonotone; among the displayed curves, $\lambda=0.05$ produces the largest peak. Across panels, the difference disappears as $w_0^{\mathrm{PV}}$ approaches the level at which the separate accounts can cover their respective goal supports.}
    \label{fig:fungibility-comparative-statics}
\end{figure}

\begin{figure}[H]
    \centering
    \begin{subfigure}[t]{0.48\textwidth}
        \centering
        \includegraphics[width=0.9\linewidth]{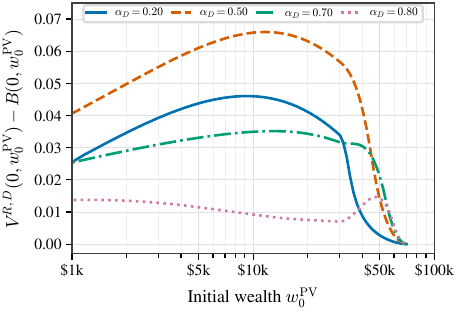}
        \caption{}
    \end{subfigure}
    \hfill
    \begin{subfigure}[t]{0.48\textwidth}
        \centering
        \includegraphics[width=0.9\linewidth]{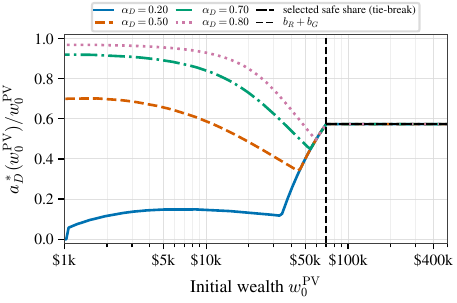}
        \caption{}
    \end{subfigure}
    \caption{Effect of accumulation-goal priority $\alpha_D$, with $\alpha_R=1-\alpha_D$. Panel (a) plots the fungibility gain $V^{R,D}(0,w_0^{\mathrm{PV}})-B(0,w_0^{\mathrm{PV}})$. Among displayed priorities, its largest peak occurs at equal priority; the high-$\alpha_D$ curve has a secondary rise. Panel (b) plots the optimal accumulation-goal bucket share $a_D^*(0,w_0^{\mathrm{PV}})/w_0^{\mathrm{PV}}$. Higher $\alpha_D$ shifts wealth toward the fixed bucket. Above the safe boundary $b_R+b_G$, many splits are optimal; the displayed high-wealth share uses the support-proportional tie-break $a_D^*(0,w_0^{\mathrm{PV}})/w_0^{\mathrm{PV}}=b_G/(b_R+b_G)$. It is not a unique economic allocation.}
    \label{fig:fungibility-priority-sweep}
\end{figure}

Figure~\ref{fig:fungibility-priority-sweep} shows that, among the displayed priorities, the pooled-minus-split difference has its largest peak at equal priority, when assigning a dollar to the wrong bucket is especially costly. The left panel also displays a smaller second rise for high $\alpha_D$ just below $b_R+b_G$. This rise is weak at $\alpha_D=0.70$ and more pronounced at $\alpha_D=0.80$. At these priorities and low wealth, the split benchmark allocates heavily toward the fixed-goal bucket, while the pooled problem is likewise dominated by the fixed goal. As wealth increases, the accumulation-goal becomes sufficiently protected for the precautionary goal to matter while wealth remains scarce enough for pooling to be valuable, producing the second hump. The hump persists until there is enough wealth to cover both goal supports. The right panel shows the associated split: higher $\alpha_D$ shifts wealth toward the accumulation-goal bucket. The curves are smooth away from the safe boundary and converge there because the boundary allocation must cover both upper support points.

We may also express the cost of account separation in dollars rather than in units of the objective by defining
\[
w_B(w_0^{\mathrm{PV}}) := \inf\left\{
\widetilde w_0^{\mathrm{PV}}\geq w_0^{\mathrm{PV}}:
B(0,\widetilde w_0^{\mathrm{PV}})
\geq V^{R,D}(0,w_0^{\mathrm{PV}})
\right\}.
\]
Thus, $\Delta w_B(w_0^{\mathrm{PV}}):= w_B(w_0^{\mathrm{PV}})-w_0^{\mathrm{PV}}$ is the additional initial wealth that a household using separate accounts requires to attain the same optimized weighted goal-attainment value as a household with pooled initial wealth $w_0$. Figure~\ref{fig:fungibility-equivalent-wealth} uses this measure to show the wealth cost of account separation.

\begin{figure}[H]
    \centering
    \includegraphics[width=0.4\textwidth]{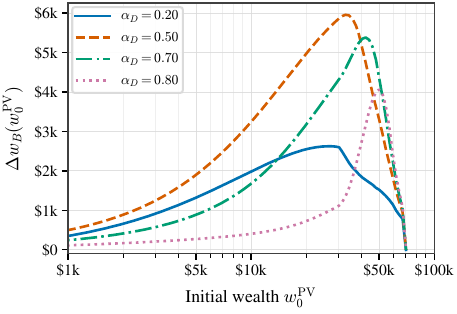}
   \caption{Equivalent-wealth cost of account separation. The figure plots $\Delta w_B(w_0^{\mathrm{PV}}) =w_B(w_0^{\mathrm{PV}})-w_0^{\mathrm{PV}}$ for $\alpha_D\in\{0.20,0.50,0.70,0.80\}$ and $\alpha_R=1-\alpha_D$. Among the displayed priorities, the premium is largest at equal priority, reaching approximately $\$6{,}000$ around $w_0^{\mathrm{PV}}=\$33{,}000$ for $\alpha_D=0.50$, and returns to zero at $b_R+b_G$, where separate accounts can fully fund both goals.}
    \label{fig:fungibility-equivalent-wealth}
\end{figure}

When the goals have equal priority, assigning wealth to the wrong account is especially costly because both goals significantly affect the objective. Consequently, a household with separate accounts requires up to approximately $\$6{,}000$ of additional initial wealth to match the pooled household. As accumulation-goal priority rises, the maximum premium shifts to higher wealth. In particular, the late increase for $\alpha_D=0.80$ is the dollar counterpart of the second hump in Figure~\ref{fig:fungibility-priority-sweep}. This is the region where the fixed-deadline goal has become sufficiently protected for the random-deadline goal to matter, but wealth remains too scarce to fund separate buffers. The premium disappears at $w_0^{\mathrm{PV}}=b_R+b_G$, where both accounts can cover their respective goal supports.

\section{Conclusion}
\label{sec:conclusion}

This paper develops a tractable framework for a household simultaneously managing a mandatory random-deadline goal and a discretionary fixed-deadline goal. Exploiting the problem's triangular structure, we characterize the joint value function as the unique viscosity solution of a Hamilton--Jacobi--Bellman equation and derive a terminal funding rule comparing the direct reward with the weighted loss in precautionary continuation value. Two effects emerge: a growth crowding-out effect, as mandatory precautionary payment for the random-deadline goal displaces growth-oriented investment, and a deadline pressure effect as shorter horizons leave less time to reach the fixed goal. The precautionary claim amplifies the latter, and both effects are quantitatively large under the reported calibration. Comparing the household's main model against a committed benchmark that has contracted away its discretion, and against a split-account benchmark that manages each goal separately, we further show that both the value of flexibility and the value of pooling wealth across goals are hump-shaped in wealth. Moreover, each peaks exactly where a goal is affordable but funding it would materially deplete the resources available for the other.

Our results have broader implications for portfolio management. For household-finance practice, robo-advisory platforms and financial planners that manage each goal in its own account are, in effect, forgoing the value of fungibility and flexibility we quantify here. Our framework offers a principled way to decide when consolidating goals into a single portfolio process is worth the loss of the tax and behavioral benefits that dedicated accounts provide. The ex ante and terminal option values developed in the paper also provide a direct way to price the flexibility a household gives up when it signs a commitment product, such as a prepaid tuition contract or a rate-locked purchase agreement, in exchange for the contract's stated benefit. More broadly, the literatures on behavioral portfolio theory and mental accounting have documented that households segregate goals into non-fungible buckets, but have offered little formal welfare accounting of that behavior. Our framework supplies exactly such an accounting. The underlying triangular reduction technique is not specific to the two goals studied here and should extend to other settings in which a household balances discretionary and non-discretionary claims on the same wealth.

A natural extension of our framework is to consider a finite sequence of fixed-deadline goals due at dates $0<T_1<\cdots<T_m$, with each amount revealed when due. One would work backward from $T_m$, using the random-goal continuation problem on paths where the random goal remains pending. At each $T_i$, conditional on affordability, the household would weigh the reward from funding the current goal against the change in continuation value associated with the payment. That continuation value would reflect the random-deadline goal if it remains unresolved and every fixed-deadline goal due after $T_i$.

Between successive fixed deadlines, the next funding decision would provide the anticipated terminal continuation value. If the random goal arrived during such an interval, its resolution would instead lead to a companion continuation problem with only the remaining fixed-deadline goals. The recursion would therefore distinguish both the next fixed deadline and whether the random goal remains pending. Economically, a current payment would be evaluated by its effect on the resources available for later goals, including the precautionary buffer while the random goal remains unresolved.

We expect this backward architecture to extend the present approach, but a rigorous treatment would require substantially more analysis including formulating and pasting admissible policies across repeated revelation and payment times, studying the measurability and regularity of the recursively generated continuation operators, and justifying dynamic programming and viscosity comparison on every intervening interval. These questions lie beyond the two-goal focus of the present paper and are not pursued here.

\textbf{Acknowledgments:} SC gratefully acknowledges support from a Columbia University CDFT Research Grant and an NSERC Postdoctoral Fellowship (PDF-599675-2025). AC and AP acknowledge support from  the University Blockchain Research Initiative (UBRI), sponsored by Ripple Labs, Inc. AC also acknowledges partial support from FiTek.

\clearpage
\appendix

\section*{Appendix}

\noindent This online appendix collects the analytical benchmarks, proofs, additional sensitivity analyses, calculations for the three-period illustration, and asymptotic comparisons that supplement the main paper. Section~\ref{app:single-goal-benchmarks} records the deterministic-amount single-goal benchmarks. Sections~\ref{app:viscosity} and~\ref{app:reduction} provide the viscosity and reduction proofs. Section~\ref{app:senscalibr} collects the sensitivity and robustness results, and Section~\ref{app:three-period-details} gives the three-period calculations. Section~\ref{app:ec-asymptotics} reports additional asymptotic comparisons. All numbered items carrying an ``EC.'' prefix refer to this document while unprefixed references refer to the main paper.

\section{Deterministic Amount Single-Goal Benchmarks}
\label{app:single-goal-benchmarks}

This appendix records the benchmarks used in Section~\ref{sec:asymptotics} to validate the numerical
solutions. These benchmarks correspond to the case where the goals are deterministic and the controls are unconstrained. They fall outside the standing assumptions of Section~\ref{sec:gen.model}, which require a compact set of controls and rewards with continuous distributions, but have explicit solutions which provide intuition for the problem structure.

\subsection{Deterministic Deadline Goal}
\label{sec:FD-det}

Fix $g>0$ and suppose $G\equiv g$. The unconstrained fixed deadline benchmark is
\[
V^{D}_{\mathrm{Browne}}(t,w;g):=
\sup_{\piv}
\mathbb E\!\left[
\mathbf 1_{\{\widetilde W_T^{t,w;\piv}\ge g\}}
\right],
\qquad (t,w)\in[0,T)\times[0,\infty),
\]
where the supremum is over admissible $\mathbb F$-predictable portfolio rules without leverage or short-selling constraints that are locally square-integrable on $[t,T)$. For this benchmark, terminal wealth is defined by the finite nonnegative limit $\widetilde W_T:=\lim_{s\uparrow T}\widetilde W_s$. Note that if $w\ge g$, the household can invest only in the money-market account and lock in success, so $V^{D}_{\mathrm{Browne}}(t,w;g)=1$, while
$V^{D}_{\mathrm{Browne}}(t,0;g)=0$.

Define $\gamma^2:=\thetavec^\top\Sigmav^{-1}\thetavec$ and assume
$\gamma^2>0$. The unconstrained solution in \citet{Browne1999} is
\begin{equation}
\label{eq:VD-det-closedform-n-pi}
V^{D}_{\mathrm{Browne}}(t,w;g)=
\Phi\!\left(
\Phi^{-1}\!\left(\frac{w}{g}\right)
+\gamma\sqrt{T-t}
\right),
\qquad 0<w<g,\qquad 0\le t<T,
\end{equation}
where $\Phi$ is the standard normal distribution function. An optimal feedback control is
\begin{equation}
\label{eq:pi-det-closedform-n}
\piv^\ast(t,w;g)= \frac{g}{w}
\frac{\phi\!\left(\Phi^{-1}\!\left(\frac{w}{g}\right)\right)}
{\gamma\sqrt{T-t}}
\Sigmav^{-1}\thetavec, \qquad 0<w<g,\qquad 0\le t<T,
\end{equation}
where $\phi$ is the standard normal density. This is locally square-integrable on $[t,T)$ and attains \eqref{eq:VD-det-closedform-n-pi}. Note that the portfolio fraction can become unbounded as $t\uparrow T$, while terminal wealth remains well defined. Moreover, $\lim_{t\uparrow T}V^{D}_{\mathrm{Browne}}(t,w;g)=w/g$ for $0<w<g$, so the unconstrained value does not, in fact, converge to the indicator payoff at the terminal boundary. The compact portfolio constraint in our main model rules out this gambling effect near the deadline and restores consistency with the boundary data in \eqref{eq:rgdt-bdy-data}.

\subsection{Random Deadline Goal}

Fix $c>0$ and suppose $R\equiv c$, so
$F_R(w)=\mathbf 1_{\{w\ge c\}}$. We consider a direct multi-asset extension of the single risky asset problem in \citet{BY16} with no consumption. The unconstrained value is
\[
V^{R}_{\rm BY}(w):=
\sup_{\piv}
\E\!\left[\mathbf 1_{\{\widetilde W_{\tau}^{0,w;\piv}\ge c\}}\right],
\qquad w\ge0,
\]
where the supremum is again over admissible $\mathbb F$-predictable portfolio rules without leverage or short-selling constraints that are locally square-integrable on $[0,\infty)$. On smooth parts of the candidate value, the formal HJB is
\begin{equation}
\label{eq:hjb-random}
\lambda V^{R}_{\rm BY}(w)=
\sup_{\piv\in\R^n}\left\{
w(V^{R}_{\rm BY})'(w)\piv^\top\thetavec
+\frac12w^2(V^{R}_{\rm BY})''(w)\piv^\top\Sigmav\piv
+\lambda F_R(w) \right\}.
\end{equation}
On $0<w<c$, a power candidate gives
\[
V^{R}_{\rm BY}(w)=\left(\frac{w}{c}\right)^\kappa,
\qquad
\piv_{\rm BY}^*=\frac{1}{1-\kappa}\Sigmav^{-1}\thetavec,
\qquad
\kappa=\frac{\lambda}{\lambda+\tfrac12\gamma^2}.
\]
The verification argument of \citet{BY16} yields
\[
V^{R}_{\rm BY}(w)=
\begin{cases}
(w/c)^\kappa, & 0\le w<c,\\[3pt]
1, & w\ge c.
\end{cases}
\]
If $\Pi_K$ contains both $0$ and $\piv_{\rm BY}^*$, the constrained value equals this benchmark. Indeed, one can use $\piv_{\rm BY}^*$ until wealth reaches $c$ and the zero strategy thereafter. This policy is both feasible for the constrained problem and attains the unconstrained value.

\section{Viscosity Proofs}\label{app:viscosity}

\subsection{Deterministic Deadline Goal}

As in Section~\ref{sec:gen.model}, fix the nonempty compact set $\Pi_K\subset\mathbb R^n$ containing $0$, and
write
\[
H(w,p,q)
:=
\sup_{\piv\in\Pi_K}
\left\{
wp\,\piv^\top\thetavec
+
\frac12 w^2q\,\piv^\top\Sigmav\piv
\right\}.
\]
We write the HJB equation appearing in \eqref{eq:hjb-rgdt-clipped} either as
\begin{equation}
\label{eq:A1-display-PDE}
 u_t+H(w,u_w,u_{ww})=0
\end{equation}
or, equivalently, in terms of $\mathcal F_D(t,w,z,a,p,q):=-a-H(w,p,q)$.

Here $z$ is the value variable.  It is included only to match standard viscosity notation since the fixed-deadline equation has no zero-order term.

A viscosity subsolution of $\mathcal F_D=0$ satisfies $\mathcal F_D\le0$.  Thus, if a smooth test function $\phi$ touches a subsolution from above at $(t,w)$, then
\[
-\phi_t(t,w)-H(w,\phi_w(t,w),\phi_{ww}(t,w))\le0,
\]
equivalently $\phi_t+H\ge0$.  A viscosity supersolution satisfies the reverse inequality.  This is the sign convention used throughout the proof below. We begin with a preliminary result about $H$.

\begin{lemma}\label{prop:FD-clipped-H}
For fixed $\Pi_K$, the Hamiltonian $H$ is finite and continuous on $\mathbb{R}^3$. Moreover, on every compact subset of $\mathbb R^3$, $H$ is Lipschitz. Finally, if $q_1\le q_2$, then $H(w,p,q_1)\le H(w,p,q_2)$ for all $(w,p)\in\mathbb R ^2$.
Consequently, $\mathcal F_D$ is continuous, proper, and degenerate elliptic.
\end{lemma}

\begin{proof}
For $\bm\pi\in\Pi_K$, define $\mathfrak g(w,p,q;\bm\pi)
:=
w\,p\,\bm\pi^\top\bm\theta
+
\tfrac12\,w^2\,q\,\bm\pi^\top\bm\Sigma\,\bm\pi$. 
Fix a compact set $C\subset \mathbb{R}^3$.
Because $\Pi_K$ is compact, there is a constant $M_K<\infty$ such that $\|\bm\pi\|\le M_K$ for every $\bm\pi\in\Pi_K$. Since $\mathfrak g$ is polynomial in $(w,p,q,\bm\pi)$, there exists $L_{C,K}<\infty$ with
\[
|\mathfrak g(\xi;\bm\pi)-\mathfrak g(\eta;\bm\pi)|
\le
L_{C,K}\,\|\xi-\eta\|
\qquad\text{for all }\xi,\eta\in C,\ \bm\pi\in\Pi_K.
\]
Taking the supremum over $\bm\pi\in\Pi_K$ yields $|H(\xi)-H(\eta)|
\le
L_{C,K}\,\|\xi-\eta\|$ for $\xi,\eta\in C$,
so $H$ is continuous and locally Lipschitz.

If $q_1\le q_2$, then for every $\bm\pi\in\Pi_K$, $\mathfrak g(w,p,q_1;\bm\pi)
\le
\mathfrak g(w,p,q_2;\bm\pi)$
because $\bm\pi^\top\bm\Sigma\,\bm\pi\ge0$. Taking suprema gives
$H(w,p,q_1)\le H(w,p,q_2)$.
Therefore $\mathcal F_D$ is nonincreasing in $q$, i.e.\ degenerate elliptic. It is proper because it is independent of the value variable (therefore nondecreasing), and it is continuous by the continuity of $H$.
\end{proof}

We note that since the control set $\Pi_K$ is compact, the drift and diffusion coefficients for the wealth $b(w,\bm\pi)=w\bm\pi^\top\bm\theta$,  $\sigma_W(w,\bm\pi)=w\bm\pi^\top\bm\sigma$
are globally Lipschitz in $w$ with constants uniform in
$\bm\pi\in\Pi_K$, and satisfy a uniform linear-growth bound.
Hence, for every admissible control, the wealth equation has a
unique strong solution and the problem is well-posed.

\begin{lemma}
\label{lem:FD-continuity}
The value function $V^D(t,w)$ is jointly continuous on $[0,T]\times [0,\infty)$.
\end{lemma}

\begin{proof}
The continuous function $F_G$ is constant outside $[0,b_G]$
and hence uniformly continuous on $\mathbb R$. Choose Lipschitz functions $F_G^m:\mathbb R\to[0,1]$ such that
\[
F_G^m(x)=0\quad\text{for }x\le0,
\qquad
F_G^m(x)=1\quad\text{for }x\ge b_G,
\qquad
\|F_G^m-F_G\|_\infty\to0.
\]
For instance, one may take the piecewise-linear approximations of $F_G$
on partitions of $[0,b_G]$ and extend them continuously to constant values outside of this interval.

For each $m$, define the auxiliary value function
\[
V_m(t,w)
:=
\sup_{\piv\in\mathcal A^{K,{\rm pre}}_{t,T}}
\mathbb E\!\left[
F_G^m\!\left(\widetilde W_T^{t,w;\piv}\right)
\right],
\qquad
(t,w)\in[0,T]\times\mathbb{R}.
\]
Observe that we have extended the domain to all of $\mathbb{R}$ but this does not pose a problem since the wealth SDE is still well-posed and clearly $V_m(t,w)=0$ if $w\leq 0$. For fixed $m$, this is a finite-horizon stochastic control problem in Mayer form with bounded control set $\Pi_K$, Lipschitz terminal payoff $F_G^m$, and coefficients that have linear growth in $(w,\pi)$ and are Lipschitz in the state uniformly over the controls.
Therefore \cite[Proposition~3.7]{Touzi2013} applies and gives $V_m\in C([0,T]\times\mathbb R)$. More precisely, $V_m$ is Lipschitz in the wealth variable, uniformly in time, and $1/2$-H\"older continuous in time on compact wealth intervals.

We now compare $V_m$ with $V^D$. For every
$(t,w)\in[0,T]\times[0,\infty)$,
\[
\begin{aligned}
\left|V_m(t,w)-V^D(t,w)\right|
&\le
\sup_{\piv\in\mathcal A^{K,{\rm pre}}_{t,T}}
\mathbb E\!\left[
\left|
F_G^m\!\left(\widetilde W_T^{t,w;\piv}\right)
-F_G\!\left(\widetilde W_T^{t,w;\piv}\right)
\right|
\right]\le
\|F_G^m-F_G\|_\infty .
\end{aligned}
\]
Thus $V_m\to V^D$ uniformly on $[0,T]\times[0,\infty)$. Since each $V_m$ is continuous, the uniform limit $V^D$ is continuous on $[0,T]\times[0,\infty)$.
\end{proof}

Our next result gives the usual comparison principle for viscosity solutions. For an open domain $Q\subset\mathbb R^2$, we write $USC(\overline Q)$ and $LSC(\overline Q)$ for the classes of real-valued functions on $\overline Q$ that are upper and lower semicontinuous relative to $\overline Q$, respectively.

\begin{lemma}[Comparison]
\label{lem:FD-comparison}
Let $Q:=(0,T)\times(0,b_G)$, and let $u\in USC(\overline Q)$ and
$v\in LSC(\overline Q)$ be bounded. Suppose that $u$ is a viscosity
subsolution and $v$ is a viscosity supersolution of
\[
-u_t-H(w,u_w,u_{ww})=0
\qquad\text{in }Q.
\]
Assume further that the parabolic boundary inequalities
\[
u(T,w)\le F_G(w)\le v(T,w),\qquad 0\le w\le b_G,
\]
and
\[
u(t,0)\le 0\le v(t,0),\qquad
u(t,b_G)\le 1\le v(t,b_G),
\qquad 0\le t\le T,
\]
hold. Then $u\le v$ on $(0,T]\times [0,b_G]$.
\end{lemma}

\begin{proof}
We reduce the statement to the Cauchy--Dirichlet comparison theorem of
\cite[Theorem~8.2]{crandall1992usersguideviscositysolutions}. Set $s:=T-t$, $h(w):=\frac{w}{b_G}$,
and define
\[
U(s,w):=u(T-s,w)-h(w),
\qquad
V(s,w):=v(T-s,w)-h(w),
\qquad (s,w)\in[0,T]\times[0,b_G].
\]
Then $U\in USC([0,T]\times[0,b_G])$ and
$V\in LSC([0,T]\times[0,b_G])$. The time reversal turns the terminal
condition into an initial condition, and the subtraction of $h$ (which is introduced to match the boundary conditions) makes it so
\[
U(0,w)\le F_G(w)-h(w)\le V(0,w),
\qquad 0\le w\le b_G,
\]
and $U(s,0)\le 0\le V(s,0)$, $U(s,b_G)\le 0\le V(s,b_G)$, for $0\le s\le T$.

We now identify the transformed equation. If $\psi$ is a smooth test function for $U$, then the corresponding test function for $u$ is $\phi(t,w):=\psi(T-t,w)+h(w)$.
Thus
\[
\phi_t(t,w)=-\psi_s(T-t,w),
\qquad
\phi_w(t,w)=\psi_w(T-t,w)+h'(w),
\qquad
\phi_{ww}(t,w)=\psi_{ww}(T-t,w)+h''(w).
\]
Since $h'(w)=1/b_G$ and $h''(w)=0$, the viscosity inequalities for
$u$ and $v$ transform into the viscosity inequalities for $U$ and $V$
associated with $z_s+F_h(w,z_w,z_{ww})=0$ in $(0,T)\times(0,b_G)$,
where
\[
F_h(w,p,q):=-H\left(w,p+\frac1{b_G},q\right).
\]
Hence $U$ is a viscosity subsolution and $V$ is a viscosity supersolution of this Cauchy--Dirichlet problem with continuous initial data $F_G-h$. This is precisely the problem in \cite[Equation (8.4)]{crandall1992usersguideviscositysolutions}.

It remains to verify the structural assumptions needed for
\cite[Theorem~8.2]{crandall1992usersguideviscositysolutions}. By Lemma \ref{prop:FD-clipped-H}, the operator $F_h$ is continuous, and it is proper because it is independent of the value variable and is degenerate elliptic: if $q_1\le q_2$, then $H(w,p,q_1)\le H(w,p,q_2)$, and therefore $F_h(w,p,q_1)\ge F_h(w,p,q_2)$.

Lastly, the structural condition \cite[Equation~(3.14)]{crandall1992usersguideviscositysolutions}
holds by a direct application of \cite[Example~3.6]{crandall1992usersguideviscositysolutions}.
Indeed, for each $\piv\in\Pi_{K}$ define
\[
b^{\piv}(w):=-w\piv^{\top}\thetavec,\qquad
Z^{\piv}(w):=\frac{w}{\sqrt{2}}\bigl(\piv^{\top}\Sigmav\piv\bigr)^{1/2},
\qquad
f^{\piv}(w):=\frac{w}{b_{G}}\piv^{\top}\thetavec .
\]
Then, identifying $Z^{\piv}(w)$ with a $1\times 1$ matrix,
\[
F_{h}(w,p,q)
=
\inf_{\piv\in\Pi_{K}}
\left\{
b^{\piv}(w)p
-\operatorname{tr}\!\left(Z^{\piv}(w)^{\top}Z^{\piv}(w)q\right)
-f^{\piv}(w)
\right\}.
\]
For each fixed $\piv$, this is a sum of the first-order, second-order, and
continuous spatial terms covered by Example~3.6. Moreover, because $\Pi_{K}$
is compact, the one-sided Lipschitz constants for $b^{\piv}$, the Lipschitz
constants for $Z^{\piv}$, and the moduli of continuity of $f^{\piv}$ can be chosen
uniformly over $\piv\in\Pi_{K}$. The final stability statement in
\cite[Example~3.6]{crandall1992usersguideviscositysolutions} therefore implies
that the infimum over $\piv\in\Pi_{K}$ satisfies the same structural condition.

The Cauchy--Dirichlet comparison theorem
\cite[Theorem~8.2]{crandall1992usersguideviscositysolutions} yields $U(s,w)\le V(s,w)$ for all $(s,w)\in[0,T)\times[0,b_G]$.
Since $U-V=u(T-s,w)-v(T-s,w)$, returning to the original variables gives
$u\le v$ on $(0,T]\times [0,b_G]$.
\end{proof}

With these preliminaries, we complete the proof of the fixed-deadline part of
Proposition~\ref{prop:single-goal-viscosity}.

\begin{proof}[Proof of the fixed-deadline part of Proposition~\ref{prop:single-goal-viscosity}]
We divide the proof into parts.

\emph{\underline{Step 1:} (Regularity and Boundary Values)} The value function is bounded because $0\le F_G\le1$ and it is continuous by Lemma~\ref{lem:FD-continuity}. The boundary values are immediate. If $w=0$, then wealth remains zero and $V^D(t,0)=F_G(0)=0$. If $w=b_G$, the household can choose $\bm\pi\equiv0$ and lock in terminal wealth $b_G$, so $V^D(t,b_G)=1$. At $t=T$, no time remains to trade, and therefore $V^D(T,w)=F_G(w)$.

\emph{\underline{Step 2:} (Viscosity Solution)} We note as in Lemma~\ref{lem:FD-continuity} that the definition of $V^D$ readily extends to $[0,T]\times \mathbb{R}$ with the same properties. Moreover, the present bounded-control Mayer problem falls within the stochastic-control framework of \cite[Section~7]{Touzi2013}. Since $V^D$ is bounded and $H$ is finite and continuous, we may apply \cite[Theorem~7.4]{Touzi2013} to obtain that $V^D$ is a viscosity solution of \eqref{eq:hjb-rgdt-clipped} on $[0,T)\times \mathbb{R}$ and hence also on $[0,T)\times(0,b_G)$.

\emph{\underline{Step 3:} (Uniqueness)} Let $u$ and $v$ be two bounded continuous viscosity solutions of the equation $-u_t-H(w,u_w,u_{ww})=0$ in  $(0,T)\times(0,b_G)$,
with the stated boundary data. Applying Lemma~\ref{lem:FD-comparison} to $u$ as subsolution and $v$ as supersolution gives $u\le v$. Reversing the roles of $u$ and $v$ gives $v\le u$. Hence $u=v$ on $(0,T]\times[0,b_G]$. This equality extends to $t=0$ by continuity.
\end{proof}

\subsection{Random Deadline Goal}

Define $\mathcal F_R(w,z,p,q) :=\lambda z-H(w,p,q)-\lambda F_R(w)$, so that \eqref{eq:hjb-random-clipped} is
\[
\mathcal F_R
\bigl(w,V^R(w),(V^R)'(w),(V^R)''(w)\bigr)=0.
\]
Because the same clipped Hamiltonian $H$ appears here as in the fixed-deadline problem, Lemma~\ref{prop:FD-clipped-H} implies that $H$ is finite and continuous, locally Lipschitz on compact sets, and
nondecreasing in its second-derivative argument. Since $F_R$ is
continuous, $\mathcal F_R$ is continuous. Moreover,
$\mathcal F_R$ is strictly increasing in the value variable $z$,
with slope $\lambda>0$, and is nonincreasing in the second-derivative
argument $q$. Hence $\mathcal F_R$ is proper and degenerate elliptic.

Our first supporting result proves the continuity of the value function $V^R$. 

\begin{lemma}
\label{lem:RD-continuity}
The value function $V^R$ is continuous on $[0,\infty)$.
\end{lemma}

\begin{proof}
Since $F_R$ is continuous and
constant outside $[0,b_R]$, it is uniformly continuous on $\mathbb R$.
Let $\omega_R$ be a modulus of continuity for $F_R$. Fix $N>0$ and define the truncated value
\[
V_N^R(w)
:=
\sup_{\piv\in\mathcal A^{K,{\rm pre}}_{0,\infty}}
\mathbb E\left[
\int_0^N
\lambda e^{-\lambda s}
F_R\!\left(\widetilde W_s^{0,w;\piv}\right)\,ds
\right].
\]
The tail is uniformly bounded by
\[
0\le V^R(w)-V_N^R(w)
\le
\int_N^\infty \lambda e^{-\lambda s}\,ds
=
e^{-\lambda N}.
\]
Hence, for $w,w'\geq 0$, $|V^R(w)-V^R(w')|
\le
|V_N^R(w)-V_N^R(w')|
+
2e^{-\lambda N}$.

We now prove continuity of $V_N^R$. Couple the wealth processes
started from $w$ and $w'$ using the same Brownian motion and the same
control $\piv$. Since $\Pi_K$ is bounded, the coefficients $b(x,\piv)=x\piv^\top\thetavec$, $\sigma_W(x,\piv)=x\piv^\top\sigmavec$
are globally Lipschitz in $x$, uniformly over $\piv\in\Pi_K$, and have
linear growth. Therefore the finite-horizon flow estimate of
\cite[Theorem~2.4]{Touzi2013} gives, for each $N<\infty$, a constant
$C_{N,K}<\infty$ such that
\[
\sup_{\piv\in\mathcal A^{K,{\rm pre}}_{0,\infty}}
\mathbb E\left[
\sup_{0\le s\le N}
\left|
\widetilde W_s^{0,w;\piv}
-
\widetilde W_s^{0,w';\piv}
\right|^2
\right]
\le
C_{N,K}|w-w'|^2 .
\]
Consequently, for every $\eta>0$,
\[
\begin{aligned}
|V_N^R(w)-V_N^R(w')|
&\le
\sup_{\piv}
\mathbb E\left[
\int_0^N
\lambda e^{-\lambda s}
\left|
F_R\!\left(\widetilde W_s^{0,w;\piv}\right)
-
F_R\!\left(\widetilde W_s^{0,w';\piv}\right)
\right|ds
\right] \\
&\le
\omega_R(\eta)
+
2\sup_{\piv}
\mathbb P\left(
\sup_{0\le s\le N}
\left|
\widetilde W_s^{0,w;\piv}
-
\widetilde W_s^{0,w';\piv}
\right|
>\eta
\right) \\
&\le
\omega_R(\eta)
+
\frac{2C_{N,K}}{\eta^2}|w-w'|^2 .
\end{aligned}
\]
Letting $w'\to w$ gives $\limsup_{w'\to w}|V_N^R(w)-V_N^R(w')|
\le
\omega_R(\eta)$.
Then $\eta\downarrow0$ proves continuity of $V_N^R$. Finally, $\limsup_{w'\to w}|V^R(w)-V^R(w')|
\le
2e^{-\lambda N}$,
and letting $N\to\infty$ proves continuity of $V^R$.
\end{proof}

Next, we prove the comparison principle for viscosity solutions of this elliptic problem.

\begin{lemma}[Comparison]
\label{lem:RD-comparison}
Let $\Omega_R:=(0,b_R)$. Suppose $u\in USC(\overline\Omega_R)$ and
$v\in LSC(\overline\Omega_R)$ are bounded, $u$ is a viscosity subsolution
and $v$ is a viscosity supersolution of
\[
\mathcal F_R(w,u,u',u'')=0
\qquad\text{in }\Omega_R,
\]
and $u(0)\le v(0)$, $u(b_R)\le v(b_R)$.
Then $u\le v$ on $\overline\Omega_R$.
\end{lemma}

\begin{proof}
We verify the hypotheses of \cite[Theorem~3.3]{crandall1992usersguideviscositysolutions}.
The domain $\Omega_R$ is bounded, and as noted above $\mathcal F_R$ is continuous and
proper. Moreover, for $r>s$,
\[
\mathcal F_R(w,r,p,q)-\mathcal F_R(w,s,p,q)
=
\lambda(r-s),
\]
so the condition in \cite[Equation~(3.13)]{crandall1992usersguideviscositysolutions}
holds with $\gamma=\lambda$.

It remains only to verify the structural condition
\cite[Equation~(3.14)]{crandall1992usersguideviscositysolutions}. We write
\[
\mathcal F_R(w,r,p,q)
=
\lambda r
+
\inf_{\piv\in\Pi_K}
\left\{
-wp\,\piv^{\top}\thetavec
-
\frac12 w^2q\,\piv^{\top}\Sigmav\piv
-
\lambda F_R(w)
\right\}.
\]
For each $\piv\in\Pi_K$, set $b^{\piv}(w):=-w\piv^{\top}\thetavec$ and $Z^{\piv}(w):=\frac{w}{\sqrt{2}}
\bigl(\piv^{\top}\Sigmav\piv\bigr)^{1/2}$. Then the controlled part has the form
\[
b^{\piv}(w)p
-
\operatorname{tr}
\left(
Z^{\piv}(w)^{\top}Z^{\piv}(w)q
\right)
-
\lambda F_R(w).
\]
The first-order coefficients $b^{\piv}$ are one-sided Lipschitz in $w$ uniformly
over $\piv\in\Pi_K$, the diffusion factors $Z^{\piv}$ are Lipschitz in $w$
uniformly over $\piv\in\Pi_K$, and the spatial term $\lambda F_R(w)$ is
continuous on the compact interval $[0,b_R]$. Therefore
\cite[Example~3.6]{crandall1992usersguideviscositysolutions} gives the
structural condition with a common modulus, and this condition is preserved
under the infimum over the compact control set $\Pi_K$.

The comparison theorem yields $u\le v$ in $\Omega_R$. The boundary inequality gives the result on $\overline\Omega_R$.
\end{proof}

We now turn to the proof of the random-deadline part of
Proposition~\ref{prop:single-goal-viscosity}.

\begin{proof}[Proof of the random-deadline part of Proposition~\ref{prop:single-goal-viscosity}]
We once again divide the proof into parts.

\emph{\underline{Step 1:} (Regularity and Boundary Values)} The value is bounded between $0$ and $1$ since $0\leq F_R\leq 1$ and continuity follows from Lemma~\ref{lem:RD-continuity}. The boundary values are also immediate as before with $V^D$. If $w=0$, wealth remains zero and
$F_R(0)=0$, so $V^R(0)=0$. If $w\ge b_R$, the household can choose
$\piv\equiv0$, keeping wealth at least $b_R$ forever, and hence
$V^R(w)=1$.

\emph{\underline{Step 2:} (Viscosity Solution)} To see that the value function $V^R$ is a viscosity solution of
\[
\lambda u(w)
-
H\bigl(w,u'(w),u''(w)\bigr)
-
\lambda F_R(w)
=0
\qquad\text{on }(0,b_R)
\]
we view the random-deadline problem as an infinite-horizon discounted
stochastic control problem, in the
framework of \cite[Section~4.3]{Pham2009}. To match the Euclidean-state formulation in \cite[Section~4.3]{Pham2009},
we use the harmless extension of the coefficients and payoff to all
$w\in\mathbb R$. The state dynamics are
autonomous:
\[
d\widetilde W_t
=
\widetilde W_t
\left(
\piv_t^{\top}\thetavec\,dt
+
\piv_t^{\top}\sigmavec\,dB_t
\right),
\qquad
\piv_t\in\Pi_K
\]
and the control set is compact. Starting from $w>0$, the solution remains strictly positive. The coefficients are continuous, globally
Lipschitz in the state uniformly over $\piv\in\Pi_K$, and satisfy a linear-growth bound. The running reward is $f(w,\piv):=\lambda F_R(w)$, which is bounded and continuous, and the discount rate is $\lambda>0$.

In the notation of \cite{Pham2009}, the Hamiltonian is
\[
\widehat H(w,p,q)
:=
\sup_{\piv\in\Pi_K}
\left\{
wp\,\piv^{\top}\thetavec
+
\frac12 w^2q\,\piv^{\top}\Sigmav\piv
+
\lambda F_R(w)
\right\}
=
H(w,p,q)+\lambda F_R(w).
\]
Since $\Pi_K$ is compact, $\widehat H$ is finite and continuous. Hence
we are in the regular compact-control case\footnote{We note that the infinite-horizon viscosity result is stated under the condition that (in our notation) $\lambda$ is ``sufficiently large.''
That condition is used in \cite[Proposition 4.3.1]{Pham2009} to ensure that the infinite-horizon payoff is
well-defined in the unbounded reward setting, and to ensure that constant controls are admissible (see the invocation of \cite[Remark 3.2.2]{Pham2009} in the proof). In the present problem these requirements hold for
every $\lambda>0$, because the running reward is uniformly bounded over the admissible controls $\piv$ and so we can dispense with this requirement.} of
\cite[Theorem~4.3.1 and Remark~4.3.4]{Pham2009} and the variational
inequality reduces to the ordinary HJB equation. The argument in \cite[Section~4.3 and Theorem~4.3.1]{Pham2009} therefore applies in
the present setting and gives that
$V^R$ is a viscosity solution of
\[
\lambda V^R(w)
-
\widehat H\bigl(w,(V^R)'(w),(V^R)''(w)\bigr)
=0
\]
on $\mathbb{R}$ and hence, also on $(0,b_R)$. Substituting the expression for
$\widehat H$ gives
\[
\lambda V^R(w)
-
H\bigl(w,(V^R)'(w),(V^R)''(w)\bigr)
-
\lambda F_R(w)
=0
\]
in the viscosity sense on $(0,b_R)$ as required.

\emph{\underline{Step 3:} (Uniqueness)} Let $u$ and $v$ be two bounded continuous
viscosity solutions of \eqref{eq:hjb-random-clipped} on $(0,b_R)$ with
boundary values $u(0)=v(0)=0$ and $u(b_R)=v(b_R)=1$. Applying
Lemma~\ref{lem:RD-comparison} first to $u$ as subsolution and $v$ as
supersolution gives $u\le v$. Reversing the roles gives $v\le u$. Hence
$u=v$ on $[0,b_R]$.
\end{proof}

\subsection{Joint-Goal Problem}
\label{app:dual-viscosity-proof}

The two funding rules give the same HJB equation \eqref{eq:dual-HJB-clipped} and differ only in the terminal condition. As a result, we prove our main characterization result in terms of a general continuous terminal payoff $\Psi$ and apply it with $\Psi=\mathcal T$ and $\Psi=\mathcal T_{\mathrm{commit}}$ to prove Theorem \ref{thm:dual-HJB-opt} and Proposition \ref{prop:committed-characterization}.

Recall $b=b_R+b_G$.  Extend $\mathcal J$ to all real wealth by setting
$\mathcal J(t,w)=0$ for $w\le0$. Its definition gives
$\mathcal J(t,w)=1$ for $w\ge b$.  Extend
$\mathcal T$ and $\mathcal T_{\mathrm{commit}}$ in the same
way, by zero on $(-\infty,0]$ and by one on $[b,\infty)$.

\begin{lemma}
\label{lem:dual-operators-continuity}
The running operator $\mathcal J$ belongs to
$C([0,T]\times\mathbb R)$, and both terminal operators
$\mathcal T$ and $\mathcal T_{\mathrm{commit}}$ belong to $C(\mathbb R)$.  Moreover, all three functions take values in $[0,1]$ and are constant outside the compact wealth interval $[0,b]$.
\end{lemma}

\begin{proof}
Boundedness and the stated constant extensions follow from
$0\le V^D,V^R\le1$, $\alpha_R+\alpha_D=1$, and the facts that
$R\le b_R$ and $G\le b_G$ almost surely.

We first prove continuity of $\mathcal J$ on $[0,T]\times[0,b]$. Let $(t_m,w_m)\to(t,w)$ in $[0,T]\times[0,b]$. For every realization $r\ne w$,
\[
\mathbf 1_{\{w_m\ge r\}}\longrightarrow\mathbf 1_{\{w\ge r\}},
\qquad
w_m-r\mathbf 1_{\{w_m\ge r\}}
\longrightarrow
w-r\mathbf 1_{\{w\ge r\}}.
\]
The continuity of $F_R$ gives $\mathbb P(R=w)=0$, while
Proposition~\ref{prop:single-goal-viscosity} gives the joint continuity of
$V^D$.  Dominated convergence in
\eqref{eq:rigorous-JK-def} therefore proves
$\mathcal J(t_m,w_m)\to\mathcal J(t,w)$.

The operator $\mathcal T_{\mathrm{commit}}$ is handled identically, using the
continuity of $F_G$ and Proposition~\ref{prop:single-goal-viscosity}.  For $\mathcal T$, for every $g\ne w$ the feasible set $\Theta(w_m,g)$ is eventually equal to $\Theta(w,g)$, and the two candidate payoffs in \eqref{eq:rigorous-Topt-def} converge by continuity of $V^R$.  The maximum therefore converges pointwise for
$G\ne w$.  Since, once again, $\mathbb P(G=w)=0$ and the maximum is bounded, dominated convergence proves continuity of $\mathcal T$.
\end{proof}

We now introduce a general class of terminal data that
includes $\mathcal T,\mathcal T_{\mathrm{commit}}$ by Lemma~\ref{lem:dual-operators-continuity}.

\begin{definition}[Admissible terminal data]
\label{def:admissible-terminal-data}
Let $\mathfrak T$ be the class of continuous functions
$\Psi:\mathbb R\to[0,1]$ satisfying $\Psi(w)=0$ for $w\le0$,
$\Psi(w)=1$ for $w\ge b$.
\end{definition}

For $\Psi\in\mathfrak T$, define
\begin{equation}
\label{eq:generic-dual-reduced-value}
\widetilde V_{\Psi}(t,w)
:=
\sup_{\piv\in\mathcal A^{K,{\rm pre}}_{t,T}}
\mathbb E_{t,w}\!\left[
\int_t^T\lambda e^{-\lambda(s-t)}
\mathcal J\!\left(s,\widetilde W_s^{t,w;\piv}\right)\,ds
+
e^{-\lambda(T-t)}
\Psi\!\left(\widetilde W_T^{t,w;\piv}\right)
\right],
\end{equation}
for $(t,w)\in[0,T]\times\mathbb R$, with the convention
$\widetilde V_{\Psi}(T,w)=\Psi(w)$.  The multiplicative
wealth equation is well posed for arbitrary real initial wealth.  If
$w\le0$, its solution remains nonpositive, so the extensions above give
$\widetilde V_{\Psi}(t,w)=0$. 

\begin{lemma}
\label{lem:dual-value-continuity}
For every $\Psi\in\mathfrak T$, the value function
$\widetilde V_{\Psi}$ is continuous and $[0,1]$-valued on
$[0,T]\times\mathbb R$.
\end{lemma}

\begin{proof}
The range follows from
$0\le\mathcal J,\Psi\le1$ and
\[
\int_t^T\lambda e^{-\lambda(s-t)}\,ds
+e^{-\lambda(T-t)}=1.
\]
By Lemma~\ref{lem:dual-operators-continuity} and
Definition~\ref{def:admissible-terminal-data},
$\mathcal J$ and $\Psi$ are uniformly continuous after their
constant extensions.  Choose bounded Lipschitz functions
$\mathcal J_m:[0,T]\times\mathbb R\to[0,1]$, 
$\Psi_m:\mathbb R\to[0,1]$
such that
\[
\varepsilon_m:=
\|\mathcal J_m-\mathcal J\|_\infty
+
\|\Psi_m-\Psi\|_\infty
\longrightarrow0.
\]
Let $V_m$ be the value in
\eqref{eq:generic-dual-reduced-value} with
$(\mathcal J_m,\Psi_m)$ in place of
$(\mathcal J,\Psi)$.  Direct comparison of the payoffs gives
\begin{equation}
\label{eq:generic-uniform-approximation}
\|V_m-\widetilde V_{\Psi}\|_\infty\le\varepsilon_m.
\end{equation}

It remains to show that each $V_m$ is continuous as the uniform convergence in \eqref{eq:generic-uniform-approximation} would imply the continuity of $\widetilde V_{\Psi}$.  The coefficients $\beta(w,\piv)=w\piv^\top\thetavec$, 
$\sigma_W(w,\piv)=w\piv^\top\sigmavec$
are Lipschitz in $w$, uniformly over the compact set $\Pi_K$, and have
uniform linear growth.  To put the running-payoff problem in the Mayer
form used in \cite[Proposition~3.7]{Touzi2013} we may use the standard augmentation described in \cite[Remark 3.2(iii)]{Touzi2013}. Take a bounded Lipschitz function $\chi$ that agrees with the identity on
$[e^{-\lambda T},1]$ and augment the state by
\[
dY_s=\chi(Z_s)\lambda\mathcal J_m(s,\widetilde W_s)\,ds,\quad Y_t=0,
\qquad
dZ_s=-\lambda Z_s\,ds,\quad Z_t=1.
\]
Then $Z_s=e^{-\lambda(s-t)}$ and the payoff is $Y_T+\chi(Z_T)\Psi_m(\widetilde W_T)$,
a Lipschitz function of the terminal augmented state.  The augmented
coefficients are Lipschitz in the state uniformly over $\Pi_K$ and have
linear growth.  Hence
\cite[Proposition~3.7]{Touzi2013} gives
$V_m\in C([0,T]\times\mathbb R)$.
\end{proof}

For $\Psi\in\mathfrak T$, consider
\begin{equation}
\label{eq:generic-dual-HJB}
\partial_t u
+H(w,u_w,u_{ww})
+\lambda\bigl(\mathcal J(t,w)-u\bigr)=0
\quad\text{on }\Omega_{RD}:=(0,T)\times(0,b),
\end{equation}
with boundary data
\begin{equation}
\label{eq:generic-dual-HJB-bc}
u(T,w)=\Psi(w),\quad 0\le w\le b,
\qquad
u(t,0)=0,\quad u(t,b)=1,\quad 0\le t\le T.
\end{equation}

\begin{lemma}[Comparison]
\label{lem:dual-comparison}
Fix $\Psi\in\mathfrak T$. Let $u\in USC(\overline{\Omega}_{RD})$, 
$v\in LSC(\overline{\Omega}_{RD})$ be bounded and suppose that $u$ is a
viscosity subsolution and $v$ is a viscosity supersolution of
\eqref{eq:generic-dual-HJB} on $\Omega_{RD}$. Assume further that
\[
u(T,w)\le\Psi(w)\le v(T,w),\qquad 0\le w\le b,
\]
and $u(t,0)\le0\le v(t,0)$, $u(t,b)\le1\le v(t,b)$ for $0\le t\le T$.
Then $u\le v$ on $(0,T]\times[0,b]$. If, in addition, $u$ and $v$ are continuous
on $\overline{\Omega}_{RD}$, the inequality also holds at $t=0$.
\end{lemma}

\begin{proof}
Set $s=T-t$, $h(w)=w/b$, and $U(s,w):=u(T-s,w)-h(w)$, $V(s,w):=v(T-s,w)-h(w)$.
As before, $h$ is used to make the lateral boundary data homogeneous:
\[
U(s,0)\le0\le V(s,0),\qquad
U(s,b)\le0\le V(s,b),\qquad 0\le s\le T.
\]
The terminal inequalities become $U(0,w)\le\Psi(w)-h(w)\le V(0,w)$ for $0\le w\le b$ and the transformed operator is
\[
\widehat F(s,w,r,p,q)
:=
\lambda\bigl(r+h(w)\bigr)
-H\!\left(w,p+\frac1b,q\right)
-\lambda\mathcal J(T-s,w).
\]
Thus $U$ and $V$ are respectively a viscosity subsolution and
supersolution of $z_s+\widehat F(s,w,z,z_w,z_{ww})=0$ on
$(0,T)\times(0,b)$, with the required initial and lateral boundary
ordering.

The operator $\widehat F$ is continuous by
Lemma~\ref{prop:FD-clipped-H} and
Lemma~\ref{lem:dual-operators-continuity}.  It is proper, with strictly positive
slope $\lambda$ in the value variable, and is degenerate elliptic
because $H$ is nondecreasing in $q$.  To verify
\cite[Equation~(3.14)]{crandall1992usersguideviscositysolutions}, write
\[
\widehat F(s,w,r,p,q)
=
\lambda r
+
\inf_{\piv\in\Pi_K}
\left\{
-\beta^{\piv}(w)p
-\bigl(Z^{\piv}(w)\bigr)^2q
-f^{\piv}(w)
\right\}
+
\lambda h(w)-\lambda\mathcal J(T-s,w),
\]
where
\[
\beta^{\piv}(w)=w\piv^\top\thetavec,\qquad
Z^{\piv}(w)=\frac{w}{\sqrt2}
\bigl(\piv^\top\Sigmav\piv\bigr)^{1/2},\qquad
f^{\piv}(w)=\frac{w}{b}\piv^\top\thetavec.
\]
Because $\Pi_K$ is compact, the Lipschitz constants of these
coefficients can be chosen uniformly in $\piv$.  The remaining spatial
term is uniformly continuous on $[0,T]\times[0,b]$.  Consequently
\cite[Example~3.6]{crandall1992usersguideviscositysolutions} supplies
the structural condition with a common modulus, and the infimum over
$\Pi_K$ preserves it.  The comparison theorem
\cite[Theorem~8.2]{crandall1992usersguideviscositysolutions} gives
$U\le V$.  Returning to $(t,w)$ proves the claim and the postulated continuity extends it
to $t=0$.
\end{proof}

\begin{theorem}[Viscosity characterization]
\label{thm:generic-dual-HJB}
Let $\Psi\in\mathfrak T$.  Then
$\widetilde V_{\Psi}$ is bounded and continuous on
$[0,T]\times[0,\infty)$, satisfies $\widetilde V_{\Psi}(t,0)=0$, $\widetilde V_{\Psi}(t,w)=1$ for $w\ge b$, $\widetilde V_{\Psi}(T,\cdot)=\Psi$,
and its restriction to $[0,T]\times[0,b]$ is the unique bounded
continuous viscosity solution of
\eqref{eq:generic-dual-HJB}--\eqref{eq:generic-dual-HJB-bc}.
\end{theorem}

\begin{proof}
Continuity follows from Lemma~\ref{lem:dual-value-continuity}.  At $w=0$ wealth
remains zero.  At $w\ge b$, the admissible zero portfolio keeps wealth constant, so $\mathcal J=1$ and $\Psi=1$ along the path and the payoff is one.  These observations give the boundary values.

The problem \eqref{eq:generic-dual-reduced-value} is a finite-horizon
stochastic control problem with compact control set, state
coefficients Lipschitz uniformly over $\Pi_K$, bounded continuous
running reward $\lambda\mathcal J$, discount rate $\lambda$, and
bounded continuous terminal reward $\Psi$.
The associated Hamiltonian is
\[
\mathfrak H(t,w,r,p,q)
:=
\sup_{\piv\in\Pi_K}
\left\{
-\lambda r
+wp\,\piv^\top\thetavec
+\frac12w^2q\,\piv^\top\Sigmav\piv
+\lambda\mathcal J(t,w)
\right\}.
\]
It is finite and continuous.  Since the value function is continuous and locally bounded, \cite[Theorem~7.4]{Touzi2013} gives the viscosity property, and substituting the expression for $\mathfrak H$ gives
\eqref{eq:generic-dual-HJB}.  Uniqueness follows by applying
Lemma~\ref{lem:dual-comparison} in both directions.
\end{proof}

\begin{proof}[Proof of Theorem~\ref{thm:dual-HJB-opt}]
By Proposition~\ref{prop:reduction-opt}, the value function agrees
with $\widetilde V_{\mathcal T}$ for $t<T$.
Lemma~\ref{lem:dual-operators-continuity} gives
$\mathcal T\in\mathfrak T$, so the result follows from
Theorem~\ref{thm:generic-dual-HJB}.
\end{proof}

For the committed benchmark, the terminal data are
\begin{equation}
\label{eq:dual-HJB-bc-forced}
u(T,w)=\mathcal T_{\mathrm{commit}}(w),\quad 0\le w\le b,
\qquad
u(t,0)=0,\quad u(t,b)=1,\quad 0\le t\le T.
\end{equation}
The main representation result for this problem (Proposition~\ref{prop:committed-characterization}) follows in the same way, given the reduction in Proposition~\ref{prop:reduction}, which is proved in Section~\ref{app:reduction}.

\begin{proof}[Proof of Proposition~\ref{prop:committed-characterization}]
Proposition~\ref{prop:reduction} gives the stated reduced representation
and identifies the committed value with
$\widetilde V_{\mathcal T_{\mathrm{commit}}}$ for $t<T$.
By Lemma~\ref{lem:dual-operators-continuity},
$\mathcal T_{\mathrm{commit}}\in\mathfrak T$.
The remaining assertions follow from
Theorem~\ref{thm:generic-dual-HJB}.
\end{proof}

\section{Proofs of Propositions~\ref{prop:reduction-opt}
and~\ref{prop:reduction}}\label{app:reduction}

We first prove the reduction with an arbitrary feasible
$\mathcal G_T$-measurable terminal decision. The committed
case follows by requiring the fixed-deadline goal to be
funded whenever affordable.

\subsection{Conventions and Preliminaries}
In order to make the conditioning arguments below formal, we take the model on a
canonical product space. Consider, for instance, 
$C([0,\infty);\R^d)\times[0,\infty)\times[0,b_R]\times[0,b_G]$, carrying the
Brownian motion, the exponential clock, and the two goal amounts. We then pass to the usual completed, right-continuous filtrations. When regular conditional probabilities appear below, they are taken on the raw countably generated sigma-fields. The resulting conditional identities are then understood almost surely under the completed filtrations. Predictable controls are likewise identified up to $ds\otimes d\PP$-null sets with raw predictable versions. Importantly, conclusions about the value function remain wholly unchanged by this choice of convention.

For completeness, the full-information filtration used in the main model is
the usual augmentation of
\[
\mathcal G_t^0
:=
\bigcap_{u>t}
\left(
\mathcal F_u
\vee\sigma(\tau\wedge u)
\vee\sigma\!\left(R\mathbf 1_{\{\tau\le u\}}\right)
\vee\sigma\!\left(G\mathbf 1_{\{T\le u\}}\right)
\right),
\qquad t\ge0.
\]
Then $\tau$ is a $\mathbb G$-stopping time, $R$ is
$\mathcal G_\tau$-measurable, and $G$ is $\mathcal G_T$-measurable.
Independence of $(\tau,R,G)$ from $\mathcal F_\infty$ implies that $B$
remains a $\mathbb G$-Brownian motion.

We will always work \emph{conditionally on $\{\tau>t\}$} and write
$\E_{t,w}:=\E[\,\cdot\mid\tau>t\,]$ for the conditional expectation with pre-payment
wealth started at $w$ at time $t$. On $\{\tau>t\}$ set $\tau_t:=\tau$. By the memorylessness of the exponential distribution and the independence of $\tau$ from $\F_\infty\vee\sigma(R,G)$, under $\PP(\cdot\mid\tau>t)$ the
residual time $\tau_t-t$ is $\mathrm{Exp}(\lambda)$ and is independent of
$\F_\infty\vee\sigma(R,G)$. In particular $\PP(\tau_t=T\mid\tau>t)=0$ and
$\PP(\tau_t>T\mid\tau>t)=e^{-\lambda(T-t)}$. For convenience we adopt the same notation for the all-or-nothing post-payment map as in the motivating section: $\Gamma_a(x):=x-a\,\1_{\{x\ge a\}}$ for $a,x\ge0$.

Recall that for $\piv\in\cA^{K,\mathrm{pre}}_{t,u}$ (or $\piv\in\cA^{K}_{t,u}$) the pre-payment
wealth $\widetilde W^{t,w;\piv}$ is the strong solution of
\begin{equation}\label{eq:prepay}
d\widetilde W_s^{t,w;\piv}
=\widetilde W_s^{t,w;\piv}\bigl(\piv_s^\top\thetavec\,ds+\piv_s^\top\sigmavec\,dB_s\bigr),
\qquad \widetilde W_t^{t,w;\piv}=w,\quad s\in[t,u].
\end{equation}
For $\piv\in\mathcal A^K_{t,\infty}$, let $Y^{t,w;\piv}$ be the
pre-decision process obtained by starting from $w$ at time $t$, following
\eqref{eq:prepay} until $\tau_t\wedge T$, applying the update
$x\mapsto\Gamma_R(x)$ if $\tau_t<T$, and then following
\eqref{eq:prepay} until $T$. Its left limit $Y_{T-}^{t,w;\piv}$ is
determined by $\piv$ alone.
An admissible pair at $(t,w)$ consists of
$\piv\in\mathcal A^K_{t,\infty}$ and a $\{0,1\}$-valued,
$\mathcal G_T$-measurable variable $d_T$ satisfying
$d_T\in\Theta(Y_{T-}^{t,w;\piv},G)$ almost surely.
For such a pair, the post-payment wealth
$\widehat W^{t,w;\piv,d_T}$ is the unique c\`adl\`ag
$\filG$-adapted solution of
\begin{equation}\label{eq:postpay}
\widehat W_s^{t,w;\piv,d_T}
=w
+\int_t^s\widehat W_{r-}^{t,w;\piv,d_T}\,\piv_r^\top\thetavec\,dr
+\int_t^s\widehat W_{r-}^{t,w;\piv,d_T}\,\piv_r^\top\sigmavec\,dB_r
-R\,\1_{\{\tau_t\le s,\ \widehat W_{\tau_t-}^{t,w;\piv,d_T}\ge R\}}
-d_TG\,\1_{\{T\le s\}},
\end{equation}
$s\ge t$. Hence
$\widehat W_{T-}^{t,w;\piv,d_T}=Y_{T-}^{t,w;\piv}$ does not depend on the choice of
$d_T$. Well-posedness and nonnegativity follow by solving the diffusion on
the stochastic intervals separated by $\tau_t$ and $T$ and applying the two
feasible jumps.

We remark that for $s\in[0,T]$ and $x\ge0$, we may write
\begin{align}
\mathcal J(s,x)
&:=\E_R\!\left[\alpha_R\1_{\{x\ge R\}}+\alpha_D V^D\!\left(s,\Gamma_R(x)\right)\right]
\label{eq:JK}\\
&\phantom{:}=\alpha_R F_R(x)+\alpha_D\!\left[(1-F_R(x))V^D(s,x)
+\int_0^x V^D(s,x-r)\,dF_R(r)\right],\nonumber\\[2pt]
\mathcal T(x)
&:=\E_G\!\left[
\max_{d\in\Theta(x,G)}
\{\alpha_Dd+\alpha_RV^R(x-dG)\}\right]
\phantom{:}=\E_G[M(x,G)],\label{eq:TK}
\end{align}
where $\E_R,\E_G$ average over the laws of $R$ and $G$ only. Moreover, the following holds for the funding rule $D^*$.

\begin{lemma}
\label{lem:optional-terminal-selector}
The function $D^*$ in \eqref{eq:rigorous-terminal-selector} is Borel,
feasible, and attains the maximum $M$ in
\eqref{eq:terminal-max-representation}.
\end{lemma}

\begin{proof}
If $x<g$, only $d=0$ is feasible, while if $x\ge g$ the two
candidate values are
$\alpha_RV^R(x)$ and
$\alpha_D+\alpha_RV^R(x-g)$. This proves attainment and
\eqref{eq:terminal-max-representation}. The set on which $D^*=1$ is
\[
\left\{(x,g):x\ge g,\
\alpha_D+\alpha_RV^R(x-g)
\ge\alpha_RV^R(x)\right\},
\]
which is closed because $V^R$ is continuous; hence $D^*$ is
Borel.
\end{proof}

We use throughout the following properties, established in the single-goal
analysis: $V^D$ and $V^R$ take values in $[0,1]$; $V^D$ is continuous on
$[0,T]\times[0,\infty)$ and $V^R$ is continuous on $[0,\infty)$; and
$V^D(s,x)=1$ for $x\ge b_G$, $V^R(x)=1$ for $x\ge b_R$. With
$\sigma_\lambda\sim\mathrm{Exp}(\lambda)$ independent of $\F_\infty$, the elementary
identity $\E[\int_0^\infty\lambda e^{-\lambda u}h(u)\,du]=\E[h(\sigma_\lambda)]$
applied to $h(u)=F_R(\widetilde W_u^{0,x;\piv})$ gives the equivalent form
\begin{equation}\label{eq:VRKrep}
V^R(x)=\sup_{\piv\in\cA^{K,\mathrm{pre}}_{0,\infty}}
\E\!\left[F_R\!\left(\widetilde W_{\sigma_\lambda}^{0,x;\piv}\right)\right].
\end{equation}
As noted in Lemma~\ref{lem:dual-operators-continuity},
$(s,x)\mapsto\mathcal J(s,x)$ is continuous
on $[0,T]\times[0,\infty)$ and
$x\mapsto\mathcal T(x)$ is continuous on
$[0,\infty)$. Hence for $\piv\in\cA^{K,\mathrm{pre}}_{t,T}$ the process
$(s,\omega)\mapsto\mathcal J(s,\widetilde W^{t,w;\piv}_s)$ is progressively
measurable and bounded by $1$,
$\mathcal T(\widetilde W^{t,w;\piv}_T)$ is bounded
by $1$, and all expressions are finite.

Finally, for $0\le t<T$ and $w\ge0$, record the main-model value and its
reduced counterpart:
\begin{equation}\label{eq:VRDK}
V^{R,D}(t,w)
:=\sup_{(\piv,d_T)}
\E_{t,w}\!\left[
\alpha_R\,\1_{\{\widehat W_{\tau_t-}^{t,w;\piv,d_T}\ge R\}}
+\alpha_Dd_T\right],
\end{equation}
\begin{equation}\label{eq:Vtilde}
\widetilde V^{R,D}(t,w)
:=\sup_{\piv\in\cA^{K,\mathrm{pre}}_{t,T}}
\E_{t,w}\!\left[\int_t^T\lambda e^{-\lambda(s-t)}
\mathcal J\!\left(s,\widetilde W_s^{t,w;\piv}\right)ds
+e^{-\lambda(T-t)}
\mathcal T\!\left(\widetilde W_T^{t,w;\piv}\right)\right],
\end{equation}
where the supremum in \eqref{eq:VRDK} runs over admissible pairs at
$(t,w)$.

\subsubsection{Additional Notation}
\begin{itemize}
    \item For a sigma-field $\mathcal H$ on $\Omega$ and a set $A\subseteq\Omega$, we write $\mathcal H|_A:=\{H\cap A:H\in\mathcal H\}$
    for the trace of $\mathcal H$ on $A$. We say that two sigma-fields
    $\mathcal H$ and $\mathcal K$ coincide on $A$ if
    $\mathcal H|_A=\mathcal K|_A$. Equalities of completed sigma-fields below are
    understood up to completion.

    \item For stopping times $S_1,S_2$, we use the standard stochastic-interval notation $\llbracket S_1,S_2\rrbracket
    :=\{(\omega,u):S_1(\omega)\le u\le S_2(\omega)\}$,
    with the analogous meanings for
    $\rrbracket S_1,S_2\rrbracket$, $\llbracket S_1,S_2\llbracket$, and
    $\rrbracket S_1,S_2\llbracket$.
    \item Let $\mathbb H=(\mathcal H_t)_{t\ge0}$ be a filtration satisfying the usual conditions with respect to which $B$ is a $d$-dimensional Brownian motion. If $S$ is an a.s.~finite $\mathbb H$-stopping time, we write $B^S_v:=B_{S+v}-B_S$, $v\ge0$, for the Brownian increment process after $S$. Its augmented natural filtration is denoted by $\mathbb B^S=(\mathcal B^S_v)_{v\ge0}$. By the strong Markov property, $B^S$ is a $d$-dimensional Brownian motion independent of $\mathcal H_S$.
\end{itemize}

\subsection{Auxiliary Results}

Many auxiliary results are needed for the main proof. We provide them here and motivate each. 

\subsubsection{Information Available at the Deadline}

Our first lemma shows that if the random-deadline goal has not arrived by $T$, then the funding decision made at that time can be taken to depend on only the market history and the revealed fixed-goal amount.

\begin{lemma}\label{lem:trace}
Set $\KK_T:=\F_T\vee\sigma(G)$, completed by  $\PP$-null sets. If $Y$ is a real-valued $\mathcal G_T$-measurable random variable,
there is a $\KK_T$-measurable $Y'$ such that $Y=Y'$ a.s.~on $\{\tau>T\}$.
Moreover, if $Y$ is $\{0,1\}$-valued, $Y'$ may be chosen to be $\{0,1\}$-valued. The same statement holds under
$\PP(\,\cdot\mid\tau>t)$ for every $t<T$.
\end{lemma}

\begin{proof}
We first show that, up to completion, $\bigcap_{u>T}\bigl(\F_u\vee\sigma(G)\bigr)=\KK_T$. To this end, fix $u_m>T$ with $u_m\downarrow T$ and set
$\mathcal A_m:=\F_{u_m}\vee\sigma(G)$ and
$\mathcal A_\infty:=\bigcap_m\mathcal A_m$. For any bounded
$\F_\infty$-measurable $Z$ and bounded Borel $h$, independence gives that $\E[Zh(G)\mid\mathcal A_m]
=h(G)\E[Z\mid\F_{u_m}]$.
L\'evy's downward theorem and the right continuity of $\filF$
therefore give
\[\E[Zh(G)\mid\mathcal A_\infty]=h(G)\E[Z\mid\mathcal F_T]=h(G)\E[Z\mid\mathcal K_T]=\E[Zh(G)\mid\KK_T].\]
A monotone-class argument extends this identity to every bounded $\F_\infty\vee\sigma(G)$-measurable variable. Applying it to $\mathcal A_\infty$-measurable variables yields $\mathcal A_\infty\subseteq\KK_T$ up to completion. The reverse inclusion is immediate.

Now suppose that the random variable $Y$ in the statement is bounded. On
$\{\tau>u_m\}$, the coordinates $\tau\wedge u_m$ and
$R\mathbf1_{\{\tau\le u_m\}}$ are constant. A monotone-class argument on the raw filtration therefore gives uniformly bounded $(\F_{u_m}\vee\sigma(G))$-measurable versions $Y_m$ equal to $Y$ on $\{\tau>u_m\}$, outside one common null set. Since the sigma-fields decrease, $Y':=\limsup_mY_m$ is measurable with respect to their intersection, which we showed equals $\KK_T$. On $\{\tau>T\}$, we eventually have $u_m<\tau$, so $Y'=Y$. For a general real-valued $Y$, apply this argument to $\arctan Y$ and invert the transform, extending the inverse arbitrarily off its range. If $Y$ is $\{0,1\}$-valued, replace $Y'$ by $\mathbf1_{\{Y'\ge1/2\}}$. Finally, $\PP(\,\cdot\mid\tau>t)\ll\PP$, so the same version gives the conditional statement.
\end{proof}

\subsubsection{Predictable Reductions}
The next result is a general result that will be used to justify a reduction in the effective class of controls from $\mathcal A_{t,\infty}^{K}$ to $\mathcal{A}_{t,\infty}^{K,\rm pre}$.

\begin{lemma}%
\label{lem:keyred}
Let $\mathbb H=(\mathcal H_s)_{s\ge0}$ satisfy the usual conditions, let
$\rho$ be a random time, and let $M$ be a real-valued random variable. Let
$\mathbb G$ be the progressive enlargement of $\mathbb H$ which makes $\rho$ a
stopping time and reveals $M$ at $\rho$, i.e.
\[
\mathcal G_s
=
\bigcap_{u>s}
\left(
\mathcal H_u\vee\sigma(\rho\wedge u)
\vee\sigma(M\mathbf 1_{\{\rho\le u\}})
\right).
\]
Then every (possibly vector-valued) $\mathbb G$-predictable process $Y$ admits an $\mathbb H$-predictable
process $Y'$ such that $Y=Y'$ on $\rrbracket0,\rho\rrbracket$ up to evanescence.
Consequently, $Y=Y'$ $ds\otimes d\mathbb P$-a.e.~on $\llbracket0,\rho\rrbracket$.
\end{lemma}

\begin{proof}
For the progressive enlargement without $M$, this is the result of \cite[Proposition~2.11(b)]{AksamitJeanblanc2017}. The
case with $M$ follows from the same monotone-class argument. Indeed, $M$ contributes no predictable information on
$\rrbracket0,\rho\rrbracket$ so every $\mathbb G$-predictable process has an
$\mathbb H$-predictable reduction on $\rrbracket0,\rho\rrbracket$, up to
evanescence.
The stated $ds\otimes d\mathbb P$ conclusion follows because the only difference
between $\rrbracket0,\rho\rrbracket$ and $\llbracket0,\rho\rrbracket$ is the
time-zero slice, which is $ds\otimes d\mathbb P$-null.
\end{proof}

As noted above, the next two lemmas use this general result in the context of our problem to argue that we may restrict to pre-arrival controls before goals have arrived.

\begin{lemma}
\label{lem:beforeT-red}
Fix $0\le t<T$ and work on $\{\tau>t\}$. For every
admissible pair $(\piv,d_T)$ there exists
$\piv^0\in\cA^{K,\mathrm{pre}}_{t,T}$ such that
\[
\piv_s=\piv^0_s
\quad\text{for }ds\otimes d\PP\text{-a.e. }(s,\omega)\text{ with }
 t\le s<\tau\wedge T .
\]
Consequently, if $\widehat W=\widehat W^{t,w;\piv,d_T}$ and
$\widetilde W=\widetilde W^{t,w;\piv^0}$, then almost surely
$\widehat W_s=\widetilde W_s$ for $s\in[t,\tau\wedge T)$, and $\widehat W_{\tau-}=\widetilde W_{\tau}$ on $\{\tau<T\}$, $
\widehat W_{T-}=\widetilde W_T$ on $\{\tau>T\}$.
\end{lemma}

\begin{proof}
On $[0,T)$ the component $\sigma(G\1_{\{T\le u\}})$ in the definition of $\mathbb{G}$ is trivial.
Thus, before $T$, the filtration $\filG$ is the progressive enlargement of $\filF$
by $\tau$ and the process $R\1_{\tau\leq u}$. Applying Lemma~\ref{lem:keyred} to
$\piv\1_{[t,T)}$ gives an $\filF$-predictable process $\rhov$ which coincides with
$\piv$ on $\{(s,\omega):0\le s<\tau(\omega)\}\cap([t,T)\times\Omega)$ up to evanescence. Let $q_K(z):=z\,\mathbf1_{\{z\in\Pi_K\}}$ for $z\in\mathbb R^n$.
Because $\Pi_K$ is closed and $0\in\Pi_K$, $q_K$ is a Borel map into
$\Pi_K$ and is the identity on $\Pi_K$. Thus
$\piv^0:=q_K(\rhov)$ is $\filF$-predictable, $\Pi_K$-valued, and still
coincides with $\piv$ wherever $\piv$ already takes values in
$\Pi_K$. 

Before the first goal resolution no jump in \eqref{eq:postpay} has occurred, so
$\widehat W$ solves the diffusion equation \eqref{eq:prepay} driven by $\piv$ on
$[t,\tau\wedge T)$. The stochastic and drift integrals up to any
$s<\tau\wedge T$ depend on the control only through its restriction to
$\{t\le r<s\}$, where $\piv=\piv^0$ $dr\otimes d\PP$-a.e. Pathwise uniqueness for
\eqref{eq:prepay} gives $\widehat W_s=\widetilde W_s$ for
$s<\tau\wedge T$, and the left-limit identities follow from continuity of
$\widetilde W$.
\end{proof}

\begin{lemma}
\label{lem:postT-red}
Let $\FG=(\F_s\vee\sigma(G))_{s\ge T}$ and augment it so that it is complete and right-continuous. For
every $\piv\in\cA^K_{T,\infty}$ there exists an $\FG$-predictable, $\Pi_K$-valued
process $\piv^G$ on $[T,\infty)$ such that
\[
\piv_s=\piv^G_s
\quad\text{for }ds\otimes d\PP\text{-a.e. }(s,\omega)\text{ with }T\le s<\tau.
\]
Moreover, for $v\ge0$, the process $\piv^G_{T+v}$ is predictable with
respect to $(\KK_T\vee\BB^T_v)_{v\ge0}$, where $\KK_T:=\F_T\vee\sigma(G)$ and
$\BB^T$ is the augmented natural filtration of $B^T=(B_{T+v}-B_T)_{v\ge0}$.
\end{lemma}

\begin{proof}
Let $\mathbb H=(\mathcal H_s)_{s\ge0}$ be the initially enlarged Brownian filtration $\mathcal H_s:=\F_s\vee\sigma(G)$, completed and made right-continuous. On the time region $[T,\infty)$, the filtration $\filG$ is the progressive enlargement of $\mathbb H$ by $\tau$ and the process $R \1_{\{\tau\leq u\}}$. Indeed, $G$ is already contained in $\mathbb H$, while the pair $(\tau,R)$ is revealed only when $\tau$ occurs. Because $G$ is independent of $\F_\infty$, $B$ remains a Brownian motion in the initially enlarged filtration $\mathbb H$. Applying Lemma~\ref{lem:keyred} to the $\filG$-predictable process $\piv\1_{[T,\infty)}$ gives an $\mathbb H$-predictable process $\rhov$ which coincides with $\piv$ on $\dbl0,\tau\dbr\cap([T,\infty)\times\Omega)$ up to evanescence. Restricting to $[T,\infty)$ and applying the Borel retraction $q_K(z)=z\mathbf1_{\{z\in\Pi_K\}}$ used in the preceding lemma yields an $\FG$-predictable, $\Pi_K$-valued process $\piv^G$ satisfying the asserted equality for $ds\otimes d\PP$-a.e.\ $(s,\omega)$ with $T\le s<\tau$. On $\{\tau\le T\}$ this set is empty, so nothing else is needed.

Finally, since $\F_{T+v}=\F_T\vee\BB^T_v$ up to null sets for Brownian natural
filtrations, the initially enlarged filtration satisfies
$\FG_{T+v}=\KK_T\vee\BB^T_v$ up to completion. Hence the process
$\piv^G_{T+v}$ is predictable with respect to $(\KK_T\vee\BB^T_v)_{v\ge0}$.
\end{proof}

\subsubsection{Representations of Expectations}

In our analysis, we will need to simplify expectations in terms of the exponential density. To this end, we supply the following two lemmas.

\begin{lemma}\label{lem:disint}
Fix $0\le t<T$ and work on $\{\tau>t\}$. Let $\piv\in\cA^{K,\mathrm{pre}}_{t,T}$ and
$\widetilde W=\widetilde W^{t,w;\piv}$. For every bounded Borel
$f:[t,T]\times[0,\infty)\to\R$,
\[
\E_{t,w}\!\left[f\!\left(\tau_t,\widetilde W_{\tau_t}\right)\1_{\{\tau_t<T\}}\right]
=\int_t^T\lambda e^{-\lambda(s-t)}\,\E_{t,w}\!\left[f\!\left(s,\widetilde W_s\right)\right]ds,
\]
and for every bounded Borel $g:[0,\infty)\to\R$,
\[
\E_{t,w}\!\left[g\!\left(\widetilde W_T\right)\1_{\{\tau_t>T\}}\right]
=e^{-\lambda(T-t)}\,\E_{t,w}\!\left[g\!\left(\widetilde W_T\right)\right].
\]
\end{lemma}

\begin{proof}
Under $\PP(\cdot\mid\tau>t)$ the residual $\eta=\tau_t-t$ has density
$\lambda e^{-\lambda u}$ and is independent of $\F_\infty$. Since $\piv$ is
$\filF$-predictable, the path $(\widetilde W_s)_{s\in[t,T]}$ is $\F_T$-measurable and the map $(s,\omega)\mapsto f(s,\widetilde W_s(\omega))$ is jointly measurable and
bounded. By conditioning on $\eta$, using the density $\lambda e^{-\lambda u}$, the independence of $\eta$ from $\mathcal F_T$, and Fubini's theorem,
\[
\E_{t,w}\!\left[f(t+\eta,\widetilde W_{t+\eta})\1_{\{\eta<T-t\}}\right]
=\int_0^{T-t}\lambda e^{-\lambda u}\,\E_{t,w}[f(t+u,\widetilde W_{t+u})]\,du.
\]
The change of variables $s=t+u$ gives the first identity. The second follows from
$\PP(\eta>T-t\mid \tau>t)=e^{-\lambda(T-t)}$ and the independence of $\eta$ from $\widetilde W_T$.
\end{proof}

\begin{lemma}\label{lem:shift}
Fix $0\le t<T$ and work on $\{\tau>t\}$. Set $\KK_T:=\F_T\vee\sigma(G)$ and let
$\sigma_\lambda\sim\mathrm{Exp}(\lambda)$ be independent of
$\F_\infty\vee\sigma(R,G)$. Let $(\psi(s))_{s\ge T}$ be a bounded process such that $(s,\omega)\mapsto\psi(s,\omega)$ is $\mathcal B([T,\infty))\otimes(\mathcal F_\infty\vee\sigma(R,G))$-measurable, where $\mathcal B([T,\infty))$ is the Borel $\sigma$-field on $[T,\infty)$. Then, almost surely,
\[
\E_{t,w}\!\left[\1_{\{\tau_t>T\}}\,\psi(\tau_t)\,\middle|\,\KK_T\right]
=e^{-\lambda(T-t)}\,\E\!\left[\psi(T+\sigma_\lambda)\,\middle|\,\KK_T\right].
\]
\end{lemma}

\begin{proof}
On $\{\tau>t\}$ write $\tau_t=t+\eta$ with $\eta\sim\mathrm{Exp}(\lambda)$
independent of $\F_\infty\vee\sigma(R,G)\supseteq\KK_T$. Since $\{\tau>t\}$ is
independent of $\F_\infty\vee\sigma(R,G)$, for each $u\ge T-t$, $\E_{t,w}[\psi(t+u)\mid\KK_T]=\E[\psi(t+u)\mid\KK_T]$.
Hence, by the conditional version of Fubini's theorem,
\[
\E_{t,w}\!\left[\1_{\{t+\eta>T\}}\psi(t+\eta)\,\middle|\,\KK_T\right]
=\int_{T-t}^{\infty}\E[\psi(t+u)\mid\KK_T]\,\lambda e^{-\lambda u}\,du .
\]
Substituting $u=(T-t)+v$ gives
$\lambda e^{-\lambda u}=\lambda e^{-\lambda(T-t)}e^{-\lambda v}$ and $t+u=T+v$.
Therefore the right-hand side equals
\[
e^{-\lambda(T-t)}
\int_0^\infty \E[\psi(T+v)\mid\KK_T]\,\lambda e^{-\lambda v}\,dv
=
e^{-\lambda(T-t)}\E[\psi(T+\sigma_\lambda)\mid\KK_T],
\]
where the last equality again uses conditional Fubini and the independence of $\sigma_\lambda$ from $\F_\infty\vee\sigma(R,G)$.
\end{proof}

\subsubsection{Alternative Representation of the Single Goal Value Functions}
Next, we will also need an alternative way to interpret the single goal value functions. Let $\beta$ be the coordinate process on $C([0,\infty);\R^d)$,
let $\Wien$ be Wiener measure, and let $\filB=(\BB_v)_{v\ge0}$ be the usual
augmentation of the canonical Brownian filtration. Let $\cU^K$ be the set of
$\Pi_K$-valued $\filB$-predictable processes on this canonical space. Thus each
$\nuv\in\cU^K$ is a predictable functional of a Brownian increment path. If $W$ is
any $d$-dimensional Brownian motion on any filtered probability space, we write
$\nuv(W)$ for the process obtained by evaluating this functional on the path of
$W$, and let $X^{x;\nuv}(W)$ solve
\[
dX_v=X_v\bigl(\nuv_v(W)^\top\thetavec\,dv+\nuv_v(W)^\top\sigmavec\,dW_v\bigr),
\qquad X_0=x .
\]
Define the following payoffs, with expectation under Wiener measure,
\[
P_D^{\nuv}(s,x):=\E_{\Wien}\!\left[F_G\!\left(X^{x;\nuv}_{T-s}(\beta)\right)\right],
\qquad
P_R^{\nuv}(x):=\E_{\Wien}\!\left[\int_0^\infty\lambda e^{-\lambda u}
F_R\!\left(X^{x;\nuv}_u(\beta)\right)du\right].
\]

Our next result relates these to $V^D$ and $V^R$.

\begin{lemma}\label{lem:fresh}
For all $(s,x)\in[0,T]\times[0,\infty)$, $V^D(s,x)=\sup_{\nuv\in\cU^{K}}P_D^{\nuv}(s,x)$ and $V^R(x)=\sup_{\nuv\in\cU^{K}}P_R^{\nuv}(x)$.
\end{lemma}

\begin{proof}
We prove the fixed-deadline identity. The random deadline identity follows analogously.

\emph{\underline{Step 1:}} We first prove the inequality $\sup_{\nuv\in\cU^K}P_D^{\nuv}(s,x)\le V^D(s,x)$. Consider the Brownian motion $B^s=(B_{s+v}-B_s)_{v\ge0}$.
For $\nuv\in\cU^{K}$ the shifted control
$\piv_u:=\nuv_{u-s}(B^{s})$, $u\in[s,T]$, is $\filF$-predictable and
$\Pi_K$-valued, so $\piv\in\cA^{K,\mathrm{pre}}_{s,T}$. Moreover, $\widetilde W_T^{s,x;\piv}
=
X^{x;\nuv}_{T-s}(B^s)$. Since $B^s$ has the Wiener
law,
\[
P_D^{\nuv}(s,x)
=
\E\!\left[
F_G\!\left(\widetilde W_T^{s,x;\piv}\right)
\right]
\le
V^D(s,x).
\]
Taking the supremum over $\nuv\in\cU^K$ gives the desired inequality.

\emph{\underline{Step 2:}} Conversely, we can show $V^D(s,x)\le\sup_{\nuv\in\cU^K}P_D^{\nuv}(s,x)$.
Fix $\piv\in\cA^{K,\mathrm{pre}}_{s,T}$. Let
$\PP^{\omega}:=\PP(\cdot\mid\F_s)(\omega)$ be a regular conditional
probability given $\F_s$, chosen on the underlying raw standard-Borel space (cf.~\cite[Thm.~8.5]{Kallenberg2021}).
Since $B^s$ is independent of $\F_s$, for $\PP$-a.e. $\omega$, under
$\PP^\omega$ the process $B^s$ is a Brownian motion and every set in $\F_s$
has $\PP^\omega$-probability $0$ or $1$.

By the standing raw-filtration convention, we choose a raw predictable version of
$\piv$. On the raw Brownian filtration one has
$\F_{s+v}=\F_s\vee\BB^s_v$ before completion, so the relative-time process
$v\mapsto\piv_{s+v}$ is predictable with respect to
$(\F_s\vee\BB^s_v)_{v\ge0}$. We use the standard predictable factorization fact that, under a
conditional law for which $\F_s$ is trivial, any $\Pi_K$-valued $(\F_s\vee\mathcal B^s_v)_{v\ge0}$-predictable process is equal (up to $dv\otimes d\PP^\omega$-null sets) to a predictable functional $\nuv^\omega(B^s)$ for some $\nuv^\omega\in\cU^K$.

Therefore, for $\PP$-a.e.\ $\omega$, there exists
$\nuv^\omega\in\cU^K$ such that
\[
\piv_{s+v}=\nuv^\omega_v(B^s)
\qquad
dv\otimes d\PP^\omega\text{-a.e. on }[0,T-s]\times\Omega.
\]
Under $\PP^\omega$, the two controls agree
$dv\otimes d\PP^\omega$-a.e.~and are bounded. Hence their drift integrals coincide, and their It\^o integrals coincide by the It\^o isometry applied to the difference of the two bounded integrands. Pathwise uniqueness for the wealth equation then gives $\widetilde W_T^{s,x;\piv}
=X_{T-s}^{x;\nuv^\omega}(B^s)$, $\PP^\omega$-a.s.
Consequently,
\[
\E^{\PP^\omega}\!\left[
F_G\!\left(\widetilde W_T^{s,x;\piv}\right)
\right]
=
P_D^{\nuv^\omega}(s,x)
\le
\sup_{\nuv\in\cU^K}P_D^{\nuv}(s,x).
\]
Taking $\PP$-expectations and using
$\E[\,\cdot\,]=\E[\E^{\PP^\cdot}[\,\cdot\,]]$ yields $\E\!\left[
F_G\!\left(\widetilde W_T^{s,x;\piv}\right)
\right]
\le
\sup_{\nuv\in\cU^K}P_D^{\nuv}(s,x)$.
Taking the supremum over $\piv\in\cA^{K,\mathrm{pre}}_{s,T}$ gives $V^D(s,x)
\le
\sup_{\nuv\in\cU^K}P_D^{\nuv}(s,x)$.
\end{proof}

We can use these representations to prove two useful estimates.

\begin{lemma}\label{lem:condbd}
Let $S$ be a finite $\filG$-stopping time and let $\HH\subseteq\mathcal F$ be a sub-$\sigma$-algebra, countably generated before completion, such that $S$ is $\HH$-measurable and $B^{S}$ is a Brownian motion independent of $\HH$, and let $X\ge0$ be $\HH$-measurable. Let $\piv$ be a $\Pi_K$-valued process that is predictable with respect to $(\HH\vee\BB^{S}_v)_{v\ge0}$, and let $Z$ solve
\[
dZ_v=Z_v(\piv_v^\top\thetavec\,dv+\piv_v^\top\sigmavec\,dB^{S}_v),
\qquad Z_0=X .
\]
Then,
\begin{enumerate}
\item[(i)] for every $A\in\HH$ such that $S\le T$ on $A$, $\1_A\,\E\!\left[F_G(Z_{(T-S)^+})\mid\HH\right]\le \1_A V^D(S\wedge T,X)$ a.s.;
\item[(ii)] if in addition $\sigma_\lambda\sim\mathrm{Exp}(\lambda)$ is independent of
$\HH\vee\sigma(B^{S})$, then $\E\!\left[F_R(Z_{\sigma_\lambda})\mid\HH\right]\le V^R(X)$ a.s.
\end{enumerate}
\end{lemma}

\begin{proof}
Let $\PP^{\omega}:=\PP(\cdot\mid\HH)(\omega)$ be a regular conditional probability
given $\HH$ \cite[Thm.~8.5]{Kallenberg2021}.

Exactly as in the proof of Lemma~\ref{lem:fresh}, for $\PP$-a.e.\ $\omega$: under $\PP^\omega$ the process $B^S$ is a Brownian
motion, $S$ and $X$ equal the constants $s:=S(\omega)$ and
$x:=X(\omega)$ almost surely, and $\piv$ is equal to
$\nuv^\omega(B^S)$, $dv\otimes d\PP^\omega$-a.e., for some
$\nuv^\omega\in\cU^K$. Since the two controls agree in this
sense and are bounded, their drift terms and It\^o integrals coincide under $\PP^\omega$. Pathwise uniqueness for the wealth SDE gives $Z=X^{x;\nuv^\omega}(B^S)$ under $\PP^\omega$. If
$\omega\in A$, then $s=S(\omega)\le T$, so $(T-S)^+=T-s$, and
Lemma~\ref{lem:fresh} gives
\[
\E^{\PP^{\omega}}\!\left[F_G(Z_{(T-S)^+})\right]=P_D^{\nuv^{\omega}}(s,x)\le V^D(s,x).
\]
Since $A\in\HH$, this proves part (i) after multiplying by $\1_A$. For (ii),
the independence of $\sigma_\lambda$ from $\HH\vee\sigma(B^{S})$ implies that under $\PP^{\omega}$ the
variable $\sigma_\lambda$ is $\mathrm{Exp}(\lambda)$ and independent of $B^{S}$,
whence
$\E^{\PP^{\omega}}[F_R(Z_{\sigma_\lambda})]=P_R^{\nuv^{\omega}}(x)\le V^R(x)$ by
Lemma~\ref{lem:fresh}.
\end{proof}

\subsubsection{Finite Patch Approximation and Predictable Pasting of Controls}
Our last auxiliary result builds a finite approximation from
$\varepsilon$-optimal controls.

\begin{lemma}[Finite-patch approximation]\label{lem:patch}
For every $\varepsilon>0$:
\begin{enumerate}
\item[(i)] there exist a finite Borel partition $A_1,\dots,A_N$ of
$[0,T]\times[0,b_G]$ and controls $\nuv^{1},\dots,\nuv^{N}\in\cU^{K}$ such that
$P_D^{\nuv^{i}}(s,x)\ge V^D(s,x)-\varepsilon$ whenever $(s,x)\in A_i$;
\item[(ii)] there exist a finite Borel partition $B_1,\dots,B_M$ of $[0,b_R]$ and
controls $\zetav^{1},\dots,\zetav^{M}\in\cU^{K}$ such that
$P_R^{\zetav^{j}}(y)\ge V^R(y)-\varepsilon$ whenever $y\in B_j$.
\end{enumerate}
For $x\ge b_G$ in \textup{(i)} or $y\ge b_R$ in \textup{(ii)}, the zero control is
exactly optimal, attaining the value $1$.
\end{lemma}

\begin{proof}
We first record two continuity facts. By scaling, the solution of \eqref{eq:prepay}
satisfies $X^{x;\nuv}_v(\beta)=x\,\mathcal E^{\nuv}_v$ with
$\mathcal E^{\nuv}_v:=\exp\!\bigl(\int_0^v(\nuv_r^\top\thetavec-\tfrac12\nuv_r^\top\Sigmav\nuv_r)\,dr+\int_0^v\nuv_r^\top\sigmavec\,d\beta_r\bigr)$,
which has continuous paths and does not depend on $x$. Hence, for fixed
$\nuv\in\cU^{K}$, $(s,x)\mapsto x\,\mathcal E^{\nuv}_{T-s}$ is a.s.\ jointly
continuous, and since $F_G\in[0,1]$ is continuous, bounded convergence gives that
$(s,x)\mapsto P_D^{\nuv}(s,x)=\E_{\Wien}[F_G(x\,\mathcal E^{\nuv}_{T-s})]$ is continuous on
$[0,T]\times[0,\infty)$. Likewise, for fixed $\zetav\in\cU^{K}$ and $x\to x'$,
$F_R(x\,\mathcal E^{\zetav}_u)\to F_R(x'\,\mathcal E^{\zetav}_u)$ for every $(u,\omega)$
by continuity of $F_R$. Since $F_R\in[0,1]$ dominated convergence on $\lambda e^{-\lambda u}\,du\otimes d\Wien$ shows $x\mapsto P_R^{\zetav}(x)$ is
continuous on $[0,\infty)$.

For (i), fix $\varepsilon>0$. For each $(s,x)\in[0,T]\times[0,b_G]$,
Lemma~\ref{lem:fresh} provides $\nuv^{s,x}\in\cU^{K}$ with
$P_D^{\nuv^{s,x}}(s,x)\ge V^D(s,x)-\varepsilon/3$. The function
$(s',x')\mapsto V^D(s',x')-P_D^{\nuv^{s,x}}(s',x')$ is continuous and is
$\le\varepsilon/3$ at $(s,x)$, so it is $<\varepsilon$ on a relatively open neighborhood
$U_{s,x}$ of $(s,x)$. Therefore, on $U_{s,x}$, $P_D^{\nuv^{s,x}}\ge V^D-\varepsilon$. These
neighborhoods cover the compact set $[0,T]\times[0,b_G]$ and so we can extract a finite subcover
$U_{s_1,x_1},\dots,U_{s_N,x_N}$. From this we can construct a disjoint cover so that we arrive at a Borel partition
$A_1,\dots,A_N$ with $A_i\subseteq U_{s_i,x_i}$ and controls $\nuv^{i}:=\nuv^{s_i,x_i}$.
Part (ii) is identical on the compact interval $[0,b_R]$. Finally, under the zero
control wealth is constant, so if $x\ge b_G$ then $F_G(x)=1=V^D(s,x)$ and if
$y\ge b_R$ then $F_R(y)=1=V^R(y)$.
\end{proof}

\begin{remark}\label{rem:pred.pasting}
In the main proof of Proposition~\ref{prop:reduction-opt} we use the following standard predictable-pasting fact. Let $S$ be a
$\filG$-stopping time, let $C_1,\dots,C_m$ be a finite
$\G_S$-measurable partition, and let
$\nuv^1,\dots,\nuv^m\in\cU^K$. Then the process
\[
\Pi_r
:=
\sum_{k=1}^m
\1_{C_k}\,
\nuv^k_{(r-S)^+}(B^S)\,
\1_{\{S<r\}},
\qquad r\ge0,
\]
is $\filG$-predictable and $\Pi_K$-valued. More generally, finite concatenations of such shifted controls at $\filG$-stopping times remain admissible.
\end{remark}

\subsection{Main Proof}

We are now ready to complete the proof of the proposition. 

\begin{proof}[Proof of Proposition~\ref{prop:reduction-opt}]
Fix $(t,w)$, work under $\E_{t,w}$, and let $(\piv,d_T)$ be an
admissible pair. Write
\[
P:=\alpha_RI_R+\alpha_Dd_T,\qquad
I_R:=\mathbf1_{\{\widehat W_{\tau_t-}^{t,w;\piv,d_T}\ge R\}},
\qquad
\widehat W:=\widehat W^{t,w;\piv,d_T}.
\]
Since $\PP(\tau_t=T\mid\tau>t)=0$, split the payoff over
$\{\tau_t<T\}$ and $\{\tau_t>T\}$.

\emph{Step 1 (upper bound).}
Let $\piv^0$ be the pre-deadline reduction of
Lemma~\ref{lem:beforeT-red}, and write
$\widetilde W:=\widetilde W^{t,w;\piv^0}$. On
$\{\tau_t<T\}$,
$\widehat W_{\tau_t-}=\widetilde W_{\tau_t}$ and the post-payment
wealth is $\Gamma_R(\widetilde W_{\tau_t})$. Feasibility gives
$d_T\le\mathbf1_{\{\widehat W_{T-}\ge G\}}$.
On the stochastic interval $\rrbracket\tau_t,T\llbracket$, no further non-Brownian information is revealed. More precisely, for
each $v\ge0$,
\[
\mathcal G_{\tau_t+v}\big|_{\{\tau_t+v<T\}}
=
\bigl(\mathcal G_{\tau_t}\vee\mathcal B^{\tau_t}_v\bigr)
\big|_{\{\tau_t+v<T\}},
\]
up to completion. A standard predictable-factorization argument shows that the shifted continuation control $\piv^{\tau_t}_v
:=
\piv_{\tau_t+v}\mathbf1_{\{0<v<T-\tau_t\}}$, $v\ge0$,
has a $\Pi_K$-valued version that is predictable with respect to $(\mathcal G_{\tau_t}\vee\mathcal B^{\tau_t}_v)_{v\ge0}$.

Consequently, on $\{\tau_t<T\}$, the continuation wealth $\widehat W_{T-}$ is measurable with respect to the trace of $\F_\infty\vee\sigma(\tau,R)$, which also contains the trace of $\mathcal G_{\tau_t}$. Since $G$ is independent of this sigma-field under $\PP(\,\cdot\mid\tau>t)$, averaging over $G$ and applying Lemma~\ref{lem:condbd}(i) therefore gives
\begin{align*}
\E[d_T\mid\mathcal G_{\tau_t}]\mathbf1_{\{\tau_t<T\}}
&\le
\E[\mathbf1_{\{\widehat W_{T-}\ge G\}}
   \mid\mathcal G_{\tau_t}]\mathbf1_{\{\tau_t<T\}}\\
&=
\E[F_G(\widehat W_{T-})
   \mid\mathcal G_{\tau_t}]\mathbf1_{\{\tau_t<T\}}\le
V^D\!\left(\tau_t,\Gamma_R(\widetilde W_{\tau_t})\right)
\mathbf1_{\{\tau_t<T\}}.
\end{align*}
Adding the random-goal reward, averaging over $R$, and applying
Lemma~\ref{lem:disint} yields
\begin{equation}\label{eq:ub-early}
\E_{t,w}[P\mathbf1_{\{\tau_t<T\}}]
\le
\int_t^T\lambda e^{-\lambda(s-t)}
\E_{t,w}[\mathcal J(s,\widetilde W_s)]\,ds .
\end{equation}

On $\{\tau_t>T\}$,
$\widehat W_{T-}=\widetilde W_T$. By Lemma~\ref{lem:trace}, choose a
$\{0,1\}$-valued $\KK_T$-measurable $d_T'$ equal to $d_T$ on this
event and set $\bar d:=d_T'\mathbf1_{\{\widetilde W_T\ge G\}}$ and $
X:=\widetilde W_T-\bar dG$.
Then $\bar d\in\Theta(\widetilde W_T,G)$ everywhere,
$\bar d=d_T$ on $\{\tau_t>T\}$, and $X$ is nonnegative and
$\KK_T$-measurable. Lemma~\ref{lem:postT-red} produces a
$\filF^G$-predictable continuation control agreeing with $\piv$ before
$\tau_t$. Let $Z$ denote its relative-time wealth process:
\[
dZ_v
=Z_v\Bigl((\piv^G_{T+v})^\top\thetavec\,dv
          +(\piv^G_{T+v})^\top\sigmavec\,dB^T_v\Bigr),
\qquad Z_0=X.
\]
Pathwise uniqueness gives
$\widehat W_{\tau_t-}=Z_{\tau_t-T}$ on $\{\tau_t>T\}$.
The process
$\psi(s):=\mathbf1_{\{Z_{s-T}\ge R\}}$, $s\ge T$, is bounded and
jointly measurable in the sense of Lemma~\ref{lem:shift}. Hence
Lemma~\ref{lem:shift}, independence of $R$, and
Lemma~\ref{lem:condbd}(ii) give $\E_{t,w}[I_R\mathbf1_{\{\tau_t>T\}}\mid\KK_T]
\le e^{-\lambda(T-t)}V^R(X)$.
The fixed-goal term satisfies $\E_{t,w}[d_T\mathbf1_{\{\tau_t>T\}}\mid\KK_T]
=e^{-\lambda(T-t)}\bar d$.
Since $\bar d\in\Theta(\widetilde W_T,G)$,
Lemma~\ref{lem:optional-terminal-selector} implies
 $\E_{t,w}[P\mathbf1_{\{\tau_t>T\}}\mid\KK_T]
\le e^{-\lambda(T-t)}M(\widetilde W_T,G)$.
Averaging over $G$ therefore gives
\begin{equation}\label{eq:ub-survival}
\E_{t,w}[P\mathbf1_{\{\tau_t>T\}}]
\le e^{-\lambda(T-t)}
\E_{t,w}[\mathcal T(\widetilde W_T)].
\end{equation}
Combining \eqref{eq:ub-early} and \eqref{eq:ub-survival}, then taking the supremum over admissible pairs, proves
$V^{R,D}\le\widetilde V^{R,D}$.

\emph{Step 2 (lower bound).}
Fix $\piv\in\cA^{K,\mathrm{pre}}_{t,T}$, set
$\widetilde W:=\widetilde W^{t,w;\piv}$, fix $\varepsilon>0$, and
take the finite patches of Lemma~\ref{lem:patch}. Construct a full
strategy $\piv^\varepsilon$ and decision $d_T^\varepsilon$ as
follows. Use $\piv$ before $\tau_t\wedge T$. On
$\{\tau_t=s<T\}$, put
$X_s:=\Gamma_R(\widetilde W_s)$ and, on
$\rrbracket s,T\rrbracket$, use the zero control when
$X_s\ge b_G$, and otherwise the shifted control
$\nuv^i(B^{\tau_t})$ selected by $(s,X_s)\in A_i$. At $T$ set $d_T^\varepsilon:=\mathbf1_{\{\widehat W_{T-}\ge G\}}$ on $\{\tau_t<T\}$,
and use the zero control after $T$ on this branch. On
$\{\tau_t>T\}$, set $d_T^\varepsilon:=D^*(\widetilde W_T,G)$,  $Y_T:=\widetilde W_T-d_T^\varepsilon G$,
and, on $\rrbracket T,\tau_t\rrbracket$, use the zero control when $Y_T\ge b_R$, and otherwise the shifted control $\zetav^j(B^T)$ selected by $Y_T\in B_j$. Use the zero control after $\tau_t$. Define the pair arbitrarily and feasibly on the null event
$\{\tau_t=T\}$.

The switching events are measurable at the corresponding stopping
times. Remark~\ref{rem:pred.pasting} therefore makes
$\piv^\varepsilon$ admissible. The decision is
$\mathcal G_T$-measurable because $D^*$ is Borel, and is feasible
because $D^*(x,g)\le\mathbf1_{\{x\ge g\}}$. Thus
$(\piv^\varepsilon,d_T^\varepsilon)$ is an admissible pair.
In the remainder of this step, let $P$ and $I_R$ denote the payoff and
random-goal indicator generated by this constructed pair.

On $\{\tau_t=s<T\}$, conditioning at $\tau_t$ and integrating the
independent Brownian increments gives
\[
\E[d_T^\varepsilon\mid\mathcal G_{\tau_t}]
\mathbf1_{\{\tau_t<T\}}
\ge
\left(V^D\!\left(
\tau_t,\Gamma_R(\widetilde W_{\tau_t})\right)-\varepsilon\right)
\mathbf1_{\{\tau_t<T\}},
\]
where there is no approximation loss when $X_s\ge b_G$.
Averaging over $R$ and using
Lemma~\ref{lem:disint} gives
\begin{equation}\label{eq:lb-early}
\E_{t,w}[P\mathbf1_{\{\tau_t<T\}}]
\ge
\int_t^T\lambda e^{-\lambda(s-t)}
\E_{t,w}[\mathcal J(s,\widetilde W_s)]\,ds
-\alpha_D\varepsilon(1-e^{-\lambda(T-t)}).
\end{equation}

On $\{\tau_t>T\}$, the post-decision wealth is
$Y_T=\widetilde W_T-D^*(\widetilde W_T,G)G$. Applying
Lemma~\ref{lem:shift} and the selected random-goal patch conditionally
on $\KK_T$ gives
\[
\E_{t,w}[I_R\mathbf1_{\{\tau_t>T\}}\mid\KK_T]
\ge e^{-\lambda(T-t)}
\bigl(V^R(Y_T)-\varepsilon\bigr),
\]
again with no approximation loss when $Y_T\ge b_R$. The fixed-goal term is
$e^{-\lambda(T-t)}D^*(\widetilde W_T,G)$ conditionally on
$\KK_T$. Since $D^*$ attains $M$, averaging over $G$ yields
\begin{equation}\label{eq:lb-late}
\E_{t,w}[P\mathbf1_{\{\tau_t>T\}}]
\ge e^{-\lambda(T-t)}
\E_{t,w}[\mathcal T(\widetilde W_T)]
-\alpha_R\varepsilon e^{-\lambda(T-t)}.
\end{equation}
Adding \eqref{eq:lb-early} and \eqref{eq:lb-late}, and using $C_{t,T}:=\alpha_D(1-e^{-\lambda(T-t)})
+\alpha_Re^{-\lambda(T-t)}\le1$,
shows that the reduced payoff generated by $\piv$ is at most
$V^{R,D}(t,w)+\varepsilon C_{t,T}$. Let
$\varepsilon\downarrow0$ and take the supremum over pre-deadline
controls to obtain the reverse inequality. This proves
Proposition~\ref{prop:reduction-opt}.
\end{proof}

We note that the construction in Step~2 of the proof demonstrates that the value is unchanged when every affordable fixed-deadline goal is funded on $\{\tau_t<T\}$ and the decision $D^*(Y_{T-}^{t,w;\piv},G)$ is used on $\{\tau_t>T\}$.

\begin{proposition}\label{prop:reduction}
For every $(t,w)\in[0,T)\times[0,\infty)$,
\begin{equation}
\label{eq:tilde-Vforced}
V_{\mathrm{commit}}^{R,D}(t,w)=\sup_{\piv\in\mathcal A^{K,{\rm pre}}_{t,T}}
\E_{t,w}\!\left[
\int_t^T
\lambda e^{-\lambda(s-t)}
\mathcal J\bigl(s,\widetilde W_s^{t,w;\piv}\bigr)\,ds
+e^{-\lambda(T-t)}
\mathcal T_{\mathrm{commit}}
\bigl(\widetilde W_T^{t,w;\piv}\bigr)
\right].
\end{equation}
\end{proposition}

\begin{proof}[Proof of Proposition~\ref{prop:reduction}]
Define
$d_T^{\rm commit}(\piv):=\mathbf1_{\{Y_{T-}^{t,w;\piv}\ge G\}}$
and fix $d_T=d_T^{\rm commit}(\piv)$.
On $\{\tau_t<T\}$, the inequality
$d_T\le\mathbf1_{\{\widehat W_{T-}\ge G\}}$ in the preceding proof is an equality, so the contribution $\mathcal J$ is unchanged. On $\{\tau_t>T\}$, we have
$\bar d=\mathbf1_{\{\widetilde W_T\ge G\}}$.
The post-decision wealth is $\Gamma_G(\widetilde W_T)$, and
$M(\widetilde W_T,G)$ is replaced by $\alpha_D\mathbf1_{\{\widetilde W_T\ge G\}}
+\alpha_RV^R(\Gamma_G(\widetilde W_T))$.
Averaging over $G$ gives $\mathcal T_{\mathrm{commit}}$. In the lower bound, use the same finite patches and the prescribed decision instead of $D^*$. The estimates
\eqref{eq:lb-early}--\eqref{eq:lb-late} are unchanged except for the same replacement of the terminal operator. Taking suprema gives \eqref{eq:tilde-Vforced}, which proves the proposition.
\end{proof}

\section{Sensitivity to Calibration Parameters}
\label{app:senscalibr}

\subsection{Additional Sensitivity Analysis}
\label{sec:baseline-sensitivity}

Note that for all the numerical experiments, as mentioned in the main text, we use time-$0$ dollars.

\subsubsection{High-Intensity Stress Test}
\label{app:ec-high-intensity-stress}

The two economic effects identified in
Section~\ref{sec:intro} --- the growth crowding-out effect and the deadline pressure effect --- arise in the ex ante investment problem while both goals are pending. We report the optional-funding value $V^{R,D}(0,w_0^{\mathrm{PV}})$ and its associated optimal risky position $\pi^*(0,w_0^{\mathrm{PV}})$. Figure~\ref{fig:calibration-nonmonotone} uses $(\alpha_D,\alpha_R)=(0.99,0.01)$, for which the optional and committed values coincide by
Proposition~\ref{prop:costless-commitment}. At the stress-test intensities $\lambda=0.70$ and $1.00$, value becomes locally nonmonotone near the precautionary affordability threshold. A barely affordable mandatory payment can exhaust resources needed for the more heavily weighted accumulation goal, leaving a slightly poorer household better positioned. This payment-induced decline requires $\alpha_D>\alpha_R$. The associated policy moves toward the negative control bound in this region. These intensities lie outside the calibrated range $\lambda\in[0.06,0.40]$, so the example is a stress test rather than a baseline prediction.

\begin{figure}[H]
\centering
\includegraphics[width=0.75\textwidth]{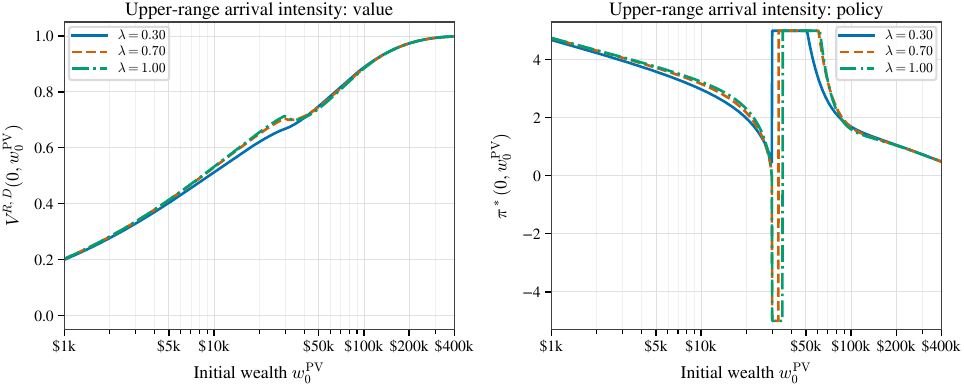}
\caption{High-intensity stress test. Panel (a) plots $V^{R,D}(0,w_0^{\mathrm{PV}})$ and panel (b) plots $\pi^*(0,w_0^{\mathrm{PV}})$ for $\lambda\in\{0.30,0.70,1.00\}$. Parameters are $(\alpha_D,\alpha_R)=(0.99,0.01)$, $T=18$, $G_0=\$124{,}000$, and a continuous near-degenerate
precautionary-goal distribution centered at $R_0=\$29{,}837.40$. Controls lie in $[-K,K]$ with $K=5$.}
\label{fig:calibration-nonmonotone}
\end{figure}

\subsubsection{Robustness to Goal Specification, Dispersion, and Saving Horizons}
\label{app:ec-committed-robustness}

We show here that the qualitative findings of the growth crowding-out effect and the deadline pressure effect from Section~\ref{sec:effects-numerics} are robust to variations in goal dispersion, the fixed-goal amount, the random-goal specification, and college saving horizons.

Figures~\ref{fig:calibration-retirement-sigma} and~\ref{fig:calibration-college-sigma} vary the dispersion $\sigma_G$ of the fixed-deadline goal distribution across retirement and college calibrations respectively. {Larger $\sigma_G$ fattens the right tail of $G$, raising the worst-case goal realization, and hence the safe level $b_G$. Greater dispersion lowers value slightly at intermediate wealth, where more mass on large realizations of $G$ makes full funding harder. Instead, additional mass on small realizations, which are achievable even from modest wealth, keeps the value curves nearly unchanged at low wealth. Because a higher $b_G$ raises the two-goal safe level $b_G+b_R$, the effect on policy is larger, with the optimal risky position $\pi^*$ elevated over a wider wealth range and reaching zero only at a higher wealth level.}

\begin{figure}[H]
\centering
\includegraphics[width=0.75\textwidth]{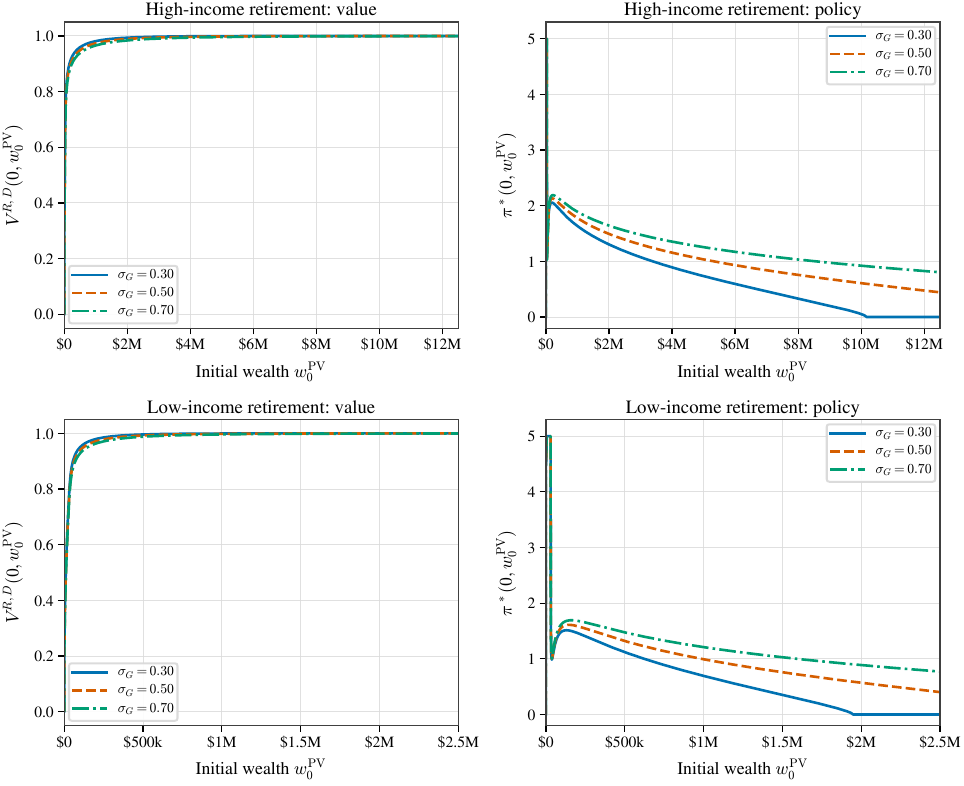}
\caption{Sensitivity of $V^{R,D}(0,w_0^{\mathrm{PV}})$ and
$\pi^*(0,w_0^{\mathrm{PV}})$ to $\sigma_G\in\{0.30,0.50,0.70\}$. The top row uses the high-income retirement target with median $G=\$4.125$ million and the bottom row uses the low-income target with median $G=\$780{,}000$; $T=40$ throughout. Value changes are modest, while greater dispersion keeps the optimal risky position elevated over a wider wealth range by expanding the upper tail of the fixed goal.}
\label{fig:calibration-retirement-sigma}
\end{figure}

\begin{figure}[H]
\centering
\includegraphics[width=0.75\textwidth]{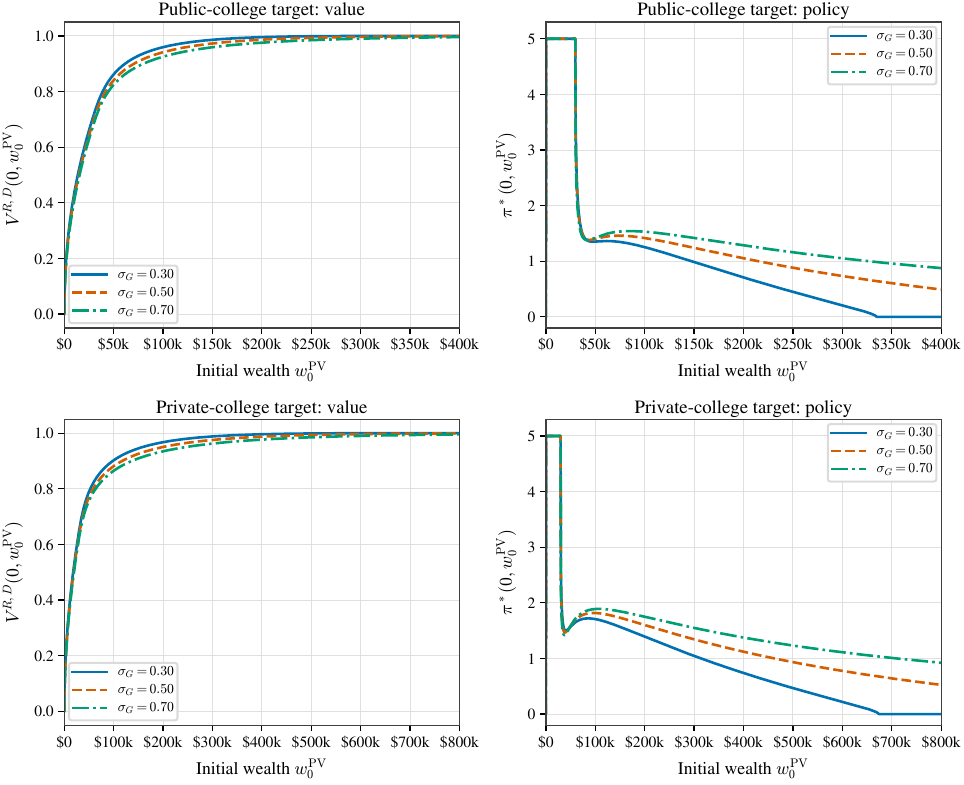}
\caption{
    Sensitivity of $V^{R,D}(0,w_0^{\mathrm{PV}})$ and $\pi^*(0,w_0^{\mathrm{PV}})$ to $\sigma_G\in\{0.30,0.50,0.70\}$ for public-college
    median $G=\$124{,}000$ (top) and private-college median
    $G=\$262{,}000$ (bottom). }
\label{fig:calibration-college-sigma}
\end{figure}

Figures~\ref{fig:calibration-college-public-horizon} and~\ref{fig:calibration-college-private-horizon} vary the saving horizon $T$. The deadline pressure effect documented in Section~\ref{sec:effects-numerics} of the main paper is present here: shorter horizons shift value toward higher wealth and produce larger risky positions farther above the random-goal affordability region. When the fixed-deadline goal is calibrated to the private college target \$262{,}000, the household requires substantially more wealth to achieve any given weighted goal-attainment value than under the public college target of \$124{,}000, but the deadline pressure effect operates similarly across both goal amounts.

\begin{figure}[H]
\centering
\includegraphics[width=0.75\textwidth]{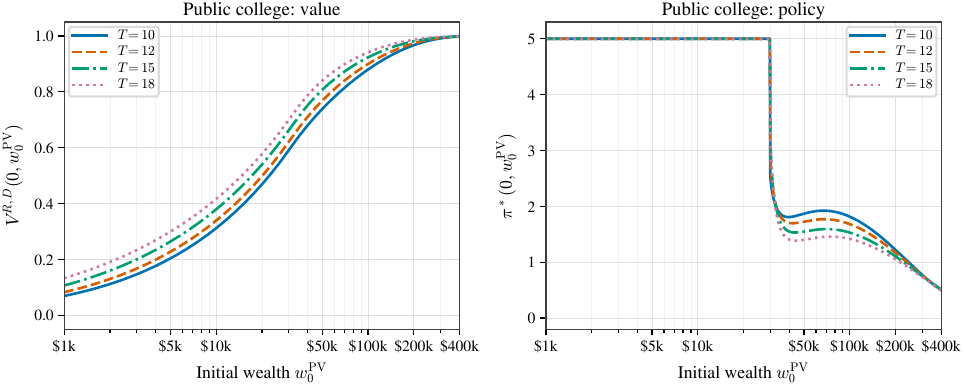}
\caption{
    Sensitivity of $V^{R,D}(0,w_0^{\mathrm{PV}})$ and $\pi^*(0,w_0^{\mathrm{PV}})$ to $T\in\{10,12,15,18\}$ under the public-college calibration, with $T=18$ representing saving from birth. Shorter horizons shift value toward higher wealth and produce larger risky positions farther above the random-goal affordability region.}
\label{fig:calibration-college-public-horizon}
\end{figure}

\begin{figure}[H]
\centering
\includegraphics[width=0.75\textwidth]{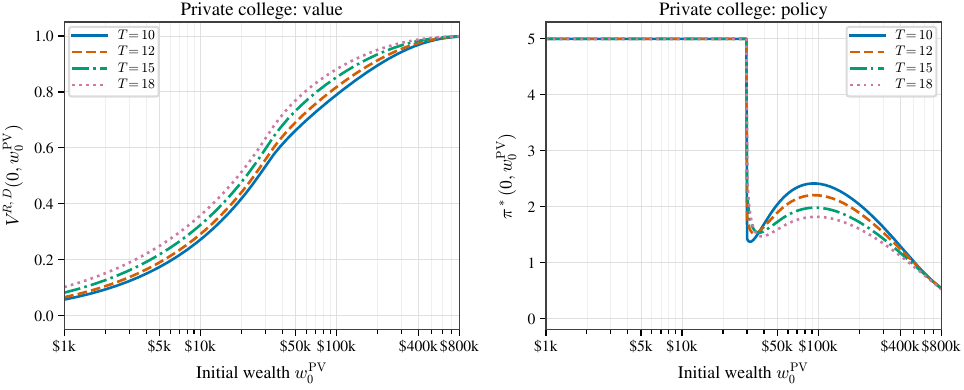}
\caption{
    Sensitivity of $V^{R,D}(0,w_0^{\mathrm{PV}})$ and $\pi^*(0,w_0^{\mathrm{PV}})$ to $T\in\{10,12,15,18\}$ under the private-college calibration. Shorter horizons shift value toward higher wealth and keep risk-taking at its cap over a wider range; the larger private-college goal raises the relevant wealth scale without changing the qualitative deadline-pressure effect.}
\label{fig:calibration-college-private-horizon}
\end{figure}

Figure~\ref{fig:random-goal-spec-public-college-highvar} replaces the near-degenerate precautionary-goal distribution with a truncated lognormal having the same mean $\$29{,}837.40$ and a standard deviation of approximately $\$15{,}437$. The value curves remain close, while dispersion in $R$ smooths the sharp policy adjustment near the random-goal affordability region. At higher wealth the policy curves are again close.

\begin{figure}[H]
\centering
\includegraphics[width=0.75\textwidth]{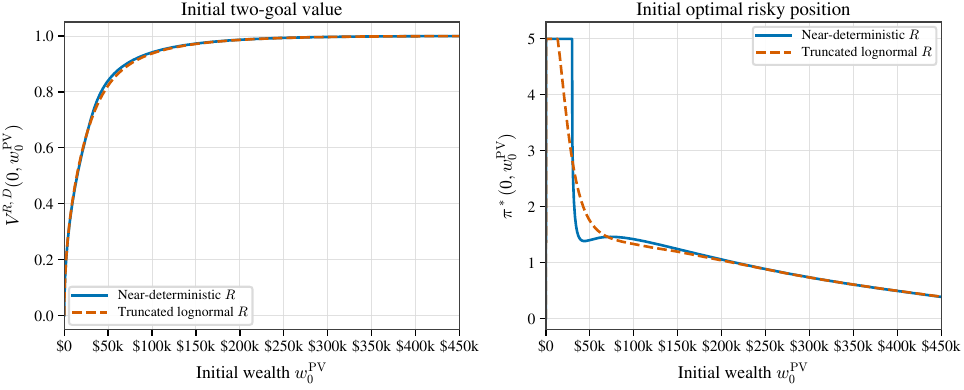}
\caption{
Robustness to the random-goal specification under the public-college calibration. The panels compare the continuous near-degenerate baseline centered at $R_0=\$29{,}837.40$ with truncated-lognormal $R$ having the same mean, $\sigma_{\log R}=0.50$, and a standard deviation of $\approx \$15{,}437$, holding $\lambda=0.20$, median $G_0=\$124{,}000$, $\sigma_G=0.50$, and $T=18$ fixed. The value curves remain close, while dispersion in $R$ smooths policy adjustment near the random-goal affordability region.}
\label{fig:random-goal-spec-public-college-highvar}
\end{figure}

\subsection{Additional Robustness of the Value-of-Flexibility Results}
\label{app:optfunding}

In this section, we complement the analysis in
Section~\ref{sec:optional-graphs} by showing the sensitivity of the ex ante and terminal option values to the precautionary-goal arrival intensity $\lambda$, the parameter that most directly governs the urgency of the random-deadline goal. {Unless otherwise indicated, all numerical results in this section use the setup of Section~\ref{sec:optional-graphs} in the main manuscript: $(\alpha_D,\alpha_R)=(0.10,0.90)$, $T=18$, $\lambda=0.20$, a truncated-lognormal accumulation goal with median $G_0=\$124{,}000$ and $\sigma_G=0.5$, and the near-degenerate precautionary-goal distribution centered at $R_0=\$29{,}837.40$. Under the main manuscript's baseline equal weights, flexibility has no value by Proposition~\ref{prop:costless-commitment}.}
We discuss sensitivity to accumulation-goal priority per dollar $\alpha_D/G_0$ in Section~\ref{sec:optional-graphs}. In Section~\ref{app:ec-optfunding-robustness}, we report an additional priority experiment that varies $\alpha_D$ with $G_0$ fixed, together with sensitivity of the ex ante option value to the deterministic deadline $T$.

\subsubsection{Arrival-Intensity Sensitivity}
\label{app:ec-lambda-sensitivity}

Figures~\ref{fig:option-value-lambda}
and~\ref{fig:terminal-option-value-lambda} show the effect of varying $\lambda$. Over the displayed range, which covers the calibrated intensities $\lambda\ge0.06$, the ex ante option value is larger when $\lambda$ is smaller. This monotonicity does not extend to rarer arrivals: Section~\ref{app:ec-random-arrival-asymptotics} shows that, over a wider range of intensities, the ex ante option value is largest at intermediate $\lambda$ and declines as $\lambda\downarrow0$. An important channel is the higher probability $e^{-\lambda T}$ that the random-deadline goal survives past $T$, which is the branch on which the funding decision is available, although $\lambda$ also affects $V^R$ and the optimal pre-deadline policy. The terminal option value isolates the terminal funding tradeoff and contains no survival factor, but it remains affected by $\lambda$ through the shape of $V^R$.

\begin{figure}[H]
    \centering
    \begin{subfigure}[t]{0.48\textwidth}
        \centering
        \includegraphics[width=0.85\linewidth]{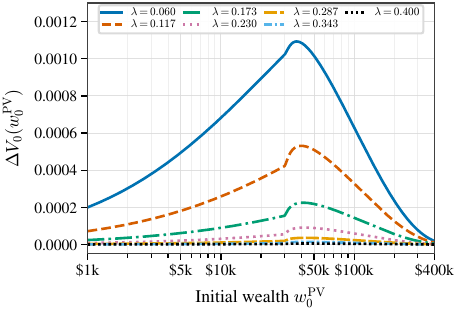}
        \caption{}
    \end{subfigure}
    \hfill
    \begin{subfigure}[t]{0.48\textwidth}
        \centering
        \includegraphics[width=0.85\linewidth]{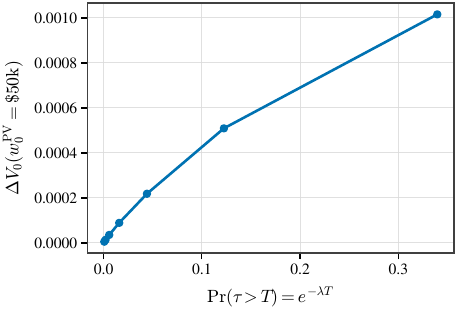}
        \caption{}
    \end{subfigure}
    \caption{Sensitivity of the ex ante option value
    $\Delta V_0(w_0^{\mathrm{PV}})$ to the random-deadline arrival intensity. Panel (a) plots $\Delta V_0(w_0^{\mathrm{PV}})$ against initial wealth for different values of $\lambda$. Panel (b) plots $\Delta V_0(w_0^{\mathrm{PV}}=\$50{,}000)$ against
    $\Pr(\tau>T)=e^{-\lambda T}$. Over the displayed range, lower $\lambda$ raises the probability that the funding decision is reached and increases the ex ante option value.}
    \label{fig:option-value-lambda}
\end{figure}

\begin{figure}[H]
    \centering
    \begin{subfigure}[t]{0.48\textwidth}
        \centering
        \includegraphics[width=0.85\linewidth]{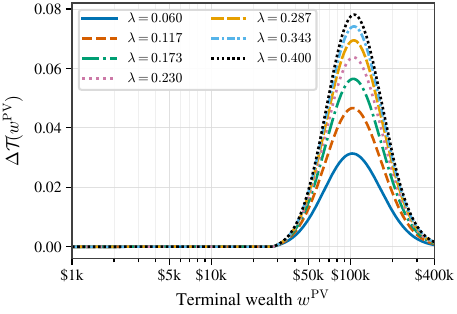}
        \caption{}
    \end{subfigure}
    \hfill
    \begin{subfigure}[t]{0.48\textwidth}
        \centering
        \includegraphics[width=0.85\linewidth]{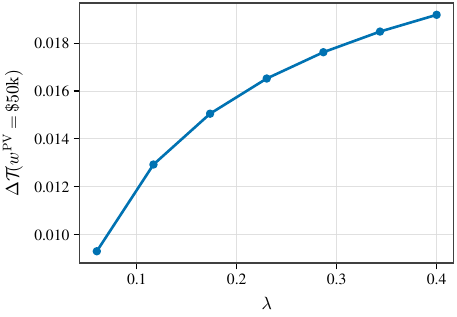}
        \caption{}
    \end{subfigure}
   \caption{Sensitivity of the terminal option value $\Delta\mathcal T(w^{\mathrm{PV}})$ to $\lambda$. Panel (a) plots $\Delta\mathcal T(w^{\mathrm{PV}})$ against terminal wealth for different $\lambda$, and panel (b) plots $\Delta\mathcal T(w^{\mathrm{PV}}=\$50{,}000)$ against $\lambda$. At $w^{\mathrm{PV}}=\$50{,}000$, the terminal option value rises over the displayed range because a more urgent unresolved precautionary goal makes preserving continuation wealth more valuable. }
    \label{fig:terminal-option-value-lambda}
\end{figure}

\subsubsection{Goal Priority, Saving Horizon, and Lower-Intensity Policy Comparisons}
\label{app:ec-optfunding-robustness}

The following experiments complement Section~\ref{sec:optional-graphs} of the main paper and the $\lambda$-sensitivity results in Section~\ref{app:ec-lambda-sensitivity} above. In Figure~\ref{fig:option-value-alpha-over-G}, $G_0$ is fixed and we vary $\alpha_D$. Over the displayed range, the time-0 option value initially rises, reaches an interior peak, and then falls. At very low fixed-goal priority, investment is directed mainly toward securing the random goal. Increasing $\alpha_D$ gives the household more incentive to accumulate toward the fixed goal, making the conflict between pursuing that goal and preserving the emergency buffer more relevant. Retaining discretion becomes more valuable in this region. At higher $\alpha_D$, however, the direct reward from funding increasingly outweighs the loss of random-goal continuation value, so the benefit of declining falls. The hump therefore reflects the interaction between portfolio choice before the deadline and the funding decision at the deadline. 

On the other hand, the terminal option value declines over the displayed range as increasing $\alpha_D$ raises the reward from funding the fixed goal while reducing the weight on the random-goal continuation value preserved by declining it. This differs from Figure~\ref{fig:asym-alpha-over-G-option} of the main paper, which holds $\alpha_D$ fixed and varies $G_0$ across both the inexpensive- and unaffordable-goal limits.

\begin{figure}[H]
    \centering
    \begin{subfigure}[t]{0.48\textwidth}
        \centering
        \includegraphics[width=0.85\linewidth]{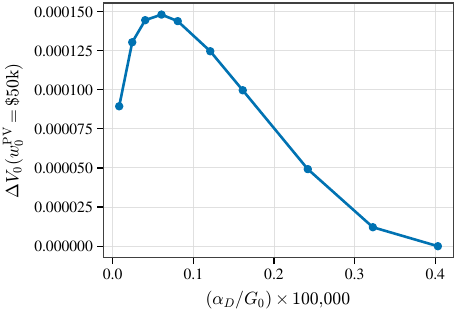}
        \caption{}
    \end{subfigure}
    \hfill
    \begin{subfigure}[t]{0.48\textwidth}
        \centering
        \includegraphics[width=0.85\linewidth]{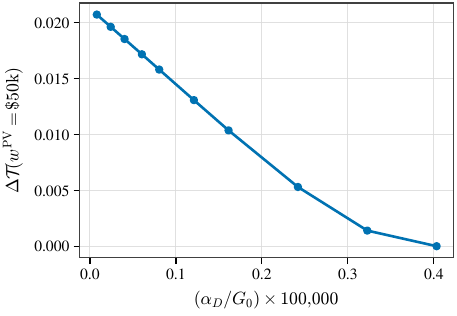}
        \caption{}
    \end{subfigure}
    \caption{Option values at initial wealth $w_0^{\mathrm{PV}}=\$50{,}000$ in panel (a) and terminal wealth $w^{\mathrm{PV}}=\$50{,}000$ in panel (b), as $\alpha_D$ varies with $G_0$ fixed, so the horizontal axis is $(\alpha_D/G_0)\times100{,}000$. Panel (a) plots $\Delta V_0(w_0^{\mathrm{PV}})$, which first rises from its low-priority limit and then declines. Panel (b) plots $\Delta\mathcal T(w^{\mathrm{PV}})$, which declines over the displayed range as funding the fixed goal increasingly dominates preserving random-goal continuation wealth.}
    \label{fig:option-value-alpha-over-G}
\end{figure}

Figure~\ref{fig:option-value-deadline-T-time0} varies $T$. A longer horizon gives more time to accumulate before the fixed goal is processed but also lowers $\Pr(\tau>T)=e^{-\lambda T}$. These forces oppose one another; over $T\in\{5,8,12,18,25\}$ the survival-probability dilution dominates, so the displayed time-0 option value decreases with the deadline. The terminal option value is independent of $T$ when the remaining primitives and wealth convention are fixed.

\begin{figure}[H]
    \centering
    \begin{subfigure}[t]{0.48\textwidth}
        \centering
        \includegraphics[width=0.85\linewidth]{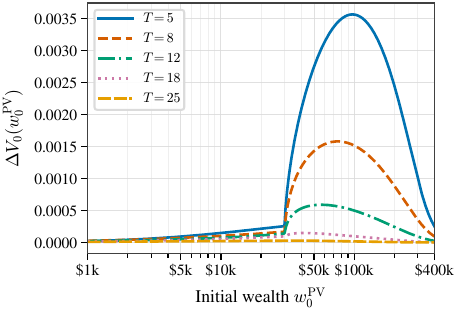}
        \caption{}
    \end{subfigure}
    \hfill
    \begin{subfigure}[t]{0.48\textwidth}
        \centering
        \includegraphics[width=0.85\linewidth]{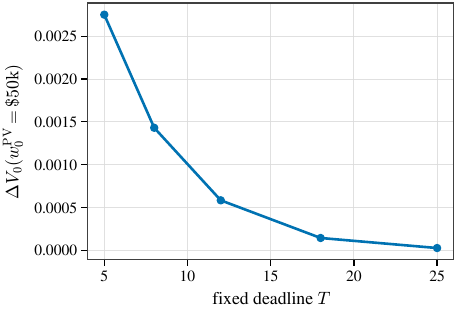}
        \caption{}
    \end{subfigure}
    \caption{Sensitivity of $\Delta V_0(w_0^{\mathrm{PV}})$ to the fixed deadline $T$. Panel (a) plots the option value against initial wealth for $T\in\{5,8,12,18,25\}$, and panel (b) plots $\Delta V_0(\$50{,}000)$ against $T$. A longer horizon permits more accumulation but reduces $\Pr(\tau>T)=e^{-\lambda T}$; over the displayed range the survival-probability dilution dominates and the time-0 option value decreases with $T$.}
    \label{fig:option-value-deadline-T-time0}
\end{figure}

Figure~\ref{fig:low-lambda-option} repeats the comparison at $\lambda=0.06$. The time-0 option value is larger in this experiment. A lower arrival intensity makes the funding decision more likely to be reached, while also changing the precautionary-goal continuation value and the portfolio chosen before the deadline. The expected policy difference remains positive throughout the displayed wealth range and has a somewhat larger peak than under the baseline arrival intensity.

\begin{figure}[H]
    \centering
    \begin{subfigure}[t]{0.48\textwidth}
        \centering
        \includegraphics[width=0.85\linewidth]{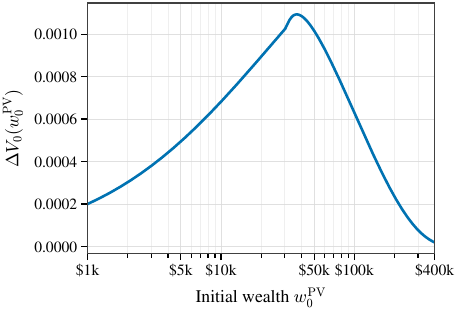}
        \caption{}
    \end{subfigure}
    \hfill
    \begin{subfigure}[t]{0.48\textwidth}
        \centering
        \includegraphics[width=0.85\linewidth]{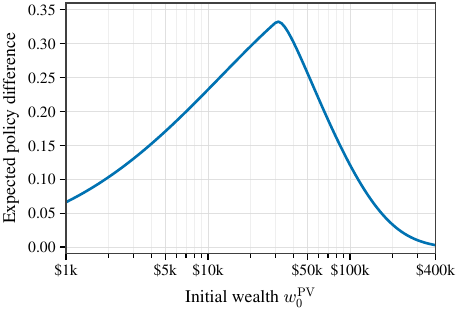}
        \caption{}
    \end{subfigure}
  \caption{Effects of retaining discretion at the lower arrival intensity $\lambda=0.06$. Panel (a) plots $\Delta V_0(w_0^{\mathrm{PV}})$. Panel (b) plots the expected policy difference $\Delta\bar\pi(T-\ell,w_0^{\mathrm{PV}})$, where $\ell=1/12$. The expected policy difference is positive throughout the displayed wealth range and reaches its largest value at intermediate initial wealth. Each expectation is computed deterministically under the corresponding funding rule's own controlled finite-state chain, conditional on $\tau>T-\ell$.}
    \label{fig:low-lambda-option}
\end{figure}

\section{Details for the Three-Period Illustration}
\label{app:three-period-details}

This appendix collects the calculations deferred from Section~\ref{sec:DiscreteBenchmark}.

\subsection{The One-Period Target Problem}

We first record portfolios attaining the value in \eqref{eq:one-step-value-symmetric}. Note that for every admissible $\pi\in[-\delta^{-1},\delta^{-1}]$, both discounted wealth multipliers $1\pm\delta\pi$ lie in $[0,2]$. If $x\ge A$, the choice $\pi=0$ guarantees the target amount in both states. If $A/2\le x<A$, either of the portfolios
\begin{equation}
\pi^+(x;A):=\frac{A/x-1}{\delta}
\qquad\text{and}\qquad
\pi^-(x;A):=-\frac{A/x-1}{\delta}
\label{eq:one-step-optimizers-symmetric}
\end{equation}
is admissible and reaches $A$ in exactly one of the up or down states, respectively. Since a guarantee in both states is not possible (the average terminal wealth for any admissible $\pi$ is less than $A$), the maximal success probability is $1/2$. If $x<A/2$, every $\pi\in\Pi$ is optimal because the goal value is not achievable. Taken together, this verifies the value \eqref{eq:one-step-value-symmetric}.

This single period analysis applies to the second stage of our problem, regardless of the order of the goals. Assuming the random goal arrives at date $1$, the same portfolios apply immediately after its arrival with target $G$ and current wealth $\Gamma_R(x)$. On the late arrival branch, the household first chooses whether to fund the deterministic goal and then applies the calculation with target $R$ and wealth $z-dG$.

\subsection{Continuation Thresholds and the Date-0 Recursion}

The continuation values in \eqref{eq:E-def-symmetric} and \eqref{eq:L-def-symmetric} are step functions. The possible jumps of the early-branch value $E$ are contained in
\begin{equation}
\mathcal C_E:=
\left\{
\frac{G}{2},\ G,\ R,\ R+\frac{G}{2},\ R+G
\right\}.
\label{eq:CE-symmetric}
\end{equation}
The corresponding points for the late-branch value $L$ are contained in
\begin{equation}
\mathcal C_L:=
\left\{
\frac{R}{2},\ R,\ G,\ G+\frac{R}{2},\ G+R
\right\}.
\label{eq:CW-symmetric}
\end{equation}
The first set combines the payment threshold $R$ with the one-period thresholds for post-payment wealth $\Gamma_R(x)$. The second combines affordability of the deterministic goal with the one-period thresholds before and after the payment of $G$. 

With ties resolved in favor of funding, an optimal late-branch decision is
\begin{equation}
d_G^*(z)
:=
\mathbf 1_{\{z\ge G\}}
\mathbf 1_{\{(1-q)+q\Phi_R(z-G)\ge q\Phi_R(z)\}}.
\label{eq:discrete-funding-selector}
\end{equation}
For fixed $q$, the comparison is constant between consecutive points of $\mathcal C_L$. Changing $q$ can reverse the decision throughout one of these intervals. The thresholds and optimal decision can be used to determine the initial portfolio weights for the two stage problem.

Let $w>0$ and fix the initial portfolio weight $\alpha\in\Pi$. The two possible date-$1$ wealth levels are $x^u(\alpha):=w(1+\delta\alpha)$ and $x^d(\alpha):=w(1-\delta\alpha)$. If the random goal has not arrived, the same weight is carried through the
second period. The date-$2$ wealth levels are then
\[
z^{uu}(\alpha):=w(1+\delta\alpha)^2,
\qquad
z^{ud}(\alpha)=z^{du}(\alpha):=w(1-\delta^2\alpha^2),
\qquad
z^{dd}(\alpha):=w(1-\delta\alpha)^2.
\]
Expanding the expectations in \eqref{eq:V0-general} gives
\begin{equation}
\begin{aligned}
V(w)
=\sup_{\alpha\in\Pi}\Bigg\{
&\frac{p}{2}\left[
E\bigl(x^u(\alpha)\bigr)+E\bigl(x^d(\alpha)\bigr)
\right] \\
&+\frac{1-p}{4}\left[
L\bigl(z^{uu}(\alpha)\bigr)
+2L\bigl(z^{ud}(\alpha)\bigr)
+L\bigl(z^{dd}(\alpha)\bigr)
\right]
\Bigg\}.
\end{aligned}
\label{eq:V0-general-expanded}
\end{equation}
The expression inside the supremum is even in $\alpha$ since replacing $\alpha$ by $-\alpha$ swaps $x^u$ with $x^d$ and $z^{uu}$ with $z^{dd}$, while leaving $z^{ud}$ unchanged. It therefore suffices to optimize over
$[0,\delta^{-1}]$.

The optimization can also be reduced to finitely many comparisons. A critical weight is any $\alpha\in[0,\delta^{-1}]$ at which one of $x^u(\alpha)$ or $x^d(\alpha)$ belongs to $\mathcal C_E$, or one of $z^{uu}(\alpha)$, $z^{ud}(\alpha)$, and $z^{dd}(\alpha)$ belongs to $\mathcal C_L$. There are finitely many such weights and in between any two such weights that are adjacent, every continuation value in \eqref{eq:V0-general-expanded} is constant. Hence, to perform the optimization it is enough to compare these critical weights, the endpoints of $[0,\delta^{-1}]$, and one point from each nonempty interval between consecutive critical weights.

To display the parameter dependence, write the optimized value as $V(w;p,q)$. For fixed $q$, the objective associated with each candidate initial portfolio is affine in $p$. For fixed $p$, its early-branch term is affine in $q$, while each late-branch term is the maximum of the funded and unfunded values, both of which are also affine in $q$. Taking the maximum over finitely many candidates shows that $V(w;p,q)$ is continuous, and separately piecewise affine and convex in $p$ and $q$. Changes in $p$ alter the relative importance of the outcomes on the early and late branches of the random goal. Changes in $q$ affect goal priority and can also reverse the date-$2$ funding decision.

\subsection{Equal Goal Amounts}

When $R=G$, the values of both continuation problems can jump only at
$R/2$, $R$, $3R/2$, and $2R$. Substitution into
\eqref{eq:E-def-symmetric} and \eqref{eq:L-def-symmetric} gives
\begin{equation}
E(x)
=
\begin{cases}
0, & 0\le x<R/2,\\
(1-q)/2, & R/2\le x<R,\\
q, & R\le x<3R/2,\\
(1+q)/2, & 3R/2\le x<2R,\\
1, & x\ge 2R,
\end{cases}
\qquad
L(z)
=
\begin{cases}
0, & 0\le z<R/2,\\
q/2, & R/2\le z<R,\\
\max\{q,1-q\}, & R\le z<3R/2,\\
\max\{q,1-q/2\}, & 3R/2\le z<2R,\\
1, & z\ge 2R.
\end{cases}
\label{eq:equal-goal-continuation-values}
\end{equation}

These thresholds have the following direct interpretation. Starting from wealth $R/2$, the remaining goal on either branch is reachable in one favorable state. At $R$, the mandatory payment is triggered on the early branch and the deterministic goal becomes affordable on the late branch. At $3R/2$, funding the first goal in either branch leaves $R/2$, so the second goal becomes feasible in one state. Finally, at $2R$, both goals can be achieved with probability one. The investor faces a tradeoff at time $0$ since a conservative portfolio protects wealth on adverse branches while a more aggressive portfolio can push favorable branches into the next threshold region. The optimal allocation reflects how these changes affect the funding probabilities of the two goals.

\section{Additional Asymptotic and Limit Results}
\label{app:ec-asymptotics}
Complementing the convergence to the Browne and Bayraktar--Young benchmarks in Section~\ref{sec:asymptotics} of the main paper, this section examines limits in the deterministic deadline $T$ and random-arrival intensity $\lambda$. All figures use $(\alpha_D,\alpha_R)=(0.10,0.90)$.

\subsection{Deterministic-Deadline Asymptotics}

We study the role of the deterministic deadline $T$ in the full two-goal problem. When $T\to0$, there is essentially no time to rebalance or accumulate wealth before the fixed-deadline goal is processed. Hence the time-0 two-goal value converges to the terminal operator $\mathcal T$ of \eqref{eq:TK}:
\begin{equation}
\label{eq:T-zero-terminal-limit-asym}
V^{R,D}(0,w_0^{\mathrm{PV}};T) \longrightarrow
\mathcal T(w_0^{\mathrm{PV}}),
\qquad T\downarrow0.
\end{equation}

\begin{figure}[H]
    \centering

    \begin{subfigure}[t]{0.48\textwidth}
        \centering
        \includegraphics[width=0.85\linewidth]{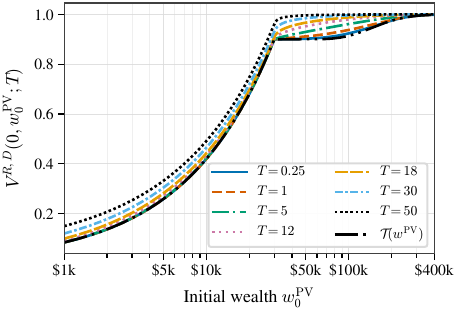}
        \caption{}
    \end{subfigure}
    \hfill
    \begin{subfigure}[t]{0.48\textwidth}
        \centering
        \includegraphics[width=0.85\linewidth]{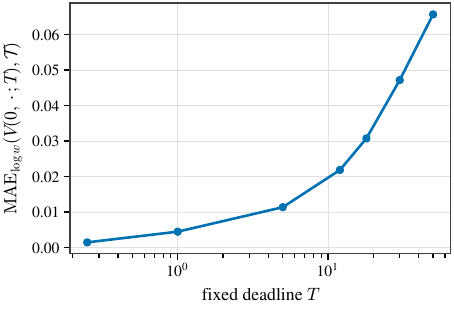}
        \caption{}
    \end{subfigure}

   \caption{Convergence as $T\downarrow0$. Panel (a) compares $V^{R,D}(0,w_0^{\mathrm{PV}};T)$ for different deadlines with the terminal operator $\mathcal T(w_0^{\mathrm{PV}})$. Panel (b) plots the mean absolute error over log wealth between these objects on the displayed interval. The declining error is consistent with convergence.}
    \label{fig:asym-deadline-terminal}
\end{figure}

% The ex ante option value also depends on $T$. The option to decline the accumulation goal matters only on the branch $\{\tau>T\}$, where the precautionary goal remains unresolved. A longer horizon lowers the probability $e^{-\lambda T}$ of reaching this branch, while also changing the time available to accumulate wealth. The conditional terminal tradeoff itself is independent of $T$ when the remaining primitives and wealth convention are held fixed.
% \begin{figure}[H]
%     \centering

%     \begin{subfigure}[t]{0.48\textwidth}
%         \centering
%         \includegraphics[width=0.9\linewidth]{optional_plots/asym_T_time0_option_value.pdf}
%         \caption{}
%     \end{subfigure}
%     \hfill
%     \begin{subfigure}[t]{0.48\textwidth}
%         \centering
%         \includegraphics[width=0.9\linewidth]{optional_plots/asym_T_terminal_option_value.pdf}
%         \caption{}
%     \end{subfigure}

%     \caption{Value of flexibility as $T$ varies. Panel (a) plots $\Delta V_0(\$50{,}000;T)$. Panel (b) plots $\max_{w^{\mathrm{PV}}}\Delta\mathcal T(w^{\mathrm{PV}})$. Over the displayed range, the time-0 option value falls as the survival probability $e^{-\lambda T}$ declines. The conditional terminal maximum is independent of $T$ when remaining primitives and wealth convention are fixed; small numerical differences arise from approximation.}
%     \label{fig:asym-deadline-option}
% \end{figure}

\subsection{Random-Arrival Asymptotics}
\label{app:ec-random-arrival-asymptotics}
 As $\lambda\downarrow0$, the random-deadline goal is unlikely to arrive before $T$, so the household increasingly concentrates on the fixed-deadline problem. For deterministic $G\equiv g$, the small-$\lambda$ reference is
\begin{equation}
\label{eq:lambda-zero-browne-shape-asym}
V^{R,D}(0,w;\lambda)
\approx \alpha_R+\alpha_D V^D(0,w;g),
\qquad w>0,\quad \lambda\approx0,
\end{equation}
where $V^D$ uses the same fixed portfolio cap $K$ as the two-goal problem. The additive $\alpha_R$ reflects the limiting precautionary-goal continuation value. Indeed, if the arrival is sufficiently far in the future, a positive remaining buffer has more time to grow toward the random goal. The wealth dependence is thus increasingly governed by the probability of funding the fixed-deadline goal.
Figure~\ref{fig:asym-lambda-small-and-option} (a) illustrates this behavior. The unconstrained Browne value is included as a separate benchmark. Since reducing $\lambda$ does not remove the portfolio cap, the Browne value is not the exact small-$\lambda$ limit.

Figure~\ref{fig:asym-lambda-small-and-option} (b) shows why the value of retaining discretion need not be monotone in $\lambda$. The option to decline the accumulation goal matters only on $\{\tau>T\}$, where the precautionary goal remains unresolved. Lowering $\lambda$ raises the probability $e^{-\lambda T}$ of reaching this branch. However, it also gives a positive post-payment buffer more time to grow toward the precautionary goal, reducing the continuation-value cost of funding the accumulation goal. In this experiment, these opposing effects produce the hump in Figure~\ref{fig:asym-lambda-small-and-option} (b), with discretion most valuable at intermediate intensities. 

\begin{figure}[htbp]
    \centering
    \begin{subfigure}[t]{0.48\textwidth}
        \centering
        \includegraphics[width=0.9\linewidth]{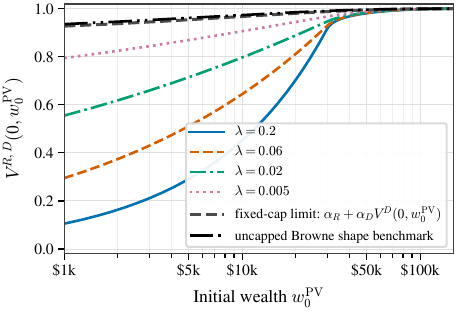}
        \caption{}
        \label{fig:limit-lambda-zero-browne}
    \end{subfigure}
    \hfill
    \begin{subfigure}[t]{0.48\textwidth}
        \centering
        \includegraphics[width=0.9\linewidth]{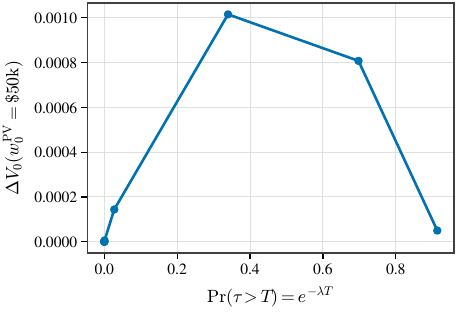}
        \caption{}
        \label{fig:asym-lambda-survival-option}
    \end{subfigure}
    \caption{Rare arrivals and the value of discretion.
    Panel (a) compares the two-goal value as $\lambda$ decreases with the fixed-cap reference
    $\alpha_R+\alpha_D V^D(0,w_0^{\mathrm{PV}};g)$ and the unconstrained Browne reference $\alpha_R+\alpha_D V^D_{\mathrm{Browne}}(0,w_0^{\mathrm{PV}};g)$, for deterministic $G\equiv g$. Panel (b) plots $\Delta V_0(\$50{,}000)$ against $\Pr(\tau>T)=e^{-\lambda T}$. Movement to the right corresponds to lower $\lambda$. The hump reflects the tradeoff.}
    \label{fig:asym-lambda-small-and-option}
\end{figure}

At the opposite extreme, as $\lambda\to\infty$, the precautionary goal arrives ``instantly'' before the household has enough time to trade. The two-goal value thus approaches $\mathcal J$ \eqref{eq:JK} at time $0$:
\begin{equation}
\label{eq:lambda-infty-reference-asym}
\mathcal J(0,w)
=\mathbb E_R\left[
    \alpha_R\mathbf 1_{\{w\ge R\}} + \alpha_D V^D\left(   0,w-R\mathbf 1_{\{w\ge R\}}
    \right)
\right].
\end{equation}
If the random goal is affordable, it is paid immediately; otherwise, it is missed and wealth is unchanged. The household then continues with the fixed-deadline problem. As $\lambda$ increases, there is less time to change wealth before the random payment, therefore value curves in Figure~\ref{fig:asym-lambda-limits} (a) approach this reference. The declining error in panel (b) is consistent with that convergence.

\begin{figure}[htbp]
    \centering
    \begin{subfigure}[t]{0.48\textwidth}
        \centering
        \includegraphics[width=0.9\linewidth]{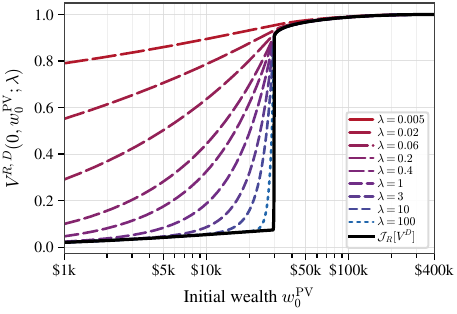}
        \caption{}
    \end{subfigure}
    \hfill
    \begin{subfigure}[t]{0.48\textwidth}
        \centering
        \includegraphics[width=0.9\linewidth]{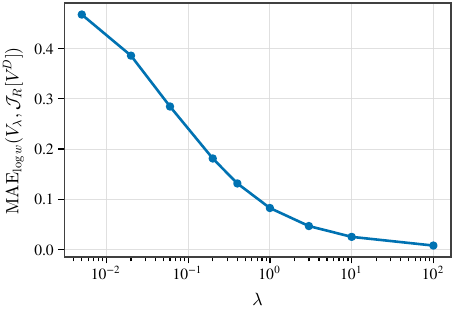}
        \caption{}
    \end{subfigure}
    \caption{Convergence to the instant-arrival limit as
    $\lambda\to\infty$. Panel (a) compares $V^{R,D}(0,w_0^{\mathrm{PV}};\lambda)$ with the instant-arrival reference $\mathcal J(0,w_0^{\mathrm{PV}})$. Panel (b) plots the mean absolute error over log wealth on the displayed interval.}
    \label{fig:asym-lambda-limits}
\end{figure}

\clearpage
\bibliographystyle{plainnat}
\bibliography{refs}

@techreport{jpmorgan2025replacement,
  author      = {{J.P. Morgan Asset Management}},
  title       = {Real-Life Data: An Innovative Approach to Calculating 
                 Income Replacement Rates},
  institution = {J.P. Morgan Asset Management},
  year        = {2025},
  url         = {https://am.jpmorgan.com/us/en/asset-management/adv/
                 insights/retirement-insights/
                 real-life-data-calculating-income-replacement-rates/}
}

@article{Bayraktar2025,
  title={Goal-based portfolio selection with mental accounting},
  author={Bayraktar, Erhan and Han, Bingyan},
  journal={arXiv preprint arXiv:2506.06654},
  year={2025}
}

@article{Browne1999,
  author  = {Sid Browne},
  title   = {Reaching goals by a deadline: digital options and continuous-time active portfolio management},
  journal = {Advances in Applied Probability},
  year    = {1999},
  volume  = {31},
  number  = {2},
  pages   = {551--577}
}

@book{Brunel2015,
  author    = {Jean L. P. Brunel},
  title     = {Goals-Based Wealth Management: An Integrated and Practical Approach to Changing the Structure of Wealth Advisory Practices},
  publisher = {Wiley},
  address   = {New York},
  year      = {2015}
}

@article{Capponi2024,
  author  = {Agostino Capponi and Yuchong Zhang},
  title   = {A Continuous Time Framework for Sequential Goal-Based Wealth Management},
  journal = {Management Science},
  year    = {2024},
  volume  = {70},
  number  = {11},
  pages   = {7664--7691}
}

@article{Chhabra2005,
  author  = {Ashvin B. Chhabra},
  title   = {Beyond {M}arkowitz: A Comprehensive Wealth Allocation Framework for Individual Investors},
  journal = {The Journal of Wealth Management},
  year    = {2005},
  volume  = {7},
  number  = {4},
  pages   = {8--34}
}

@article{Das2010,
  author  = {Sanjiv Das and Harry Markowitz and Jonathan Scheid and Meir Statman},
  title   = {Portfolio Optimization with Mental Accounts},
  journal = {Journal of Financial and Quantitative Analysis},
  year    = {2010},
  volume  = {45},
  number  = {2},
  pages   = {311--334}
}

@article{Das2022,
  author  = {Sanjiv R. Das and Daniel Ostrov and Anand Radhakrishnan and Deep Srivastav},
  title   = {Dynamic optimization for multi-goals wealth management},
  journal = {Journal of Banking and Finance},
  year    = {2022},
  volume  = {140},
  pages   = {106192}
}

@article{Markowitz1952,
  author  = {Harry Markowitz},
  title   = {Portfolio Selection},
  journal = {The Journal of Finance},
  year    = {1952},
  volume  = {7},
  number  = {1},
  pages   = {77--91}
}

@article{Shefrin2000,
  author  = {Hersh Shefrin and Meir Statman},
  title   = {Behavioral portfolio theory},
  journal = {Journal of Financial and Quantitative Analysis},
  year    = {2000},
  volume  = {35},
  number  = {2},
  pages   = {127--151}
}

@article{LettauLudvigson2001,
  author  = {Lettau, Martin and Ludvigson, Sydney},
  title   = {Consumption, Aggregate Wealth, and Expected Stock Returns},
  journal = {The Journal of Finance},
  year    = {2001},
  volume  = {56},
  number  = {3},
  pages   = {815--849},
  doi     = {10.1111/0022-1082.00347}
}

@techreport{collegeboard2025pricing,
  author      = {Ma, Jennifer and Pender, Matea and Hu, Xiaowen},
  title       = {Trends in College Pricing and Student Aid 2025},
  institution = {College Board},
  address     = {New York},
  year        = {2025},
  month       = {November},
  url         = {https://research.collegeboard.org/trends/college-pricing}
}

@article{lusardi2007baby,
  author  = {Lusardi, Annamaria and Mitchell, Olivia S.},
  title   = {Baby Boomer Retirement Security: The Roles of Planning,
             Financial Literacy, and Housing Wealth},
  journal = {Journal of Monetary Economics},
  year    = {2007},
  volume  = {54},
  number  = {1},
  pages   = {205--224}
}

@incollection{duly2003necessities,
  author    = {Duly, Abby},
  title     = {Consumer Spending for Necessities},
  booktitle = {Consumer Expenditure Survey Anthology, 2003},
  publisher = {Bureau of Labor Statistics},
  address   = {Washington, DC},
  year      = {2003},
  pages     = {35--42}
}

@book{Kallenberg2021,
  title={Foundations of modern probability},
  author={Kallenberg, Olav},
  publisher={Springer},
  edition = {Third},
  year={2021}
}

@book{AksamitJeanblanc2017,
  title={Enlargement of filtration with finance in view},
  author={Aksamit, Anna and Jeanblanc, Monique},
  year={2017},
  publisher={Springer}
}

@article{Young2016,
title = {Lifetime ruin under ambiguous hazard rate},
journal = {Insurance: Mathematics and Economics},
volume = {70},
pages = {125-134},
year = {2016},
issn = {0167-6687},
doi = {https://doi.org/10.1016/j.insmatheco.2016.06.007},
url = {https://www.sciencedirect.com/science/article/pii/S0167668715300810},
author = {Virginia R. Young and Yuchong Zhang},
}

@Article{Young2004,
  author={Virginia Young},
  title={{Optimal Investment Strategy to Minimize the Probability of Lifetime Ruin}},
  journal={North American Actuarial Journal},
  year=2004,
  volume={8},
  number={4},
  pages={106-126},
  month={},
  doi={10.1080/10920277.2004.10596174},
  url={https://ideas.repec.org/a/taf/uaajxx/v8y2004i4p106-126.html}
}

@article{BY16,
  title={Optimally investing to reach a bequest goal},
  author={Bayraktar, Erhan and Young, Virginia R},
  journal={Insurance: Mathematics and Economics},
  volume={70},
  pages={1--10},
  year={2016},
  publisher={Elsevier}
}

@article{crandall1992usersguideviscositysolutions,
  title={User’s guide to viscosity solutions of second order partial differential equations},
  author={Crandall, Michael G and Ishii, Hitoshi and Lions, Pierre-Louis},
  journal={Bulletin of the American mathematical society},
  volume={27},
  number={1},
  pages={1--67},
  year={1992}
}

@book{Pham2009,
  title={Continuous-time stochastic control and optimization with financial applications},
  author={Pham, Huy{\^e}n},
  volume={61},
  year={2009},
  publisher={Springer Science \& Business Media}
}

@book{Touzi2013,
  title={Optimal stochastic control, stochastic target problems, and backward {SDE}},
  author={Touzi, Nizar},
  volume={29},
  year={2013},
  publisher={Springer Science \& Business Media}
}

@book{Deguest2021,
  author    = {Romain Deguest and Lionel Martellini and Vincent Milhau},
  title     = {Goal-Based Investing: Theory and Practice},
  publisher = {World Scientific},
  address   = {Singapore},
  year      = {2021}
}

@article{GarganoRossi2024,
  author  = {Antonio Gargano and Alberto G. Rossi},
  title   = {Goal Setting and Saving in the {FinTech} Era},
  journal = {The Journal of Finance},
  year    = {2024},
  volume  = {79},
  number  = {3},
  pages   = {1931--1976},
  doi     = {10.1111/jofi.13339}
}

@incollection{GrealishKolm2023,
  author    = {Adam Grealish and Petter N. Kolm},
  title     = {Robo-Advisory: From Investing Principles and Algorithms
               to Future Developments},
  booktitle = {Machine Learning and Data Sciences for Financial Markets:
               A Guide to Contemporary Practices},
  editor    = {Capponi, Agostino and Lehalle, Charles-Albert},
  chapter   = {4},
  pages     = {60--85},
  publisher = {Cambridge University Press},
  address   = {Cambridge, UK},
  year      = {2023}
}

@techreport{Schwab2025,
  author      = {Rob Williams and Eric Tarkin and Inga Rachwald},
  title       = {Goals-Based Investing: A Time-Segmented Approach to
                 Portfolio Selection},
  institution = {Schwab Center for Financial Research},
  type        = {Wealth Management Insights},
  month       = oct,
  year        = {2025},
  url         = {https://bynder.schwab.com/m/45ffbed6d7cbbc1c/}
}
\end{document}